\documentclass[final]{siamltex}
\usepackage{amsmath,amssymb}
\usepackage{a4wide}
\usepackage{epsfig,subfigure}
\usepackage{dsfont}
\usepackage{algorithm,algorithmic}

\usepackage{color}

\usepackage[latin1]{inputenc}

\DeclareMathOperator*{\Argmax}{argmax}
\author{S. Allassonnière\thanks{CMAP Ecole Polytechnique, Route de
    Saclay, F-91128 Palaiseau, France}  \and  E. Kuhn\thanks{LAGA,
    Université Paris 13, 99,
    Av. Jean-Baptiste Clément, F-93430 Villetaneuse, France} \and
  A. Trouvé\thanks{CMLA, ENS
  Cachan, CNRS, PRES UniverSud, 61 Av. Président Wilson, F-94230
  Cachan, France}}
\title{Construction of Bayesian Deformable Models \\ via Stochastic Approximation
  Algorithm:\\ A
    Convergence Study
}
\newcommand{\R}{\mathbb{R}}
\newcommand{\Z}{\mathbb{Z}}
\newcommand{\N}{\mathbb{N}}

\newcommand{\s}{\mathcal{S}}

\newcommand{\te}{\theta}
\newcommand{\si}{\sigma}
\newcommand{\Te}{\Theta}
\newcommand{\Sa}{\mathcal{S}_a}

\newcommand{\tS}{S}

\newtheorem{notation}{Notation}
\newtheorem{Def}{Definition}%[section]
\newtheorem{coro}{Corollary}%[section]
\newtheorem{rem}{Remark}%[section]

\newcommand{\Kp}{\mathbf{K_p}}
\newcommand{\Kpb}{K_{p}^\beta}

\newcommand{\Kg}{\mathbf{K_g}}

\newcommand{\Syme}{\text{Sym}_{2k_g}^+}
\newcommand{\Symep}{\text{Sym}_{k_p}^+}

\newcommand{\lstat}{\boldsymbol{\pi}}
\newcommand{\ntrans}{\boldsymbol\Pi}

\newcommand{\bfe}{\boldsymbol{\varepsilon}}

\newcommand{\bdbeta}{\boldsymbol\beta}
\newcommand{\dy}{y}
\newcommand{\bdy}{\bold{\dy}}
\newcommand{\z}{z}
\newcommand{\bdz}{\bold{\z}}
\def\pixel{u}
\def\locpixel{v}

\def\Gammap{\Sigma_p}
\def\Gammag{\Sigma_g}
\def\I0{I_\alpha}

\def\pk{\locpixel_{p,j}}
\def\pkp{\locpixel_{p,j'}}
\def\gk{\locpixel_{g,j}}
\def\gkp{\locpixel_{g,j'}}
\def¶{\mathbb{P}}

\def\E{\mathcal{E}}

\def\Kapa{\mathcal{K}}

\newcommand{\qobs}{q_{obs}}%densite observee de y
\newcommand{\qmiss}{q_{m}}% densite a priori sur beta
\newcommand{\qcond}{q_c}% densite de y sachant beta
\newcommand{\qcomp}{q}% densite de (y,beta)
\newcommand{\qbayes}{q_B}%densite de theta sachant y
\newcommand{\qpara}{q_{para}}%densite de theta
\newcommand{\qj}{q_{j}}%densite de beta j sachant beta -j
\newcommand{\qpost}{q_{post}}%densite de beta | y
\newcommand{\bdbetamj}{\boldsymbol{\beta}^{-j}}% beta gras moins j
\newcommand{\betajb}{\boldsymbol{\beta}_{b\to j}} % beta j b
\begin{document}

\maketitle

\begin{abstract}
The problem of the definition and the estimation of generative models
based on deformable templates from raw data is of particular
importance for modeling non-aligned data affected by various types of
geometrical variability. This is especially true in shape modeling
in the computer vision community or in probabilistic atlas building
in Computational Anatomy. A first coherent statistical framework
modeling the geometrical variability as hidden variables
%has been given
was described by Allassonnière, Amit and Trouvé in \cite{AAT}. The present
paper gives a theoretical proof of convergence of effective stochastic
approximation expectation strategies to estimate such models and shows the
robustness of this approach against noise through numerical experiments in
the context of handwritten digit modeling.

\end{abstract}



\pagestyle{myheadings}
\thispagestyle{plain}
\markboth{S. ALLASSONNIERE, E. KUHN AND A. TROUVE}{BAYESIAN DEFORMABLE MODELS BUILDING}

%\tableofcontents

\section{Introduction}

In the field of image analysis, the
statistical analysis and modeling of variable objects from a limited
set of examples is still a quite challenging and a largely unsolved
problem and depends strongly on the use of adequate representations of
data.
%% YA UNCLEAR  %% Alain: corrigé ci dessus.
% Even more, the analysis of shape variability, even coded as
% functional data, thanks to the imaging process, cannot be efficiently
% done ``as is'' without using a more adequate representation.
One such representation is the so-called dense deformable template
(DDT) framework \cite{ag[89]}.
Observations are defined as deformations, taken from a family of
deformations of moderate ``dimensionality'', of a given exemplar or
template.
Such a representation appears particularly adapted to the emerging
field of Computational Anatomy where one aims at building
statistical models of the anatomical variability
within a given population \cite{grenander-miller}.
However, research  on DDT has been mainly focused on the variational point
of view, in which DDT is used as an efficient vehicle for a wide range
of registration algorithms \cite{chef02:_variat}. The problem
of template estimation, viewed as a statistical estimation problem of
parameters of generative models of images of deformable
objects, has received much less attention.

In this paper, we consider the hierarchical Bayesian framework for
dense deformable templates developed by Allassonnière, Amit and Trouvé in
 \cite{AAT} . Each image in a given population is assumed to be
 generated as a noisy and randomly deformed version of a
 common template drawn from a prior distribution on the set of
 templates. Individual deformations in their framework are treated as
 \emph{ hidden variables}  (or equivalently
 random effects in the mixed effects setting), whereas the
 template and the law of the deformations are parameters (or
 equivalently fixed effects) of interest. Parameter estimation for
 this model could be  performed by Maximum A Posteriori (MAP) for
 which existence and consistency (as the number of parameters observed
 images tends to infinity) has been proved (see \cite{AAT}). This
 contrasts with earlier work in \cite{glasbey01} using a penalized
likelihood (PL) or the more recent maximum description length approach in
\cite{marsland07} for which consistency cannot be proved  because the
deformations are considered as \emph{nuisance parameters} to be
estimated.

Our contribution in this paper is in defining effective and
theoretically proven convergent stochastic algorithms for computing
(local) maxima of the posterior on the parameters for Bayesian
deformable template models.
First, we specify an adapted stochastic
approximation expectation
 minimization algorithm (SAEM algorithm) in this
highly demanding framework where the hidden variables are non rigid
deformation fields living in finite but high dimensional space
(typically
hundreds or more dimensions). In particular, special attention is
needed
to the sampling of the posterior distribution on the
deformations. Obviously,
MCMC samplers are unavoidable, but non adaptive proposal
distributions yielding
simple symmetric random steps are of limited practical interest. The
present  paper introduces a more sophisticated hybrid Gibbs sampling
scheme allowing an acceptable rejection rate during the Estimation step. The
overall algorithm is cast in the larger class of SAEM-MCMC algorithms
introduced in \cite{kuhnlavielle}.
Second,  we extend the
convergence theory of SAEM-MCMC algorithms developed in
\cite{kuhnlavielle} to cover the case of \emph{unbounded} random
effects  arising naturally for deformation fields. The core material
of this extension is based on the general stability and convergence
results for
stochastic algorithms with truncation on random boundaries given in
\cite{andrieumoulinespriouret}. The main technical point is that in
the presence
of unbounded random effects and sequential estimation of the
covariance matrix
of the random effects, the usual regularity conditions of the solutions
of the Poisson equations for the Markovian dynamic as a function of the
parameters cannot be verified and have to be relaxed. As a result we provide
a new general stochastic approximation convergence theorem with a weaker
set of assumptions.
Third, we prove that the conditions for stability and convergence are
fulfilled for our general SAEM-MCMC estimation algorithm of Bayesian dense
deformable templates. Indeed, a well known weakness of general
stochastic approximation algorithm convergence results is that they
rarely provide
proofs of convergence for the algorithms \emph{used in practice} since
in these implementations the assumptions
are not satisfied or hard to verify (see
\cite{andrieumoulinespriouret}). Since stochastic approximation
algorithms have started recently to attract interest for deformable
model estimation (see
\cite{allassoniere06:_gener_model_and_consis_estim} and
\cite{richard08:_saem_algor_for_estim_of}) our results provide the
missing
theoretical foundations and guidelines for their effective use. As an
illustration of the potential of such SAEM-MCMC approaches in the
context of deformable templates, in particular in the presence of
noisy data, we present a set of experiments with images of handwritten
digits.

This article is organized as follows. Section \ref{obs} briefly
reviews the hierarchical Bayesian deformable template model proposed
by Allassonnière, Amit and Trouvé
in \cite{AAT}. In Section \ref{SAEM}, we develop the SAEM-MCMC
strategy for the estimation of the parameters. Then in Section 
\ref{sec:convergence}, we state our general convergence result for
truncated stochastic approximation algorithm
extending the Andrieu et al. Theorem of convergence in
\cite{andrieumoulinespriouret} and state that the designed family of
SAEM-MCMC algorithms in the
previous section satisfy the assumptions. The proof of this last
statement is postponed to Section \ref{appendix} after Section
\ref{experiments} concentrates on  experiments. In a final Section,
we provide a short discussion and conclusion.

\section{Observation model}\label{obs}

Let us recall the model introduced in \cite{AAT}.
We are given gray level images $(\dy_i)_{1\leq i\leq n}$ observed on
a grid of pixels $\{\locpixel_\pixel \in D \subset \R^2, \pixel \in \Lambda\}$ which is
embedded in a
continuous
domain $D \subset \R^2$, (typically $D=[-1,1]× [-1,1].$).
Although the images are observed only at the pixels
$(\locpixel_\pixel)_\pixel$, we are
looking for
a template image $I_0 : \R^2 \to \R$ defined on the plane (the
extension to images on $\mathbb{R}^d$ is straightforward).
Each observation $\dy$ is assumed to be the discretization on a fixed
pixel grid of a deformation of the template plus independent noise.
Specifically for each observation there
exists an \textit{unobserved} deformation field $z:\R^2 \to
\R^2$ such that for $\pixel \in \Lambda$
\begin{eqnarray*}
\dy(\pixel) = I_0(\locpixel_\pixel - z(\locpixel_\pixel)) + \epsilon (\pixel)\ ,
\end{eqnarray*}
where $ \epsilon$ denotes an independent additive noise.

\subsection{Models for template and deformation}
\label{modeltempdef}
Our model takes into account two complementary sides: photometry
-indexed by $p$, and geometry -indexed by $g$. Estimating the template
and the distribution on deformations
directly as a continuous function would be an infinite dimensional problem. We
reduce this problem to a finite dimensional one by restricting the
search to
a parameterized space of functions.
The template $I_0~: \ \R^2 \to \R$ and the deformation $z ~: \R^2 \to
\R^2$ are assumed to belong to
 fixed reproducing kernel
Hilbert spaces $V_p$ and $V_g$ defined by their respective kernels
$K_p$  and $K_g$.
Moreover, we restrict them to the subset of
linear combinations of the
kernels centered
at some fixed control points in the domain $D$: $(\pk)_{1\leq j\leq
  k_p}$
respectively
$(\gk)_{1\leq j\leq k_g}$.
 They are therefore parameterized by the
coefficients
$\alpha \in \R^{k_p}$ and $\beta \in
\R^{k_g} × \R^{k_g}$ as follows. For all $ \locpixel$ in $D $, let
\begin{eqnarray*}\label{template}
\I0(\locpixel) \triangleq  (\Kp \alpha)(\locpixel) \triangleq  \sum\limits_{j=1}^{k_p}
K_p (\locpixel,\pk) \alpha^j \,,
\end{eqnarray*}
and
\begin{eqnarray*}\label{bz}
z_\beta(\locpixel) \triangleq  (\Kg \beta)(\locpixel) \triangleq  \sum\limits_{j=1}^{k_g} K_g
(\locpixel,\gk) \beta^j\,.
\end{eqnarray*}
Other forms of smooth parametric representations of the images and of
the deformation fields could be used without changing the overall results.

\subsection{Parametric model} \label{model}
For clarity, we denote by $\bdy^t=(\dy_1^t,\ldots, \dy_n^t)$ and by
$\bdbeta^t=(\beta_1^t,\ldots,\beta_n^t)$ the collection of data and their
corresponding deformation coefficients.
The statistical model of the observations we consider is a generative
hierarchical one.
We assume conditional normal distributions for  $\bdy$ and $\bdbeta$:
\begin{equation}\label{pre-priors}
%  \label{eq:25}
\left\{
  \begin{array}[h]{l}
  \bdbeta\sim \otimes_{i=1}^n\mathcal{N}_{2k_g}(0,\Gamma_g)\  |\
\Gamma_g \, ,\\\\
\bdy\sim
  \otimes_{i=1}^n\mathcal{N}_{|\Lambda|} (z_{\beta_i}\I0,\sigma^2\text{Id})\ |\
  \bdbeta,\alpha, \sigma^2 \, ,
  \end{array}\right.
\end{equation}
where $\otimes$ denotes the product of distributions of independent variables and
$zI_\alpha(\pixel) = I_\alpha(\locpixel_\pixel - z(\locpixel_\pixel))$, for
$\pixel$ in $\Lambda$ denotes the action of the deformation on the template
image.
The parameters of interest are $\alpha$ which determines the template
image, $\sigma^2$ the variance
of the additive noise  and $\Gamma_g$ the covariance matrix
 of the variables $\beta$.
We assume that $\theta=(\alpha, \sigma^2, \Gamma_g)$ belongs to
an open parameter space $\Theta$:
$$  \Theta \triangleq \{\ \theta=(\alpha,\sigma^2,\Gamma_g)\ |\
\alpha\in\mathbb{R}^{k_p},\ \|\alpha\| <R \ |, \ \sigma>0,\ \Gamma_g \in\Syme\ \}\,,
$$
where $\|.\|$ is the Euclidean norm,  $\Syme$ is the cone of real positive $2k_g× 2k_g$ definite
symmetric matrices and $R$ an arbitrary positive constant.

The likelihood of the observed data $\qobs$ can be written as an integral over
the unobserved deformation variables. Let us denote by $\qcond$ the
conditional likelihood of the
observations given the hidden variables and by $\qmiss$ the likelihood
of these missing variables. Then,
\begin{eqnarray*} \label{dp/dm}
\qobs(\bdy|\theta) =\int \qcond(\bdy|\bdbeta, \alpha, \sigma^2)\qmiss(\bdbeta | \Gamma_g) d\bdbeta \ ,
\end{eqnarray*}
where all the  densities are determined by the model \eqref{pre-priors}.

\subsection{Bayesian model}\label{bayes}

Even though the parameters are finite dimensional,
the maximum-likelihood
estimator can yield degenerate estimates when the training sample is small.
By introducing  prior distributions
on the parameters, estimation with small samples is still
possible. The regularizing effect of such priors can be seen in the parameter
update steps (cf. \cite{AAT}).
We use a generative model based on standard conjugate prior
distributions for parameters $\te=(\alpha, \sigma^2,\Gamma_g)$ with
fixed hyper-parameters. Specifically, we assume
 a normal prior for $\alpha$, an inverse-Wishart prior on
$\sigma^2$ and an inverse-Wishart prior on
 $\Gamma_g$.
Furthermore, all priors are assumed to be independent. This yields
$\te=(\alpha, \sigma^2,\Gamma_g)\sim \qpara\triangleq \nu_p\otimes\nu_g$ where
\begin{equation}\label{priors} \left\{
\begin{array}{l}
\displaystyle{\nu_p(d\alpha, d\sigma^2) \varpropto
\exp \left(-\frac{1}{2} (\alpha-\mu_p)^t (\Gammap)^{-1} (\alpha - \mu_p) \right)
\left(
\exp\left(-\frac{\sigma_0^2}{2 \sigma^2}\right) \frac{1}{\sqrt{\sigma^2}}
\right)^{a_p}
  d\sigma^2
 d\alpha}, \ a_p \geq 3 \, , \\
\displaystyle{\nu_g(d\Gamma_g) \varpropto \left( \exp(-\langle
 \Gamma_g^{-1} ,
 \Gammag \rangle_F /2) \frac{1}{\sqrt{|\Gamma_g|}}  \right)^{a_g}
 d\Gamma_g},\ a_g\geq 4k_g+1
\, .
\end{array} \right.
\end{equation}
For two matrices $A$ and $B$, we define $\langle A,B \rangle_F \triangleq tr(A^tB)$ the
Frobenius dot product on the set of matrices where $tr$ denotes the
trace of the matrix.

\section{Parameter estimation based on stochastic approximation
  EM}\label{SAEM}
In our Bayesian framework, we obtain from \cite{AAT}
the existence of the MAP estimator
 \begin{eqnarray*} \label{MAP}
 \tilde{\theta}_n = \Argmax_{\theta\in \Theta} \qbayes(\theta|\bdy)\,,
 \end{eqnarray*}
where $\qbayes$ denotes the posterior likelihood of the parameters
given the observations. The dependence on $n$ refers to the sample
size.

We turn now to the maximization problem of the penalized posterior
distribution
$\qbayes(\theta|\bdy)$ % and we review briefly the framework of our model.
which has no closed form in our case. Indeed, the probability density
function is known up to a renormalization constant.
%% YA  Worse - it involves an integral you can't do?? %% Alain: laisser comme ça
That prevents
a direct computation of $\tilde \theta_n.$

In order to solve this problem, we apply an ``EM like''
algorithm to approximate the MAP estimator $\tilde\theta_n$. The
solution we propose is to base our algorithm on the use of the
Stochastic Approximation EM (SAEM). First,
we outline certain characteristics of our model, which highlight the
reasons for the choice of the particular procedure and enable us to
simplify its implementation.

\subsection{Model characteristics}

%\begin{notation}
% \label{not} To simplify the presentation, let us denote
% in the sequel
% %$\beta\triangleq \beta_1^n\in\mathbb{R}^N$, with
% $N \triangleq 2nk_g$.
%, the vector collecting all the missing variables and $\dy\triangleq \dy_1^n$ the collection of observations.
%\end{notation}

%Consider curved exponential densities, that is to say, situations
%where the
%complete likelihood can be written as:
An important characteristic
of our model is that it belongs
to the curved exponential family. In other words the complete
likelihood $\qcomp$ can be written as:
\begin{equation*}
\label{expform}
  \qcomp(\bdy,\bdbeta, \te ) = \exp \left[ -\psi(\te) + \langle \tS(\bdbeta),
  \phi(\te)\rangle \right] \, ,
\end{equation*}
where the sufficient statistic $\tS$ is a Borel function on $\R^N$,
with $N \triangleq 2nk_g$, taking its
values in an open subset $\mathcal{S}$ of $\R^m$ and $\psi$, $\phi$
two Borel functions on $\Te$. (Note that $\tS$, $\phi$ and $\psi$ may
depend also on $\bdy$, but since $\bdy$ will stay fixed in what follows, we
omit this dependence).

In our setting, we obtain the following formula:
\begin{eqnarray*}
  \log \qcomp(\bdy , \bdbeta , \te) & = &
 \log \qcond(\bdy |\bdbeta , \te)+ \log \qmiss( \bdbeta | \te)+ \log
 \qpara(\te) \, ,
\end{eqnarray*}
where $\qpara$ denotes the prior density of the parameters defined in
the previous paragraph.\\

For any $ 1\leq j \leq k_p$ and any $ \pixel \in \Lambda$, we
denote by
\begin{eqnarray*}
\Kpb (\pixel,j) = K_p (\locpixel_\pixel - z_\beta(\locpixel_\pixel), \pk) \
\end{eqnarray*}
 the matrix which corresponds to the deformation of the kernel $K_p$
 through $z_\beta$ at pixel $\pixel$ and evaluated at pixel location
 $\\locpixel_\pixel $.
 Then, for some constant $C$ independent of $\theta$,
\begin{eqnarray*}
 \log \qcomp(\bdy ,\bdbeta , \te)&=& \sum\limits_{i=1}^n \left\{
-\frac{|\Lambda|}{2} \log (\sigma^2)
-\frac{1}{2\si^2} \|\dy_i -
K_p^{\beta_i}\alpha\|^2 \right\} \\
&+&\sum\limits_{i=1}^n \left\{
%-\frac{2k_g}{2} \log(2\pi)
-\frac{1}{2}  \log (|\Gamma_g|) -\frac{1}{2} \beta_i^t \Gamma_g^{-1} \beta_i
 \right\} \\
&+& a_g \left\{
 - \frac{1}{2} \log(|\Gamma_g|)- \frac{1}{2}\langle \Gamma_g^{-1} , \Gammag \rangle_F \right\}
- \frac{1}{2} (\alpha - \mu_p)^t \Gammap^{-1} (\alpha - \mu_p)\\
&  +&a_p\left\{ -
\frac{1}{2}\log(\sigma^2) -\frac{\sigma_0^2}{2\sigma^2} \right\}
 + C\, .
\end{eqnarray*}
Note that $\|\dy_i-K_p^{\beta_i}\alpha\|^2 = (\dy_i-K_p^{\beta_i}\alpha)^t
(\dy_i-K_p^{\beta_i}\alpha)$, where $K_p^{\beta_i}\alpha$ is another way to write the
action of the deformation $ z_{\beta_i}$ on the template $I_\alpha$
denoted previously by
$ z_{\beta_i}I_\alpha$. This form emphasizes the dot product between  the
sufficient statistics and a function of the parameters.
It can be easily verified that the following
matrix-valued functions are the sufficient statistics (up to a
multiplicative constant)~:
%and using the
%fact that $\langle
%K_p^{\beta_i}\alpha,K_p^{\beta_i}\alpha\rangle=\langle
%(K_p^{\beta_i})^tK_p^{\beta_i},\alpha^t\alpha\rangle_F$,
% we get easily  the following matricial form of the sufficient
%statistics:
\begin{eqnarray*}
  \label{statexhmat}
S_1(\bdbeta) & =  &  \sum\limits_{1\leq i\leq n} \left( K_p^{\beta_i}
\right)^t \dy_i \,,\\
S_2(\bdbeta) & = &\sum\limits_{1\leq i\leq n} \left( K_p^{\beta_i}
\right)^t \left( K_p^{\beta_i}
\right) \,,\\
S_3(\bdbeta) & = & \sum\limits_{1\leq i\leq n} \beta_i ^t \beta_i   \, .
\end{eqnarray*}
For simplicity, we denote
$S(\bdbeta)=(S_1(\bdbeta),S_2(\bdbeta),S_3(\bdbeta))$ for any
$\bdbeta\in\mathbb{R}^N$ and
 define the sufficient statistic space as
$$\mathcal{S}=\left\{
(S_1,S_2,S_3) \ | \ S_1 \in \R^{k_p} , \ S_2 +\sigma_0^2 \Sigma_p^{-1}\in  \Symep, \
\ S_3 + a_g \Sigma_g  \in \Syme
\right\}\,.$$
Identifying  $S_2$ and $S_3$ with their lower triangular
parts, the set $\mathcal{S}$ can be viewed as an open set of
$\R^{n_s}$ with $n_s=k_p +
\frac{k_p(k_p+1)}{2} + k_g(2k_g+1)$.

% As already proved in \cite{AAT} the maximising function $\hat{\te}$
% satisfying (\ref{eq:thetahat})
% exists. We can thus give an explicit form of $\hat\te(s)=(\alpha(s),\Gamma_g(s))$ for our
% sufficient statistic vectors and matrices $(s_1,s_2,s_3)$: % which
% can be
 In \cite{AAT}, the existence of the parameter estimate $\hat{\te}(S)$ that
 maximizes the complete log-likelihood  has
 been proved. It can easily be shown that  $\alpha$, $\si^2$ and $ \Gamma_g$ are
 explicitly expressed with the above sufficient statistics
as follows:

\begin{equation}  \label{PhotoUpdateClust}
\left\{
\begin{array}{lll}
\Gamma_{g} (S) & = & \frac{1}{n+a_g}
(S_3 + a_g \Gammag) \, ,\\
\\
\alpha (S) &=& \left(
S_{2}
+ \sigma^2(S)  (\Gammap)^{-1}
\right)^{-1}  \left(
S_1
+  \sigma^2(S)
(\Gammap)^{-1} \mu_p \right) \, ,\\
\\
\sigma^2     (S) &= &  \frac{1}{n |\Lambda| +a_p} \left(
n \|\bdy \| ^2 +\alpha(S)^t S_2 \alpha(S)-2 \alpha(S)^t S_1 + a_p\sigma_0^2 \right) \, .
\end{array} \right.
\end{equation}

All these formulas also prove the smoothness of $\hat{\te}$ on the
subset $ \mathcal{S}$.

% \begin{rem}
%   We consider that the variance of the
%   independent Gaussian noise is fixed as later in
%   the proof of the convergence theorem. However, in the applications,
%   we estimate this variance together with the other parameters. This
%   additional estimation does not change the characteristics of the
%   model and in
%   particular neither the sufficient statistics nor the algorithm.
% \end{rem}

\subsection{SAEM-MCMC algorithm with truncation on random
  boundaries}

In order to compute the MAP estimator for our Bayesian model, we use
a variant of the EM (Expectation-Maximization, \cite{DLR}) algorithm. This
algorithm is quite natural when we have to maximize a likelihood under
a hierarchical model with missing
variables. Unfortunately, direct computation is not tractable and
we have to find a solution to overcome the problematic E step
 where we have to compute an expectation with respect to the
posterior distribution on $\bdbeta$ given $\bdy$. A first
attempt was proposed in \cite{AAT} where this conditional distribution
is approximated by a Dirac distribution at its mode (Fast Approximation
with Mode -FAM-EM). The results are very interesting, however, the
authors point out the lack of convergence of the FAM-EM algorithm when the
quality of the input images is not good, typically when they are noisy.
This is the issue we consider here. We propose an
algorithm that ensures the convergence of the resulting
sequence of estimators toward the MAP whatever the quality of the input.

This solution is a procedure combining the Stochastic Approximation
EM (SAEM) with Markov Chain Monte Carlo
(MCMC) in a more general framework than
that proposed by \cite{kuhnlavielle}, which in turn
generalized the algorithm introduced by \cite{DLM}.
Indeed, the $k^{th}$ iteration of the SAEM-MCMC algorithm consists of
three steps:
\begin{description}
\item[Step 1~: Simulation step.]
  The missing data, i.e. the deformation parameters $\bdbeta$,
  are drawn using the transition probability of a convergent Markov
  chain $\ntrans_\te$ having the posterior distribution
  $q_{post}(.|\bdy,\te)$ as its stationary distribution:
$$\bdbeta_k
  \sim\ntrans_{\theta_{k-1}}(\bdbeta_{k-1},\cdot)\,.$$
\item[Step 2~: Stochastic approximation step.]
A stochastic approximation is done on the complete log-likelihood
using the simulated value of the missing data:
\begin{equation*}
\label{appsto}
Q_{k}(\theta)=Q_{k-1}(\theta)+\Delta_{k-1}[\log
\qcomp(\bdy,\bdbeta_k,\theta)-Q_{k-1}(\theta)] \,,
\end{equation*}
where $\boldsymbol\Delta=(\Delta_k)_k$ is a decreasing sequence of positive step-sizes.
\item[Step 3~: Maximization step.]
The parameters are updated in the M-step:
$$\theta_{k}=\Argmax\limits_{\te\in\Theta}
  Q_{k}(\theta)\,.$$
\end{description}
The initial values $Q_0$ and $\te_0$ are arbitrarily chosen.
\begin{rem}
  We cannot use the direct SAEM algorithm. Indeed, this would require
   sampling the hidden variable from the posterior
  distribution which is known only up to a normalization constant. This
  sampling is not possible here due to the complexity of the posterior
  probability density function.
\end{rem}
%  Since our model is exponential, the stochastic
% approximation can be done on the complete log-likelihood as well as on a
% sufficient statistic.
% This
% is due to the fact that the missing data
% only appears linearly through a sufficient statistic $\tS$ in
% the exponential exponent.
% This yields the following stochastic
% approximation:
% \begin{equation}
% s_k  = s_{k-1} + \Delta_{k-1} ( \tS(\beta_k) -  s_{k-1}) 
% \end{equation}
% which is none other than equation \eqref{eq:b14} in the previous
% section.

Since our model belongs to the curved exponential
family, the stochastic
approximation step can easily be done on the sufficient statistics $S$
instead of on the complete log-likelihood. Then the maximization
step 3 is straightforward, replacing
in \eqref{PhotoUpdateClust}  the sufficient statistics with their
corresponding stochastic
approximations. \\

The convergence of this algorithm has been proved in
\cite{kuhnlavielle} in the particular case of  missing
variables living in a compact subset of $\R^N$.
However, as we set a Gaussian prior on the
missing variables $\bdbeta$, we cannot assume that their support is
compact.  In order to provide an algorithm whose convergence can be
proved in the current framework we
have to use a more general
setting introduced in \cite{andrieumoulinespriouret}
which involves truncation on random boundaries. The proof is given in
Section \ref{sec:convergence}.
%Thanks to this approach, we end up with an algorithm using an MCMC coupling
%procedure into the SAEM algorithm and a truncation on random
%boundaries.
This can be formalized as follows.

Let
$(\Kapa_q)_{q\geq 0} $ be a sequence of increasing compact subsets of
$\mathcal{S}$ such as $\cup_{q\geq 0} \Kapa _q = \mathcal{S} $
and $ \Kapa _q \subset  \text{int}(\Kapa _{q+1}) ,$ for all $q \geq 0$. Let
$\bfe=(\varepsilon_k)_{k\geq 0}  $ be a monotone non-increasing sequence of positive
numbers and $\mathrm{K}$ a compact subset of $\R^{N}$. We construct a
sequence $((\bdbeta_k,s_k ))_{k\geq 0}$ as described in Algorithm
 \ref{AlgoMoulines} as follows.
As long as the stochastic approximation does not
 wander outside the current
compact set and is not too far from its previous value, we run the
SAEM-MCMC algorithm. As soon as one of these conditions is
not satisfied, we reinitialize the sequences of $\bdbeta$ and $s$ using
a projection (for more details see \cite{andrieumoulinespriouret} ), we
increase the size of the compact set
and continue the iterations until convergence.
This is detailed in the following steps~:

\begin{description}
\item[Initialization step~:]  Initialize  $\bdbeta_0$  and $s_0$ in two
  fixed compact sets $\mathrm{K}$ and  $\Kapa_0$  respectively.\\
\item[]Then, for the $k^{th}$ iteration, repeat the  following four steps~:\\
\item[Step 1~: MCMC simulation  step]. Draw  one new element
  $\bar\bdbeta$ of the
  non-homogeneous Markov Chain with respect to
 the kernel with the current parameters $\ntrans_{\theta_{k-1}}$ and
 starting at  $\bdbeta_{k-1}$.\\
 $$\bar\bdbeta
  \sim\ntrans_{\theta_{k-1}}(\bdbeta_{k-1},\cdot)\,.$$
 \item[Step 2~: Stochastic approximation step].
Compute
\begin{equation}\label{eq:sbar}
\bar{s} =
s_{k-1} + \Delta_{\zeta_{k-1}} (
\tS(\bar{\bdbeta}) -  s_{k-1})\,.
\end{equation}

\item[Step 3~: Truncation on random boundaries].
 If $\bar{s}$ is outside the current compact set $\Kapa_{\kappa_{k-1}} $ or too far
 from the previous value $s_k$, then
 restart the stochastic approximation in the initial compact
 set, extend the truncation boundary to $\Kapa_{\kappa_{k}} $ and start again
 with a bounded value of the missing
 variable. Otherwise,
 set $ ( \bdbeta_k,s_k ) = (\bar{\bdbeta},\bar{s})$ and keep the
 truncation boundary to $\Kapa_{\kappa_{k-1}}$.\\
\item[Step 4~: Maximization step].
Update the parameters using \eqref{PhotoUpdateClust}. \\
\end{description}
 % Then,
% we draw one element of the non-homogeneous Markov Chain with respect to
% the kernel with the current parameters. If this yields a too large
% stochastic approximation or too far from the previous one, then we
% restart the stochastic approximation in the initial compact
% set. Otherwise, we update the parameters.
%More details are given in Algorithm \ref{AlgoMoulines}.

In this algorithm, the MCMC simulation step has
to be explained since it involves the choice of the transition kernel
of the Markov chain.
Usually, one uses a Metropolis-Hastings algorithm in which
 a candidate value is sampled from a proposal
distribution followed by an
accept-reject step. However, there are
different possible proposal distributions.
The only requirement is that all these kernels
lead to an ergodic
Markov chain whose stationary distribution is our posterior
distribution. The choice among these possibilities should be based on
the specific
 framework we are
working in.

%The first question to be answered is whether or not the
%  proposal should depend on the current parameters.
While minimizing the Kullback-Leibler distance between the stationary
  distribution $\bdbeta\to \lstat_\theta(\bdbeta)$ and a tensorial product $
  \bdbeta\to \otimes_{i=1}^n p(\beta_i)$ corresponding to
  independent identically distributed missing variables,
 we get that $p$ is proportional to
 $\frac{1}{n} \sum\limits_{i=1}^{n}\qpost(.| y_i,\theta) $.
As $n$ tends to $\infty$ and for a given $\theta$, $p$ converges a.s. towards the
prior pdf on the missing variable $\qmiss(.|\theta)$.
This suggest to use as
proposal the prior distribution which involves the current parameters.

On the other hand, the setting we have in this paper deals with high
dimensional missing variables. This raises several issues. If we
%use a Metropolis-Hastings algorithm,
simulate candidates for the hidden variable as a complete vector,
it appears that most of the
candidates are rejected.
This is a typical high dimensional concentration phenomenon~:
locally around a current point, the proportion of the space
occupied by acceptable moves becomes negligible when the space dimension
grows.
%% YA What is image point of view? %% Alain: C'est vrai que ça sonne bancal
%% Je propose: From a more practical point of view
%From an image point of view,
From a more practical point of view, even
if the proposed candidate is drawn with respect to the current prior
distribution, it creates a deformation that is very different from the
current one and too large for the corresponding deformed template to fit the
observations.
%This is why it is rejected.
This yields very few possible moves from the current
missing variable value and the algorithm is stuck in a non-optimal
location or converges very slowly.
% As a consequence, the infinite
% support of the unobserved variables is not well covered by the
% simulated variables.

One solution is to update the chain one coordinate at a time
conditionally on the others. This
corresponds to a Gibbs sampler and leads to more relevant candidates which have a higher
chance to be accepted (cf. \cite{AmitGibbs}).
%% YA tu me paye pour rediger - merci beaucoup.
%% YA The following is unclear. %%Alain: Oui il faut revoir cela lundi
From an image analysis point of view, this put stronger conditions on
the kind of
deformations which are produced when proposing a candidate for each
coordinate.
Knowing the tendency of the
movement given by the other coordinates,
the candidate will either confirm it or not depending if this is
a suitable movement. It will thus be accepted with a corresponding
probability.
 Even if some coordinates remain unchanged, some others are
updated which enables the algorithm to visit a larger part of the missing variable
support.

% The geometrical covariance matrix is poorly
% estimated as well as the corresponding template.
% Therefore, the geometrical covariance matrix and
% the template are well estimated.

\begin{algorithm}
 \caption{Stochastic approximation with truncation on random boundaries}
\label{AlgoMoulines}
\begin{algorithmic}
\STATE Set $\bdbeta_0 \in \mathrm{K}$, $s_0 \in \Kapa_0$, $\kappa_0=0$,  $\zeta_0=0$,
 and
 $\nu_0=0$.
\FORALL{$k\geq 1$}
\STATE\texttt{compute} $\bar{s} =
s_{k-1} + \Delta_{\zeta_{k-1}} (
\tS(\bar{\bdbeta}) -  s_{k-1}) $
\STATE \texttt{where} $\bar{\bdbeta}$ \texttt{is sampled
from a transition kernel} $\ntrans_{\theta_{k-1}}(\bdbeta_{k-1},.)$.
\IF{ $\bar{s} \in \Kapa_{\kappa_{k-1}} $
\texttt{and} $\|\bar{s} - s_{k-1}\|\leq \varepsilon_{\zeta_{k-1}}  $}
\STATE \texttt{set} $( \bdbeta_k,s_k) = (\bar{\bdbeta},\bar{s})$
\texttt{and}
$\kappa_k=\kappa_{k-1} $, $\nu_k=\nu_{k-1}+1$, $\zeta_k=\zeta_{k-1}+1$
\ELSE
\STATE \texttt{set}  $( \bdbeta_k,s_k) = (\tilde{\bdbeta},\tilde{s}) \in
 \mathrm{K}× \Kapa_0 $ \texttt{and}  $\kappa_k=\kappa_{k-1}+1 $, $\nu_k=0 $, $\zeta_k =
\zeta_{k-1} + \phi(\nu_{k-1})$
\STATE \texttt{where} $\phi : \ \N \to \Z$ \texttt{is a
  function such that}
$\phi(k)> -k$ \texttt{for any} $k$
\STATE \texttt{and}
$(\tilde{\bdbeta},\tilde{s})$  \texttt{can be chosen through
  different ways (cf. \cite{andrieumoulinespriouret})}.

\ENDIF

\STATE $\te_k$% =\Argmax\limits_\te L(s_k,\te)
%$ solution of \eqref{PhotoUpdateClust}\\
$=\hat{\theta}(s_k)$.

\ENDFOR\\

\vspace{0.5cm}
\end{algorithmic}
\end{algorithm}

\begin{rem}
  The index $\kappa$ denotes the current active truncation set, the index
  $\zeta$ is the current index in the sequences $\boldsymbol\Delta$ and
  $\boldsymbol\varepsilon$ and  the index
  $\nu$ denotes the number of iterations since the last projection.
\end{rem}

\subsection{Transition probability of the Markov chain}
\label{HGS}
We now explain  how to simulate the missing variables thanks to a
Markov Chain Monte Carlo algorithm having the posterior distribution as
its stationary distribution. Due to the inherent
 high dimensionality $N$ of $\bdbeta$, we consider a Gibbs
 sampler to sequentially scan all
coordinates $\bdbeta^j$ for $1\leq j \leq N$.

%For each $j$,
Denote by
$\bdbetamj=(\bdbeta^l)_{l\neq j}$.
%As mentioned earlier, direct sampling
%of $\bdbeta^j$ is not
%possible due to the intractability of the conditional likelihood of
%the coordinate $j$,
%$\qj(\cdot|\bdbetamj,\bdy,\theta)$.
We consider here a hybrid Gibbs sampler
i.e. each step of the Gibbs sampler includes a
Metropolis-Hastings step. The proposal law is
chosen as $\qj(\cdot|\bdbetamj,\theta)$ i.e. the  conditional law
based on the current parameter value $\theta$ derived from the normal
distribution $\qmiss$.

 If $b$ is a proposed value at coordinate $j$, the acceptance rate of
 the
Metropolis-Hastings algorithm is given by
$$r_{j}(\bdbeta^j,b;\bdbetamj , \theta )=\left[
  \frac{\qj(b |
    \bdbetamj,\bdy,\theta)\qj(\bdbeta^j|\bdbetamj,\theta)}{\qj(\bdbeta^j
    |\bdbetamj,
    \bdy,\theta)\qj(b|\bdbetamj,\theta)}\land 1\right]\,.$$
Since
$$\qj(\bdbeta^j|\bdbetamj,\bdy,\theta)\propto
\qobs(\bdy|\bdbeta,\theta)\qj(\bdbeta^j|\bdbetamj,\theta)\, ,$$
the acceptance rate can be simplified to
$$r_{j}(\bdbeta^j,b;\bdbetamj,\theta)=\left[
  \frac{\qobs(\bdy|\betajb,\theta)}{\qobs(\bdy|\bdbeta,\theta)}\land
  1\right]\,,$$
where for any $b\in\mathbb{R}$ and $1\leq j\leq N$, we denote by
$\betajb$ the unique vector which is equal to $\bdbeta$ everywhere
except at coordinate $j$ where it equals $b$.
An illustration of the hybrid Gibbs sampler can be found in
\cite{Robert_MCMC}.
The following steps are performed
for each coordinate $j$~:

\begin{description}
\item[]
\item[Step 1~: Proposition]. Sample $b$ with respect to the density
  $\qj(.|\bdbetamj,\theta) $. \\
\item[Step 2~: Accept-reject]. Compute $r_j(\bdbeta^j, b;
  \bdbetamj, \theta)$
  and  with probability $r_{j}(\bdbeta^j,b;\bdbetamj,\theta)$, update
  $\bdbeta^j$ to $b$. \\
\end{description}

In
Algorithm \ref{Algo1}, we summarize the transition step of the Markov
chain.
\begin{algorithm}
 \caption{Transition step $k\to k+1$ using a hybrid Gibbs sampler}
\label{Algo1}
\begin{algorithmic}
\REQUIRE $\bdbeta = \bdbeta_k$; $\theta=\theta_k$
\STATE \texttt{Gibbs sampler:}
\FORALL{ $j = 1: N$}
\STATE \texttt{Metropolis-Hastings procedure:}
\STATE $b \sim \qj(\cdot|\bdbetamj,\theta);$
\STATE Compute $
r_{j} (\bdbeta^j,b;\bdbetamj,\theta) =\left[
  \frac{\qobs(\bdy|\betajb,\theta)}{\qobs(\bdy|\bdbeta,\theta)}\land
  1\right]
$
\STATE With probability $r_{j}(\bdbeta^j,b;\bdbetamj,\theta)$, update
$\bdbeta^j$: $\bdbeta^j \gets b$
\ENDFOR
\end{algorithmic}
\end{algorithm}

This yields the transition probability kernel of our Markov chain on
$\bdbeta$~: for coordinate $j$, the kernel is
\begin{multline}\label{eq:kernelj}
\ntrans_{\theta,j}(\bdbeta,d\bdz)=\left( \otimes _{m\neq j} \delta_{\bdbeta^m} (d\bdz^m)
\right) × \left[
 \qj(d\bdz^j|\bdbetamj,\theta)r_{j}(\bdbeta^j,d\bdz^j;\bdbetamj,\theta )
+\right.\\
\left.
\delta_{\bdbeta^j} (d\bdz^j) \int (1- r_{j}(\bdbeta^j,b;\bdbetamj,\theta ) )
 \qj(b|\bdbetamj,\theta) db
\right]
\end{multline}
and
$\Pi_\theta=\Pi_{\theta,N}\circ \cdots\circ \Pi_{\theta,1}$ is therefore the kernel associated with a complete scan.

\section{Convergence analysis}\label{sec:convergence}

 We prove a general theorem on the
convergence of stochastic approximations for which our algorithm
convergence is a special case.
% an application.

The hybrid Gibbs sampler used to generate the ergodic Markov chain
does not satisfy some of
the assumptions of the convergence result presented in
\cite{andrieumoulinespriouret}.
We therefore weaken some of their conditions,
introducing an absorbing set for the stochastic approximation and
weakening their Hölder conditions on some functions of the Markov chain.

\subsection{Stochastic approximation  convergence Theorem}
\label{mainth}
%We give now a convergence result for a general
%truncated procedure similar to the one we described in Algorithm
%\ref{AlgoMoulines}. \\

% We follow the general framework of stochastic approximations usually
% used to approach the zeros of a given function .
% We introduce a sequence of stochastic approximations which
% we prove to converge to the set of zeros.

Let $\mathcal{S}$ be a subset of $\R^{n_s}$ for some integer
$n_s$. Let $X$ be a measurable space.
For all $s\in\mathcal{S}$ let $H_s: \  X \to \mathcal{S} $ be
a measurable
function. Let $\boldsymbol\Delta=(\Delta_k)_k$ be a sequence of positive step-sizes.

%We prove a theorem that, under some assumptions, will ensure the
%convergence of
Define the stochastic approximation sequence $(s_k)_k$
as follows~:
\begin{equation}
  \label{eq:ASsk}
\left\{
\begin{array}{lllll}
s_{k} &=& s_{k-1} + \Delta_{k-1} H_{s_{k-1}}( \bdbeta_{k} )& \text{ with }
\bdbeta_{k} \sim \ntrans_{s_{k-1}}(\bdbeta_{k-1},\cdot)\,, & \text{ if } s_{k-1} \in \mathcal{S} \\
  s_{k} &=& s_{c}  & \text{ with }
\bdbeta_{k} =\bdbeta_c\,, & \text{ if } s_{k-1} \notin  \mathcal{S}\,,
\end{array}\right.
\end{equation}
where $s_c\notin \mathcal{S}$, $\bdbeta_c\notin X$  and
$(\ntrans_s)_{s\in\mathcal{S}}$ is a family
of Markov transition probabilities on $X$.
Denote by $Q_{\Delta} $ the transition which generates $((\bdbeta_k,s_k))_k$.
We consider the natural filtration of the \textit{non-homogeneous}
chain $((\bdbeta_k,s_k))_k$ and denote respectively by
$\mathbb{P}^{\boldsymbol\Delta}_{\bdbeta,s}$ and $ \mathbb{E}^{\boldsymbol\Delta}_{\bdbeta,s}$ the
probability measure and the corresponding expectation generated by this
Markov chain starting at $(\bdbeta,s)$ and using the sequence $\boldsymbol\Delta$.

If the transition kernel $\ntrans_s$ of the Markov chain admits a
stationary distribution $\pi_s$ and if for any $s\in
\mathcal{S}$,  $H_s$ is
integrable with respect to $\pi_s$, then we denote by $h$ the mean field
associated with our stochastic approximation so that~:
\begin{eqnarray*}
   h(s) = \int
    H_s(\bdbeta) \pi_s(\bdbeta)d\bdbeta \, .
\end{eqnarray*}

The algorithm defined in \ref{eq:ASsk} is usually designed to solve
the equation $h(s)=0$ where $h$ is called the mean field function.

%This condition is stated in assumption (\textbf{A2}).\\

 % For any $A\in \mathcal{B}(X\cup \{ \bdbeta_c \} )$ and
% $B\in \mathcal{B}(\mathcal{S}\cup \{ s_c \} )$

% \begin{equation}
%   \label{eq:trannonhom}
%   Q_\Delta (\bdbeta,s,A× B)= \int_A \ntrans_s (\bdbeta,\bdbeta')
%   \mathds{1}_{s+\Delta H(s,\bdbeta')\in B} + \delta_{s_c}(B) \int_A
%   \ntrans_s(\bdbeta,\bdbeta') \mathds{1}_{s+\Delta H(s,\bdbeta')\notin  \mathcal{S}}\,.
% \end{equation}

Let
$(\Kapa_q)_{q\geq 0} $ be a sequence of increasing compact subsets of
$\mathcal{S}$ such as $\cup_{q\geq 0} \Kapa _q = \mathcal{S} $
and $ \Kapa _q \subset  \text{int}(\Kapa _{q+1}) , \forall q \geq 0$. Let
$\bfe =(\varepsilon_k)_{k\geq 0}  $ be a monotone non-increasing sequence of positive
numbers and $\mathrm{K}$ a subset of $X$.

Let $\Phi : X × \mathcal{S} \to \mathrm{K}× \Kapa_0 $ be a measurable
function and $\phi : \ \N \to \Z$ be a function such that $\phi(k)> -k$ for
any $k$. Define the \textit{homogeneous} Markov chain
\begin{equation}
\label{eq:Zk}
(Z_k=(\bdbeta_k,s_k,\kappa_k,\zeta_k,\nu_k))_k
\end{equation}
 on $\mathcal{Z}\triangleq X×  \mathcal{S}×
\N^3$
with the following transition at iteration $k$~:
\begin{itemize}
\item If $\nu_{k-1}=0$ then draw $(\bdbeta_{k} , s_{k}) \sim
  Q_{\Delta_{\zeta_{k-1}}} (\Phi(\bdbeta_{k-1},s_{k-1}),\cdot )$; otherwise draw
  $(\bdbeta_{k} , s_{k}) \sim
  Q_{\Delta_{\zeta_{k-1}}} ((\bdbeta_{k-1},s_{k-1}),\cdot )$;
\item If $\|s_{k}-s_{k-1}\| \leq \varepsilon_{\zeta_{k-1}}$ and $s_{k} \in
  \Kapa_{\kappa_{k-1}}$ then set $\kappa_{k}=\kappa_{k-1}$, $\zeta_{k}=\zeta_{k-1} +1$
  and $\nu_{k} =\nu_{k-1}+1$ ; otherwise set $\kappa_{k}=\kappa_{k-1} +1$,
  $\zeta_{k}= \zeta_{k-1} +\phi(\nu_{k-1})$
  and $\nu_{k} =0$.
\end{itemize}
\vspace{4mm}

%  This
% is a generalisation of the convergence theorem of
% \cite{andrieumoulinespriouret}. Indeed, our model defined in
% (\ref{pre-priors}) and (\ref{priors}) does not satisfy some conditions of
% \cite{andrieumoulinespriouret}.
% On one hand, we can not directly prove that level sets of $w$  are
% compact, so we lighten the condition by introducing an auxiliary set
% $\mathcal{S}_a$ which is
% absorbing. On the other
% hand, the solution of the Poisson equation $s \to g_{\hat\theta(s)}$ related to our
% markov chain
% and  the transition kernel applied to this function $s \to
% \ntrans_{\hat{\theta}(s)} g_{\hat\theta(s)}$ are
% not Hölderian. Therefor, we introduce a weaker Hölder
% condition for the two of them.

%Let first introduce some more notations.
%  We introduce the following
% function: $L : \mathcal{S} × \Te \to \R $ as
% \begin{equation}
% L(s ; \te)=  -\psi(\te) + \langle s,  \phi(\te)\rangle\label{eq:a13} \, .
% \end{equation}
%For any $s\in\mathcal{S}$, let $H_s: \R^N \to \mathcal{S}$ such that
%\begin{equation*}
%  H_s(\beta) \triangleq  S(\beta)-s\label{eq:b13}\, .
%\end{equation*}
%and

Consider the following assumptions, generalized from
\cite{andrieumoulinespriouret}.
Define for any $V:X\to [1,\infty]$ and
any $g:X \to
\R^{n_s}$ the norm
\begin{equation*}
\|g\|_V = \sup\limits_{ \bdbeta \in X } \frac{\|g(\bdbeta)\|}{V(\bdbeta)}\, .
\end{equation*}

\begin{description}
\item[A1'.] $\mathcal{S}$ is an open subset of $\R^{n_s}$, $h: \mathcal{S}
  \to \R^{n_s}$
is continuous and there exists a continuously differentiable function
$w : \mathcal{S} \to [0,\infty[$ with the following properties.
\begin{description}
\item[(i)] There exists an $M_0>0$ such that
\begin{equation*}
\mathcal{L} \triangleq \left\{
s \in \mathcal{S} , \left\langle
\nabla w(s), h(s) \right\rangle =0
\right\}
\subset \{
 s \in \mathcal{S}, \   w(s ) < M_0
\}\,.
\end{equation*}
\item[(ii)]There exists a closed convex set $ \mathcal{S}_a\subset\mathcal{S}$
for which $s\to s+\rho H_s(\bdbeta)\in \mathcal{S}_a$ for any $\rho\in [0,1]$
and $(\bdbeta,s)\in X× \mathcal{S}_a$ ($\mathcal{S}_a$ is
absorbing) and such that for any $M_1 \in
  ]M_0, \infty]$, the set
$
\mathcal{W}_{M_1} \cap \mathcal{S}_a $
 is a compact set of $\mathcal{S}$ where $
\mathcal{W}_{M_1} \triangleq \{
s \in \mathcal{S} , \ w(s) \leq M_1\}
$.
\item[(iii)] For any $s \in \mathcal{S} \backslash \mathcal{L} $
 $ \left\langle
\nabla w(s), h(s) \right\rangle <0 $.
\item[(iv)] The closure of $w(\mathcal{L}) $ has an empty interior.
\end{description}
\vspace{0.2cm}
\item[A2.] For any $ s \in  \mathcal{S}$, the Markov kernel $\ntrans_{s}$
 has a single stationary distribution $\pi_{s}$, $\pi_{s}
  \ntrans_{s}
 = \pi_{s}$.
 In addition for all $s\in \mathcal{S}$, $H_s: X \to \s$ is measurable and
 $ \int_{X} \| H_s(\bdbeta)\| \pi_{s} (d\bdbeta) <\infty$.
\end{description}
\vspace{0.2cm}

\begin{description}
\item[A3'.] For any $s\in \mathcal{S}$, the Poisson equation
   $g-\ntrans_s g = H_s - \pi_s(H_s)$ has a solution $g_{s}$. There
   exist a function $V:X \to [1,\infty] $ such that $ \{\bdbeta
   \in X , V(\bdbeta)<\infty \} \neq \emptyset$,
  constants $a\in ]0,1]$,  $q\geq 1$ and  $p\geq2$ such
  that for any compact subset
  $\Kapa \subset \mathcal{S}$,
 \begin{description}
 \item[(i)]
 \begin{eqnarray}
 \sup\limits_{s\in\Kapa} \| H_s\|_V <\infty \,, \label{A2i1}\\
 \sup\limits_{s\in\Kapa} (\|g_{s}\|_V + \| \ntrans_{s}g_{s}\|_V )< \infty
 \,, \label{A2i2}
\end{eqnarray}
\item[(ii)]
\begin{eqnarray}
 \sup\limits_{s,s'\in\Kapa} \|s-s'\|^{-a} \{ \|g_{s}-g_{s'}\|_{V^{q}}  + \|
 \ntrans_{s} g_{s} -\ntrans_{s'}g_{s'} \|_{V^{q}} \}<
 \infty \,.
 \end{eqnarray}
\item[(iii)] Let $k_0$ be an integer. There exist an $\bar\varepsilon>0$ and a
  constant $C$ such that
   for any sequence
$\bfe=(\varepsilon_k)_{k \geq 0}$ satisfying $0<\varepsilon_k\leq \bar\varepsilon$ for all $k\geq k_0$,
for any sequence
 $\boldsymbol\Delta= (\Delta_k)_{k\geq 0}$ and for any
   $\bdbeta\in X$,
 \begin{equation}
\label{eq:majPQ}
 \sup\limits_{s\in \Kapa} \sup\limits_{k \geq 0} \mathbb{E}_{\bdbeta,s}^{\boldsymbol\Delta} \left[
 V^{pq}(\bdbeta_k)\mathds{1}_{ \si(\Kapa) \land \nu(\boldsymbol\varepsilon) \geq k}
 \right]\leq C V^{pq} (\bdbeta) \,,
 \end{equation}
 where $ \nu(\boldsymbol\varepsilon)  = \inf\{k\geq 1 , \|s_k-s_{k-1}\|\geq \varepsilon_k\}$ and $\si(\Kapa) =
 \inf\{k\geq 1 ,s_k\notin \Kapa\} $ and the expectation is related to the
 non-homogeneous Markov chain $((\bdbeta_k, s_k))_{k\geq 0}$ using the
 step-size sequence
 $(\Delta_k)_{k\geq 0}$.
\end{description}

\item[A4.] The sequences $\boldsymbol\Delta=(\Delta_k)_{k\geq 0} $ and
  $\boldsymbol\varepsilon=(\varepsilon_k)_{k\geq 0} $   are non-increasing,
  positive and satisfy: $\sum\limits_{k=0}^\infty \Delta_k =\infty$,
$\lim\limits_{k\to\infty} \varepsilon_k =0$ and
$ \sum\limits_{k=1}^\infty \{
\Delta_k^2 +
\Delta_k \varepsilon_k^a +
(\Delta_k \varepsilon_k^{-1})^p  \}
 <\infty$,
where $a $ and $p$ are defined in $(\textbf{A3'})$.
\end{description}
\vspace{0.3cm}

\begin{theorem}[General Convergence Result for Truncated Stochastic
  Approximation]\label{maintheo}
Assume
  (\textbf{A1}'),(\textbf{A2}), (\textbf{A3'}) and
  (\textbf{A4}). Let $\mathrm{K}\subset X$ be such that
  $\sup\limits_{\bdbeta\in\mathrm{K}} V(\bdbeta) < \infty $ and
$\Kapa_0 \subset \mathcal{W}_{M_0} \cap \mathcal{S}_a$ (where $M_0$ is defined
in (\textbf{A1}')),
 and let
    $(Z_k)_{k\geq 0}$ be the sequence defined in equation \eqref{eq:Zk}.
Then, for all $\bdbeta_0 \in \mathrm{K}$ and
    $s_0\in \Kapa_0$, we have $\lim\limits_{k\to\infty}
    d(s_k,\mathcal{L})=0$ $\bar{\mathbb{P}}_{\bdbeta_0, s_0, 0,0,0} $-a.s,
    where $ \bar{\mathbb{P}}_{\bdbeta_0, s_0, 0,0,0} $ is the
    probability measure
    associated with the chain $(Z_k=(\bdbeta_k,
    s_k,\kappa_k,\zeta_k,\nu_k))_{k\geq 0} $ starting at $(\bdbeta_0,s_0,0,0,0) $.

\end{theorem}

\vspace{0.3cm}

\begin{proof}
$\bullet$
The deterministic results obtained by \cite{andrieumoulinespriouret}
under their assumption (\textbf{A1}) remain true if we suppose the
existence of an absorbing set as defined in
assumption (\textbf{A1'}). Indeed, the proofs in
\cite{andrieumoulinespriouret} can be carried through in
the same way restricting the sequences to the absorbing set. Therefore
we obtain the same properties. The first one (stated in Lemma 2.1 of
\cite{andrieumoulinespriouret})
gives the contraction property of the
Lyapunov function $w$. Then, we have (as in Theorem 2.2 of
\cite{andrieumoulinespriouret}) the fact that
a sequence of stochastic approximations stays almost surely in a
compact set under some conditions on the perturbation. Lastly, we
establish the  convergence of such a stochastic approximation.\\

 $\bullet$ We then state a relation between the homogeneous and
 non-homogeneous chains as done in Lemma 4.1 of
 \cite{andrieumoulinespriouret}.
 \\

$\bullet$ We now prove an equivalent version of Proposition $5.2$ of
\cite{andrieumoulinespriouret} under our conditions.
Indeed the upper bound on the
fluctuations of the noise sequence stated in
this proposition is relaxed in our case, involving a different power on the
function $V$.

\vspace{2mm}

\begin{proposition}\label{prop1}
  Assume (\textbf{A3}'). Let $\Kapa $ be a compact subset of
  $\mathcal{S}$ and let $\boldsymbol\Delta=(\Delta_k)_k $ and $\bfe=(\varepsilon_k)_k$
  be two non-increasing sequences of positive numbers such that
  $\lim\limits_{k\to\infty}\varepsilon_k=0 $. Then, for $p$
  defined in (\textbf{A3}'),
\begin{description}
\item[1.] there exists a constant $C$ such that, for any $(\bdbeta,s)\in
  X × \Kapa$, any integer $l$, any $\delta>0$
\begin{eqnarray*}
  \label{eq:5.2.1}
  \mathbb{P}_{\bdbeta,s}^{\boldsymbol\Delta}  \left(
\sup\limits_{n\geq l} \|S_{l,n}(\boldsymbol\varepsilon,\boldsymbol\Delta,\Kapa)\| \geq \delta
\right)  \leq C \delta^{-p} \left\{
\left(
\sum\limits_{k=l}^\infty \Delta_k^2
\right)^{p/2} +
\left(
\sum\limits_{k=l}^\infty \Delta_k \varepsilon_k^a
\right)^{p}
\right\} V^{pq}(\bdbeta) \,,
\end{eqnarray*}
where $ S_{l,n}(\boldsymbol\varepsilon,\boldsymbol\Delta,\Kapa) \triangleq
\mathds{1}_{\si(\Kapa) \land  \nu(\boldsymbol\varepsilon) \geq n }
\sum\limits_{k=l}^n \Delta_k (H_{s_{k-1}}(\bdbeta_k ) - h(s_{k-1}))
$ and
$\mathbb{P}_{\bdbeta,s}^{\boldsymbol\Delta}  $ is the probability measure generated by the non
homogeneous Markov chain $((\bdbeta_k,s_k))_k$ started from the
initial condition $(\bdbeta,s) $;

\item[2.]
there exists a constant $C$ such that for any $(\bdbeta,s)\in  X ×
\Kapa$
\begin{eqnarray*}
  \label{eq:5.2.2}
  \mathbb{P}_{\bdbeta,s}^{\boldsymbol\Delta} ( \nu(\boldsymbol\varepsilon) < \si(\Kapa) ) \leq C \left\{
    \sum\limits_{k=l}^\infty (\Delta_k \varepsilon_k^{-1} )^p
\right\}   V^{pq}(\bdbeta) \, .
\end{eqnarray*}

\end{description}
\end{proposition}

\begin{proof}
The proof of this proposition can proceed as in
\cite{andrieumoulinespriouret} except for the upper bound on the term
involving the Hölder property (second term in the following).
Under  (\textbf{A3'(ii)}), this
upper bound brings into play an exponent $pq$ on the function $V$.

Indeed, rewrite $ S_{1,n}(\boldsymbol\varepsilon,\boldsymbol\Delta,\Kapa)$ using the Poisson
equation and decompose it into a sum of the following five terms~:
\begin{eqnarray}
  \label{eq:Tn1-5}
  T_n^{(1)} & = & \sum\limits_{k=1}^{n} \Delta_k (g_{s_{k-1}} (\bdbeta_k)
  - \ntrans_{s_{k-1}}g_{s_{k-1}} (\bdbeta_{k-1})) \mathds{1}_{ \{\sigma(\Kapa)
 \land  \nu(\bfe)\geq k \}} \\
 T_n^{(2)} & = & \sum\limits_{k=1}^{n-1}  \Delta_{k+1}  (\ntrans_{s_k} g_{s_{k}}(\bdbeta_k)
   -\ntrans_{s_{k-1}}  g_{s_{k-1}}(\bdbeta_k) )  \mathds{1}_{ \{\sigma(\Kapa)
 \land  \nu(\bfe)\geq k+1\}}  \\
 T_n^{(3)} & = & \sum\limits_{k=1}^{n-1}  (\Delta_{k+1}  - \Delta_{k}
 )  \ntrans_{s_{k-1}}  g_{s_{k-1}}(\bdbeta_k)  \mathds{1}_{ \{\sigma(\Kapa)
 \land  \nu(\bfe)\geq k+1\}}  \\
 T_n^{(4)} & = &  \Delta_1 \ntrans_{s_0} g_{s_0}(\bdbeta_0)
 \mathds{1}_{ \{\sigma(\Kapa)
 \land  \nu(\bfe)\geq 1 \}}- \Delta_n \ntrans_{s_{n-1}}
g_{s_{n-1}}(\bdbeta_n) \mathds{1}_{ \{\sigma(\Kapa)
 \land  \nu(\bfe)\geq n \}} \\
 T_n^{(5)} & = & -\sum\limits_{k=1}^{n-1}  \Delta_{k}
 \ntrans_{s_{k-1}} g_{s_{k-1}}(\bdbeta_k) \mathds{1}_{ \{\sigma(\Kapa)
 \land  \nu(\bfe)= k\}} \,.
\end{eqnarray}

We evaluate bounds for the first four quantities. Using the Minkowski
inequality for $p/2\geq 1$ and the Burkholder inequality (for
$T^{(1)}_n$) we have~:
\begin{eqnarray}
  \label{eq:MajTn1-5}
\sup\limits_{s\in \mathcal{S}} \mathbb{E} ^{\boldsymbol\Delta} _{\bdbeta_0,s}
\left[ \sup\limits_{n\geq 0} \left\| T^{(1)} _n \right\|^p \right] &\leq& C
\left(
\sum\limits_{k=1}^\infty \Delta_k ^2
\right)^{p/2} \sup\limits_{s\in \mathcal{S}} \sum\limits_{k} \mathbb{E}
^{\boldsymbol\Delta} _{\bdbeta_0,s} \left[ V^{p} (\bdbeta_k) \mathds{1}_{
    \{\sigma(\Kapa)  \land  \nu(\bfe)\geq k \}} \right]\,, \\
\sup\limits_{s\in \mathcal{S}} \mathbb{E} ^{\boldsymbol\Delta} _{\bdbeta_0,s}
\left[ \sup\limits_{n\geq 0} \left\| T^{(2)} _n \right\|^p \right] &\leq& C
\left(
\sum\limits_{k=1}^\infty \Delta_k \varepsilon_k^\alpha
\right)^{p} \sup\limits_{s\in \mathcal{S}} \sum\limits_{k} \mathbb{E}
^{\boldsymbol\Delta} _{\bdbeta_0,s} \left[ V^{pq} (\bdbeta_k) \mathds{1}_{
    \{\sigma(\Kapa)  \land  \nu(\bfe)\geq k \}} \right]\,, \\
\sup\limits_{s\in \mathcal{S}} \mathbb{E} ^{\boldsymbol\Delta} _{\bdbeta_0,s}
\left[ \sup\limits_{n\geq 0} \left\| T^{(3)} _n \right\|^p \right] &\leq& C
\Delta_1 ^p
\sup\limits_{s\in \mathcal{S}} \sum\limits_{k} \mathbb{E}
^{\boldsymbol\Delta} _{\bdbeta_0,s} \left[ V^{p} (\bdbeta_k) \mathds{1}_{
    \{\sigma(\Kapa)  \land  \nu(\bfe)\geq k \}} \right] \,,\\
\sup\limits_{s\in \mathcal{S}} \mathbb{E} ^{\boldsymbol\Delta} _{\bdbeta_0,s}
\left[ \sup\limits_{n\geq 0} \left\| T^{(4)} _n \right\|^p \right] &\leq& C
\left(
\sum\limits_{k=1}^\infty \Delta_k ^2
\right)^{p/2} \sup\limits_{s\in \mathcal{S}} \sum\limits_{k} \mathbb{E}
^{\boldsymbol\Delta} _{\bdbeta_0,s} \left[ V^{p} (\bdbeta_k) \mathds{1}_{
    \{\sigma(\Kapa)  \land  \nu(\bfe)\geq k \}} \right]\,.
\end{eqnarray}
where $C$ is a constant which depends only upon the compact
set $\Kapa$.
The higher power $p q$ appears because of the Hölder condition we
assume on the solution of the Poisson equation.

Since now $T^{(5)}_n \mathds{1}_{
    \{\sigma(\Kapa)  \land  \nu(\bfe)\geq n \}} =0 $ and noting that
  $V(\bdbeta)\geq 1, \ \forall \bdbeta \in X$, we have
$V(\bdbeta)^p\leq V^{pq}(\bdbeta)$. Applying successively (as in
\cite{andrieumoulinespriouret}) the Markov inequality,
condition \eqref{eq:majPQ}  and the Markov property to these
upper bounds  concludes the  proof of the first part of
Proposition \ref{prop1}.\\

Concerning the second part, it follows from the same trick as above for
upper-bounding the expectation of $V^p$ by $V^{pq}$.

This ends the proof of the proposition.
\end{proof}
\vspace{2mm}

It is now straightforward to prove the following
proposition which corresponds to Proposition $5.3$ in
\cite{andrieumoulinespriouret}.
\vspace{2mm}

\begin{proposition}
  Assume (\textbf{A3}') and (\textbf{A4}). Then, for any subset
  $\mathrm{K} \subset  X  $ such that $\sup\limits_{\bdbeta\in \mathrm{K}
  } V(\bdbeta)<\infty $, any $M \in (M_0,M_1]$ and any $ \delta>0$, we have
$\lim\limits_{k\to\infty} A(\delta,\bfe^{\gets k} , M,\boldsymbol\Delta^{\gets k} ) = 0 $
where $\bfe^{\gets k}$ stands for the sequence $\bfe$ delayed by $k$
switches ($\bfe^{\gets k} _l = \bfe_{k+l}$ for all $l\in \N$) and
\begin{multline*}
A(\delta,\bfe , M,\boldsymbol\Delta )= \sup\limits_{s\in\Kapa_0}
\sup\limits_{\bdbeta\in\mathrm{K} }
\left\{
  \mathbb{P}_{\bdbeta,s}^{\boldsymbol\Delta}  \left(
\sup\limits_{k\geq 1} \|S_{1,k}(\bfe,\boldsymbol\Delta,\mathcal{W}_M)\| \geq \delta
\right) +  \mathbb{P}_{\bdbeta,s}^{\boldsymbol\Delta} \left( \nu(\bfe) < \si(
  \mathcal{W}_M) \right)
\right\} \, .
\end{multline*}
\end{proposition}
The convergence of the sequence $(s_k)_k$ follows from the proof of
Theorem $5.5$ of \cite{andrieumoulinespriouret} which states the almost
sure convergence due to the previous propositions.
\end{proof}

\begin{rem}
  We can weaken the condition on $p$ given in (\textbf{A3'}).
Indeed, we can assume that
  (\textbf{A3'}) holds for any $p> 0$ as soon as at least condition
  \eqref{eq:majPQ} is true also for power $2$ of $V$.
This is needed in the
  proof
  while giving an upper bound for all the $T_n$'s using the Jensen
  inequality instead of the Minkowski inequality as in
  \cite{andrieumoulinespriouret}. In this case, the assumption
  (\textbf{A4}) would have to be satisfied for a power $\max(p,2)$
  instead of power $2$.
\end{rem}

\subsection{Convergence Theorem for Dense Deformable Template
  Model}\label{ConvTheo}
We now give the convergence result of our estimation process which is
an application of the previous theorem. In this section, we assume
that $\si^2$ is fixed which reduces $\theta$ to $(\alpha, \Gamma_g)$. In fact,
due to the implicit definition of $\hat\theta$ given in equation
\eqref{PhotoUpdateClust}, we were not able to prove the smoothness of
the inverse of the function $s\mapsto \hat\theta (s)$ which is straightforward
for fixed $\si^2$.

We can easily exhibit some of the functions involved in our
procedure. Comparing equation \eqref{eq:sbar} to equation
\eqref{eq:ASsk}, we
have
\begin{equation}
  \label{eq:H}
  H_s(\bdbeta) = S(\bdbeta) -s\,.
\end{equation}
Equation \eqref{PhotoUpdateClust} gives the existence of the function
$s\to \hat\theta(s)$.
We  denote by $l$ the observed log-likelihood~:
$l(\theta)\triangleq \log \int \qcomp(\bdy,\bdbeta,\te)d\bdbeta $, and let $w(s)\triangleq -l\circ
\hat{\theta}(s)$ and $h(s)\triangleq \int H_s(
\bdbeta)\qpost(\bdbeta|\bdy, \hat\theta(s)) d\bdbeta$ for
$s\in \mathcal{S}$.

\begin{theorem}
  \label{th:condition}
The sequence of  stochastic approximations $(s_k)_k$ related to the
model defined
in Section
\ref{obs}  and generated by Algorithms \ref{AlgoMoulines} and
\ref{Algo1}
 satisfies
the assumptions (\textbf{A1}') (\textbf{ii}), (\textbf{iii}),
(\textbf{iv}),(\textbf{A2})  and
(\textbf{A3'}).
\end{theorem}
\vspace{2mm}

\begin{proof}
  The details of the proof  are given in appendix  (Section
  \ref{appendix}).
\end{proof}
\vspace{2mm}

\begin{coro}[Convergence of Dense Deformable Template building
  via Stochastic Approximation]\label{convcoup}
\quad \\
Assume
\begin{enumerate}
\item  there exist  $p\geq 1$ and  $a\in]0,1[$ such that the sequences
  $\boldsymbol\Delta=(\Delta_k)_{k\geq 0} $ and
  $\boldsymbol\varepsilon=(\varepsilon_k)_{k\geq 0} $   are non-increasing,
  positive and satisfy:\\ $\sum\limits_{k=0}^\infty \Delta_k =\infty$,
$\lim\limits_{k\to\infty} \varepsilon_k =0$ and
$ \sum\limits_{k=1}^\infty \{
\Delta_k^2 +
\Delta_k \varepsilon_k^a +
(\Delta_k \varepsilon_k^{-1})^p  \}
 <\infty$;
\item  $\mathcal{L} \triangleq \left\{
s \in \mathcal{S} , \left\langle
\nabla w(s), h(s) \right\rangle =0
\right\}$ is included in a level set of $w$.
\end{enumerate}
\vspace{2mm}
Let $\mathrm{K}$ be a compact subset of $\R^N$ and $\Kapa_0$ a compact
subset of $ S(\mathbb{R}^N)$.

 Let   $(s_k)_{k\geq 0}$ and $(\theta_k)_{k\geq 0}$ be the two sequences defined
 in Algorithms     \ref{AlgoMoulines}  and
\ref{Algo1}.
We denote by
$\mathcal{L}'\triangleq
  \{ \theta\in\hat{\theta}(\mathcal{S}),  \frac{\partial l}{\partial
    \theta}(\te)=0\}$, then $\hat{\theta}(\mathcal{L})=\mathcal{L}'$
  and $$\lim\limits_{k\to
  \infty} d(\te_k,\mathcal{L}')=0 \hspace{5mm}\bar{\mathbb{P}}_{\bdbeta_0, s_0, 0,0,0}
-a.s.$$ for all $\bdbeta_0 \in \mathrm{K}$ and
    $s_0\in \Kapa_0$,
    where $ \bar{\mathbb{P}}_{\bdbeta_0, s_0, 0,0,0} $ is the probability measure
    associated with the chain $(Z_k=(\bdbeta_k,
    s_k,\kappa_k,\zeta_k, \nu_k))_{k\geq 0} $ starting at $(\bdbeta_0,s_0,0,0,0) $.
%  Moreover, $\lim\limits_{k\to
%   \infty} d(\te_k,\mathcal{L})=0$ $\bar{\mathbb{P}}_{\beta_0, s_0} $-a.s, where
% $\mathcal{L}  = \{ \te \in \Te , \
% \frac{\partial l}{\partial \theta}(\te)=0 \}$.
\end{coro}
\hspace{0.3cm}

\begin{proof}
We first  notice that, as mentioned in
  \cite{DLM} (Lemma 2 equation 36), since $\hat\theta$,
  $\phi$ and $  \psi$ are smooth functions, it is easy to relate the convergence of
  the stochastic approximation sequence $(s_k)_k$ to the convergence
  of the estimated parameter sequence $(\theta_k)_k$.

 Then the proof follows from the general stability result Theorem
  \ref{maintheo} stated in
  the subsection \ref{mainth} and from the previous theorem
  \ref{th:condition}.
\end{proof}
\begin{rem}
  Note that condition (1) is easily checked for
  $\Delta_k=k^{-c}$ and $\varepsilon_k=k^{-c'}$ with
  $1/2<c'<c<1$. However, condition (2) has not been successfully
  proved yet and should be relaxed in future work.
\end{rem}

\section{Experiments}\label{experiments}
To illustrate our stochastic algorithm for the deformable template
models, we consider handwritten digit images. For each digit class,
 we learn the template, the
corresponding noise variance and the geometric covariance matrices.
(Note that in this experiment the noise variance is no longer fixed and is estimated as the other parameters). We
use the US-Postal
database which contains a training set of around 7000 images.

Each picture is a $(16× 16)$ gray level image with intensity in
$[0,2]$ where $0$ corresponds to the black background. We will also use these
sets in the special case of a noisy setting by adding independent
centered Gaussian noise to  each image.

To be able to compare the results with the previous deterministic
algorithm proposed in \cite{AAT}, we use
the same samples.
In Figure \ref{fig-training20} below, we show some
of the training images.

\begin{figure}[htbp]\begin{center}
\includegraphics[width=6cm,height=4cm]{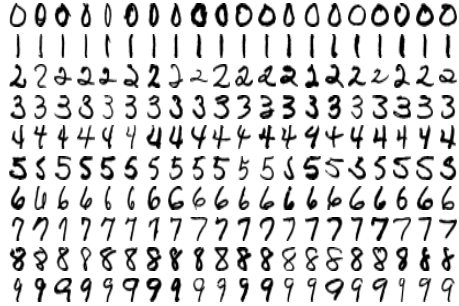}
\end{center}
\caption{Some images from the training set used for the estimation of the model
  parameters (inverse video).}
\label{fig-training20}
\end{figure}

\begin{figure}[h]
  \centerline{\epsfxsize=1.5cm
  \epsfbox{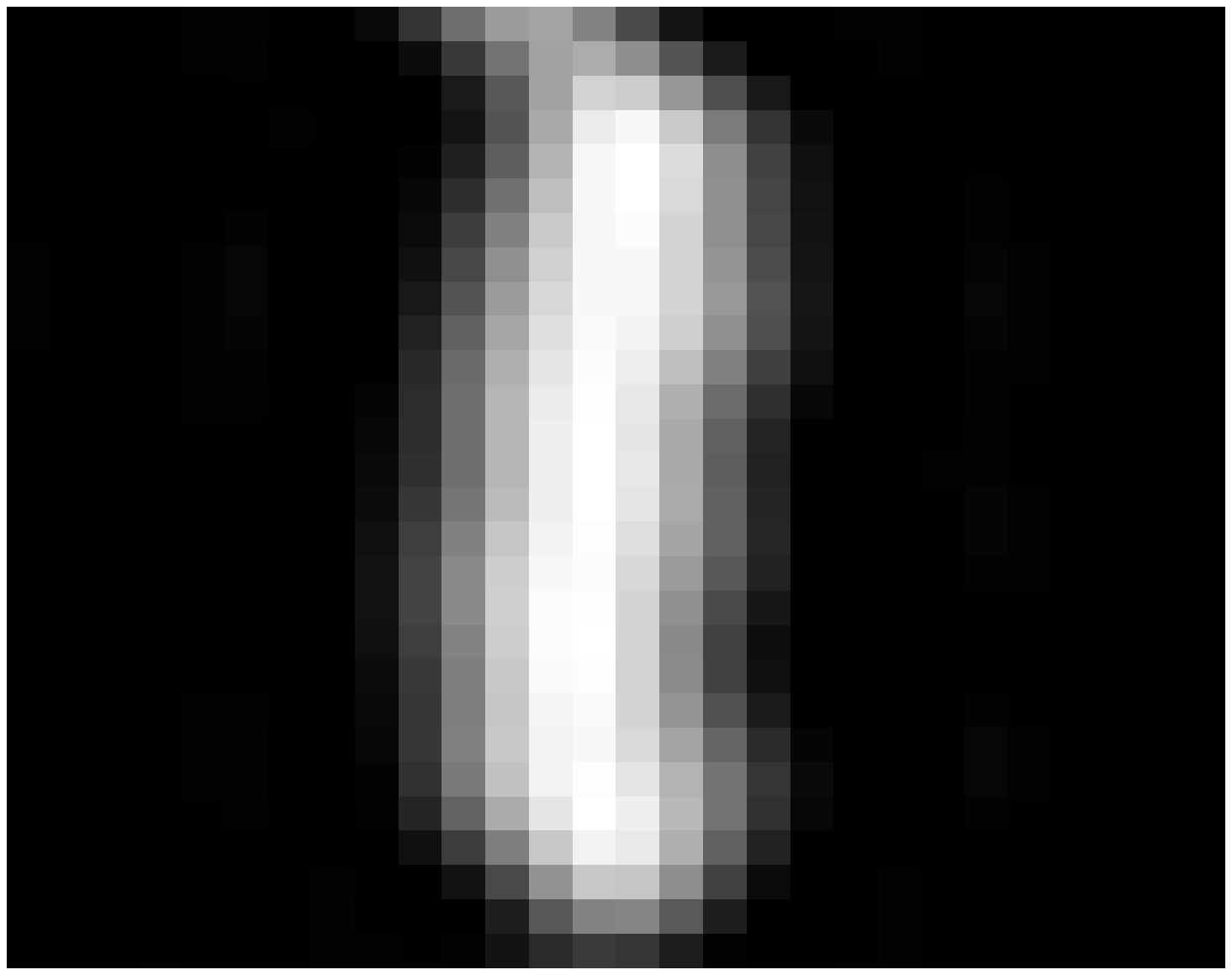} \hspace{0.5cm} \epsfxsize=1.5cm
  \epsfbox{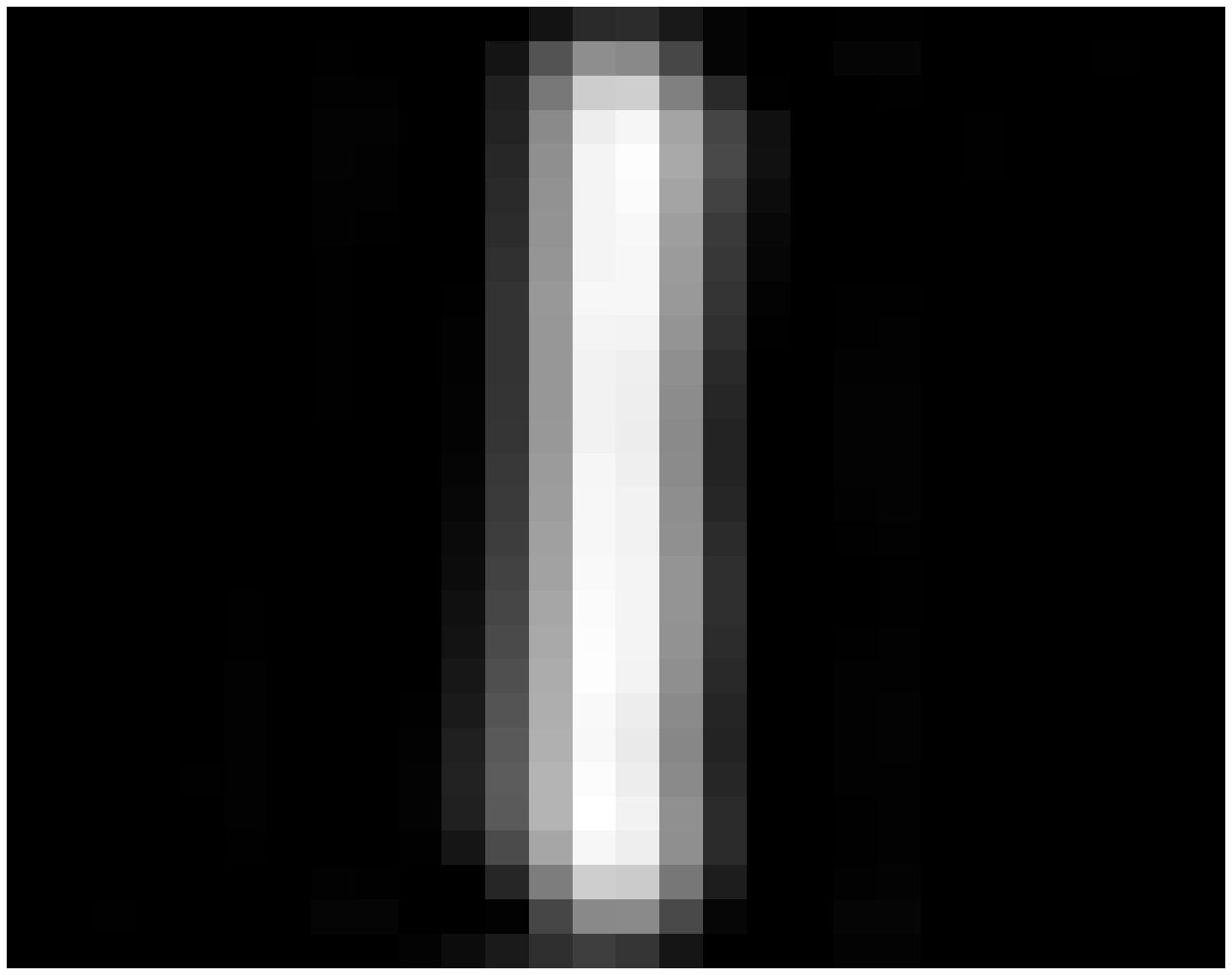} }
\caption{Estimated prototypes of digit 1 (20 images per class) for
  different hyper-parameters. Left: smoother geometry but larger
  photometric covariance in the spline kernel. Right: more rigid
  geometry and smaller photometric covariance.}
 \label{fig-templatedig1}
\end{figure}

\begin{figure}[htbp]\begin{center}
\includegraphics[width=7cm,height=0.8cm]{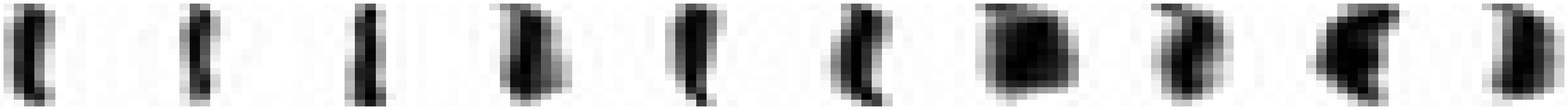}
  \includegraphics[width=7cm,height=0.8cm]{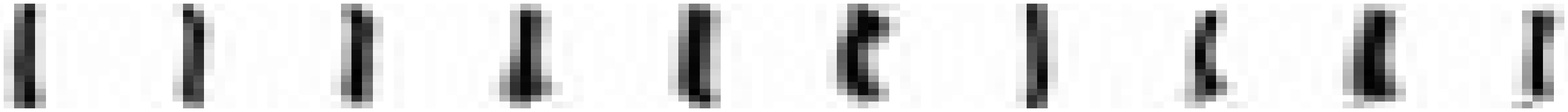}
\end{center}
\caption{Synthetic examples corresponding to the two previous
 estimated templates of digit 1 (inverse video). Left~: with a fatty
 shape. Right~: with a correct shape thickness. }
\label{fig-sampledig1}
\end{figure}

\begin{figure}[h]
  \centerline{\epsfxsize=4.5cm
  \epsfbox{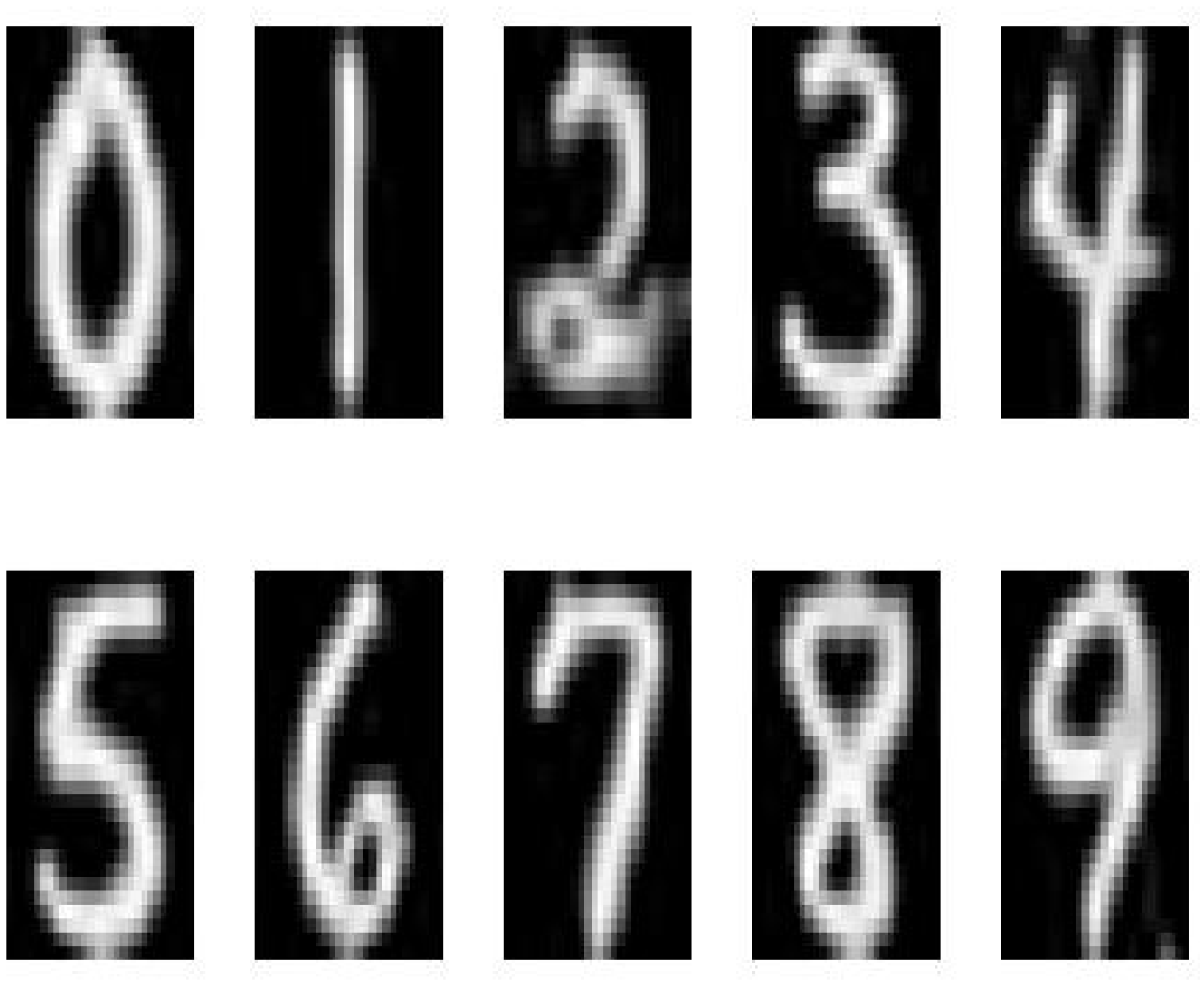} \hspace{0.5cm} \epsfxsize=4.5cm  \epsfbox{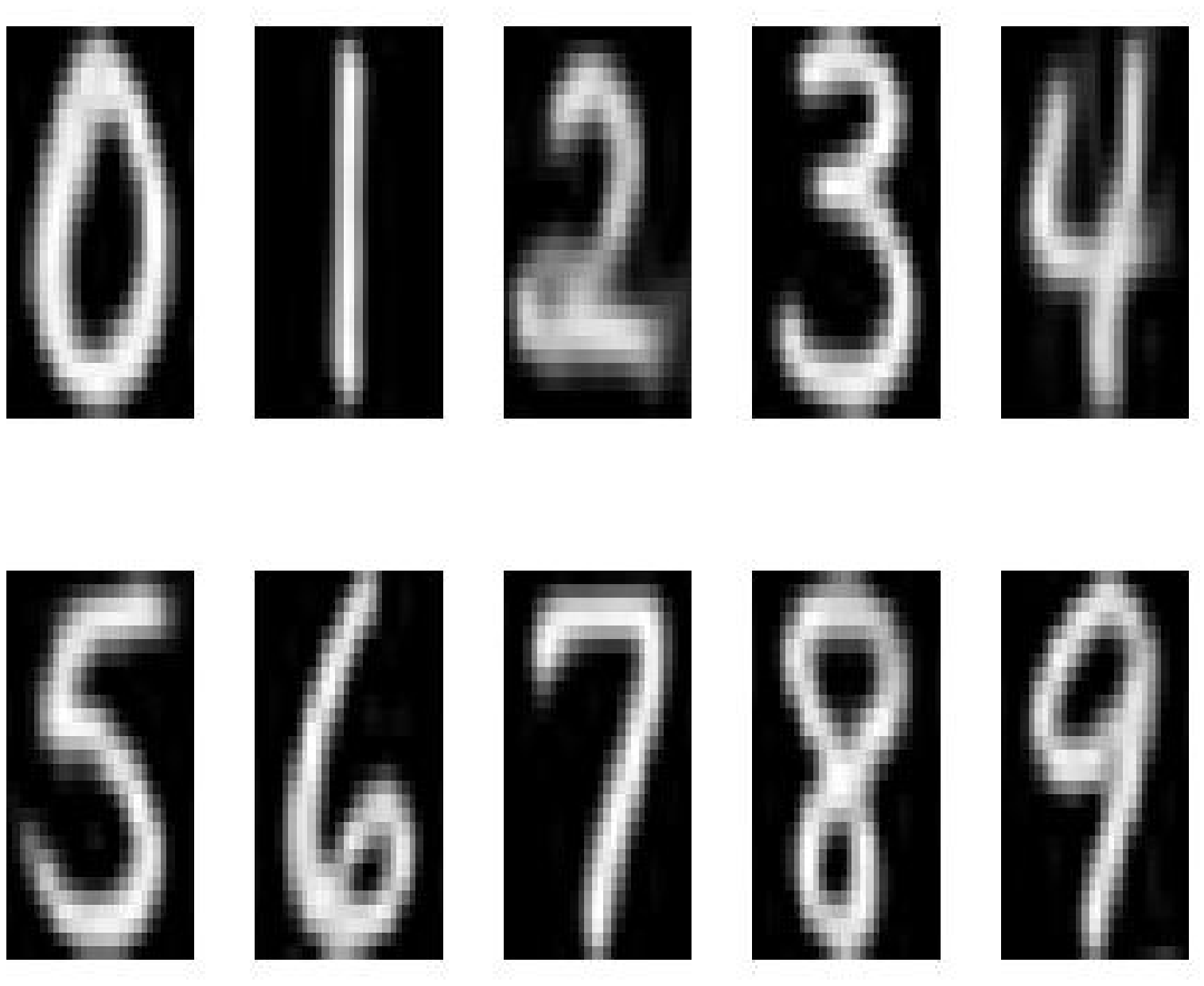}}
\caption{Estimated prototypes issued from left 10 images per class and
  right 20 images per class in the
  training set.}
 \label{fig-template}
\end{figure}

A natural choice for the hyper-parameters on $\alpha$ and $\Gamma_g$ is
$\mu_p=0$ and we
induce the two covariance matrices $\Sigma_p$ and $\Sigma_g$ by the metric of
the Hilbert spaces
$V_p$ and $V_g$ (defined in Section \ref{modeltempdef}) involving the
correlation between the landmarks determined by the kernel.
Define the square matrices
\begin{equation}\label{prior_cov}
\begin{array}{l}
M_p(k,k')=K_p(\pk,\pkp) \ \forall 1 \leq k,k'\leq k_p\\
M_g(k,k')=K_g(\gk,\gkp) \ \forall 1 \leq k,k'\leq k_g \, ,
\end{array}
\end{equation}
then $\Gammap = M_p^{-1}$ and  $\Gammag = M_g^{-1}$.
In our experiments, we have chosen Gaussian kernels for both $K_p$ and
$K_g$, where the standard deviations are fixed at $\sigma_p = 0.12$ and
 $\sigma_g=0.3$. The deformation is computed in the $[-1,1]^2 $ square with
 $k_g=6$ equi-distributed landmarks on this domain. The template has
 been estimated with $k_p=15$ equi-distributed control points on
 $[-1.5,1.5]^2$.

These two covariance matrices are important hyper-parameters; indeed,
 it has been shown in \cite{AAT} that changing the geometric
 covariance has an effect on the sharpness of the template
 images. As for  the photometric hyper-parameter,
 it affects both the template and the geometry in the sense that with
 a large variance, the kernel centered on one landmark spreads out
 to many of its neighbors. This leads to thicker shapes as
 shown in the left panel of Figure \ref{fig-templatedig1}. As
 a consequence, the template is biased: it is not ``centered'' in the
 sense  that
 the mean of the deformations required to fit the data is not close to
 zero. For example for digit ``1'', the main deformations should
 be contractions or dilations of the template. With a large
 variance $\si^2_p$, the template is thicker yielding larger
 contractions and smaller dilations. Since we have set a Gaussian
 law on the deformation variable $\beta$ and $z_{-\beta} =-z_\beta $,
 the deformations $(Id+z_\beta) $ and $(Id-z_\beta) $ have the same probability
 to be drawn under the estimated model.
%% YA what's learnt?? maybe observed?
% learnt, its symmetric
% deformation $(Id-z_\beta) $ will be learnt as well.
% Alain: J'ai corrigé ci-dessus
As shown on synthetic
 examples given in Figure \ref{fig-sampledig1} left
 panel, there
 are many large dilated shapes. However, these examples were not in
 the training set and are not
%% YA accompanied?? Unclear %% Alain: C'est vrai que c'est pas clair. Qu'est-ce qu'on veut dire là ?
 generated with other hyper-parameters (Figure
 \ref{fig-sampledig1} right panel). % This particular
%  issue is due to the spline representation of the template image we
%  set. Indeed, the
%  spline model requires some landmarks on the domain and the variance
%  of the kernel $K_p$ has to be fixed according to the distance between
%  landmarks (and the kind of images treated).
 We have tried different
 relevant values and kept the best with regard to the visual
 results. We present in the following only the results with the
 adapted variances.\\

For the stochastic approximation step-size, we allow a heating period
which corresponds to the absence of memory for the first
iterations. This allows the Markov chain to reach a region of interest
in the posterior probability density function  before
exploring this particular region.

In the experiments presented here, the heating time lasts  $k_h $ (up to
$150$) iterations and the whole algorithm stops after, at most, $200$
iterations depending on the data set (noisy or not). This number of
iterations corresponds to a point where the
convergence seems to have been reached. This yields:
\begin{equation*}
 \Delta_k = \left\{
\begin{array}{ll}
  1 \,, \ &\forall 1 \leq k \leq k_h\\
\frac{1}{(k-k_h)^d}\,, \ &\forall  k > k_h \ \text{ for } d=0.6 \text{ or } 1 \ .
\end{array}\right.
\end{equation*}

To optimism the choice of the transition kernel $\ntrans_\theta $, we have
run the
algorithm with different kernels and compared the evolution of the
simulated hidden variables as well as the results on the estimated
parameters. Some kernels, as the ones mentioned above do not yield good coverage
of the infinite
support of the unobserved variable. From this point of view the hybrid
Gibbs sampler we used has better properties and gives nice estimation
results which are presented below.

\subsection{Estimated Template}
We show here the results of the statistical learning algorithm
 for this model. Figure \ref{fig-template} shows two runs of the
 algorithm for a non-noisy database with $10$
 and $20$ images per class.
Ten images per class are enough to obtain
satisfactory template images with high contrast.

Although it was proved in \cite{AAT}
that the Kullback-Leibler divergence between $q(\cdot;\tilde\te)$ and the common
density function for observations from a given class
converges to its minimal value on the family $q(\cdot;\te)$, we note
that increasing the number of training images does not significantly
improve the estimated photometric template. This apparently surprising fact
can be explained as follows~: since strong variations
in appearance among the images may happened within a given class (think about
topological changes for instance), the image distribution can not be
perfectly represented as a distribution around a single template. This
distribution is better represented  as clustered around a major template
and minor ones in a multimodal way. When the sample size is moderate, with a
high probability, the sample contains basically images around the major mode
and parametric model fits these data quite accurately. When the sample size
increases, the minor modes start to play a significant role as
``outliers'' with
respect to the major mode in the data, resulting in a slightly more blurry template trying to accommodate the different modes.
One way to overcome this fact is to use
 some clustering methods as proposed in \cite{AAT}.
% Although it was proved in \cite{AAT}
% that the Kullback-Leiber divergence between $q(\cdot;\hat\te_n)$ and the common density function of the observations
% converges to zero, we   note that increasing the
%  number of training images does not significantly improve the estimated
% photometric template.
%   Indeed, when there are similar images, the model fits
%  these data perfectly. When there are ``outliers'' in the data, the
%  estimated variability is often relatively large, resulting in less
%  sharp estimated parameters. One way to overcome this fact is to use
%  some clustering methods as proposed in \cite{AAT}.
% The stochastic
 %version of the multi-component model introduced in that paper is in progress.
%Figure (\ref{fig-template}) shows two runs of the one
% component algorithm for a non noisy database with $10$
% and $20$, respectively, images per class.
To visualize robustness with respect to the training set, we ran
 this algorithm with $20$ images per class randomly chosen in the
 whole database. The different runs are presented in Figure
 \ref{fig:templatesrandom}. The two left images show some templates
 which look like the ones obtained in the left image of Figure
 \ref{fig-template} with the $20$ first examples of the
 database. When outliers appear among the $20$ randomly chosen training images
 the template may become somewhat more blurry.
 This is observed for
 digits '2' and '4' (apparently the most variable digits) in
 the right image of Figure  \ref{fig:templatesrandom}. For digits
 where the outliers are less far from the other images, the templates
 are stable.
% the model will try to explain them as well by enlarging the estimated
% variability. The resulting estimated parameters can thus be less
% accurate.

\begin{figure}[htbp]\begin{center}
  \includegraphics[width=4.5cm,height=3.5cm]{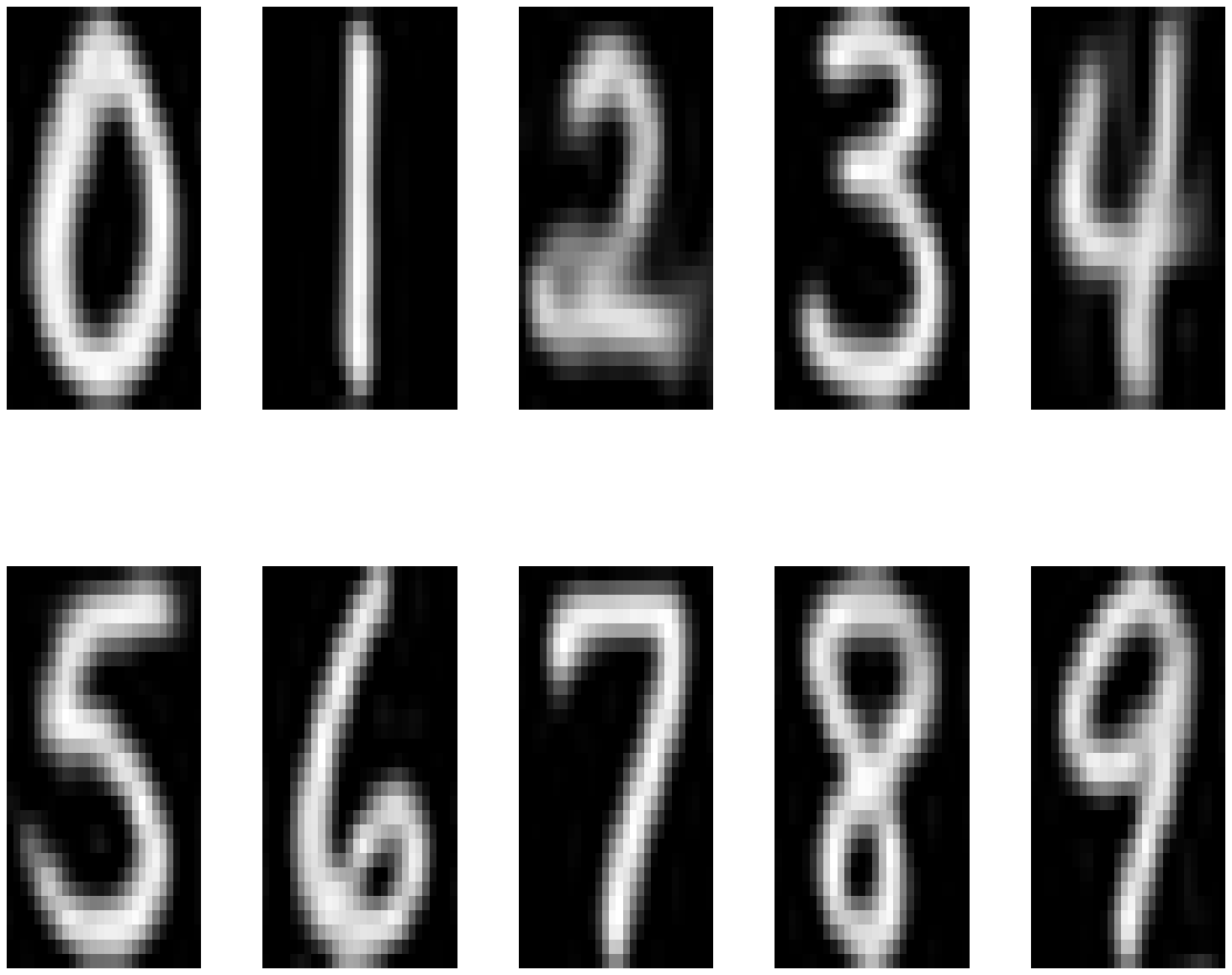}
 \includegraphics[width=4.5cm,height=3.5cm]{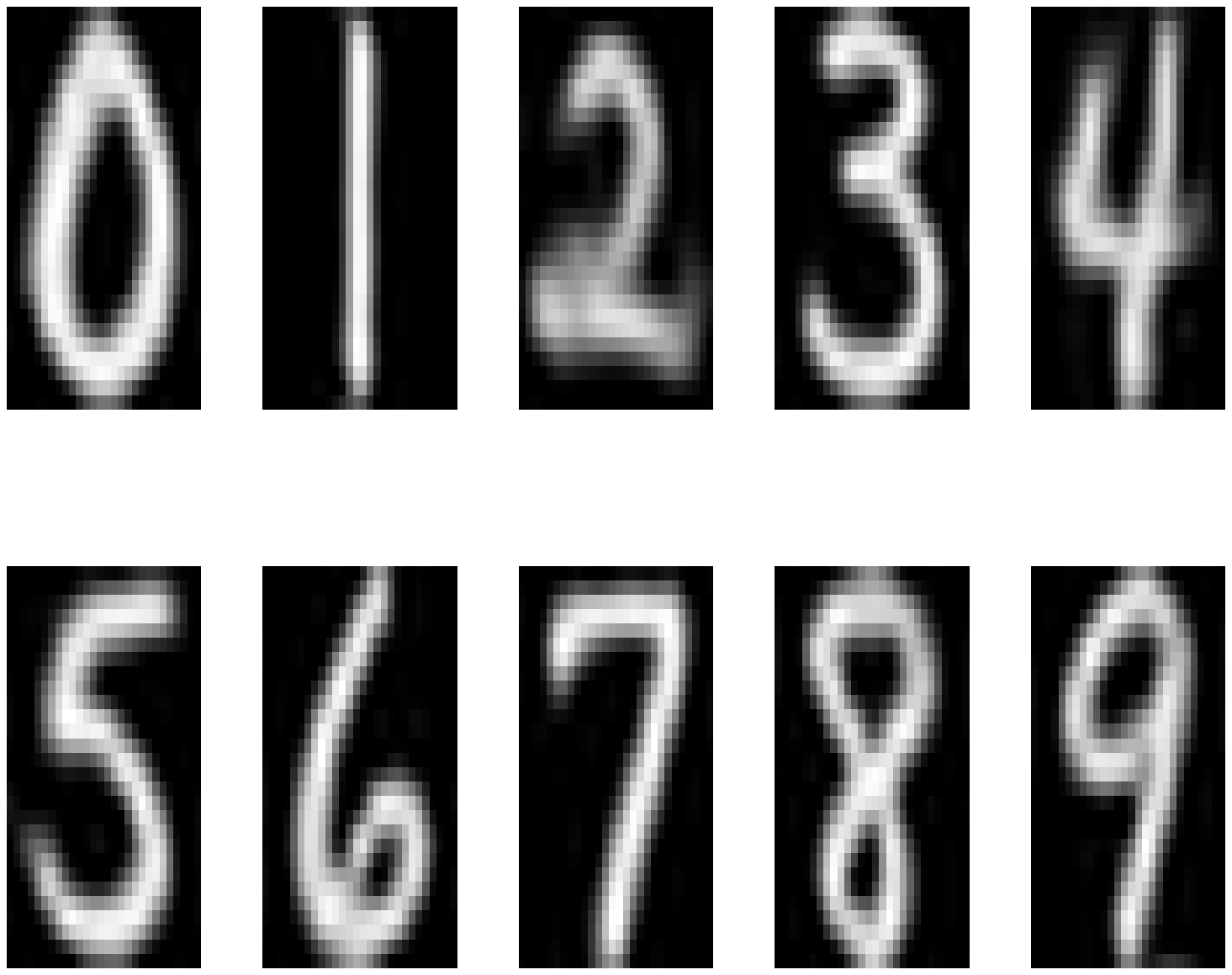}
 \includegraphics[width=4.5cm,height=3.5cm]{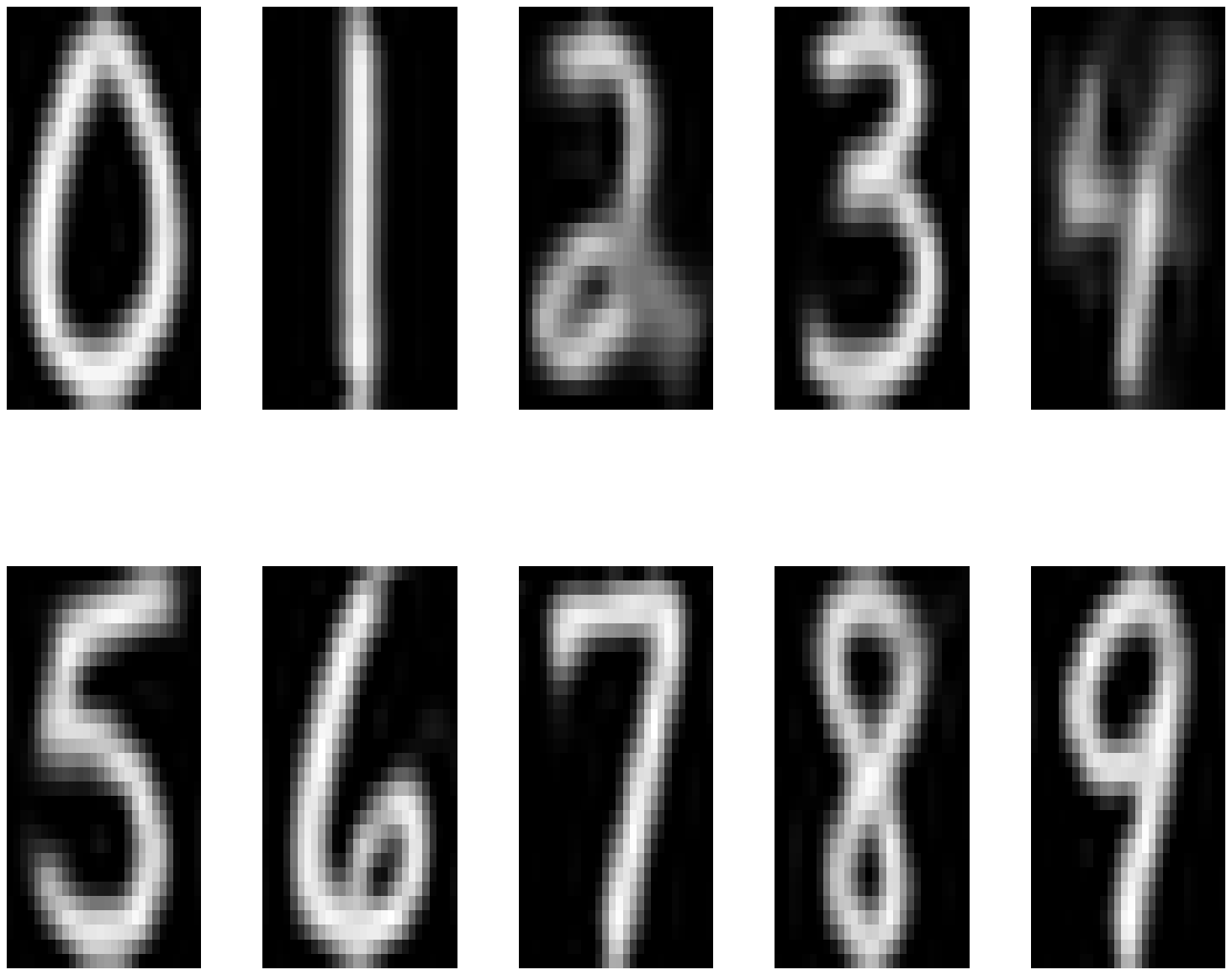}
\end{center}
\caption{Templates estimated with randomly chosen samples from the whole
  USPostal database. Each image is one run of the algorithm with same
  initial conditions but different training sets of $20$ images per
  digit each. The variability of the results is related to the huge
  variability inside the USPS database.}
\label{fig:templatesrandom}
\end{figure}

\begin{figure}[htbp]\begin{center}
\includegraphics[width=4.5cm,height=3.5cm]{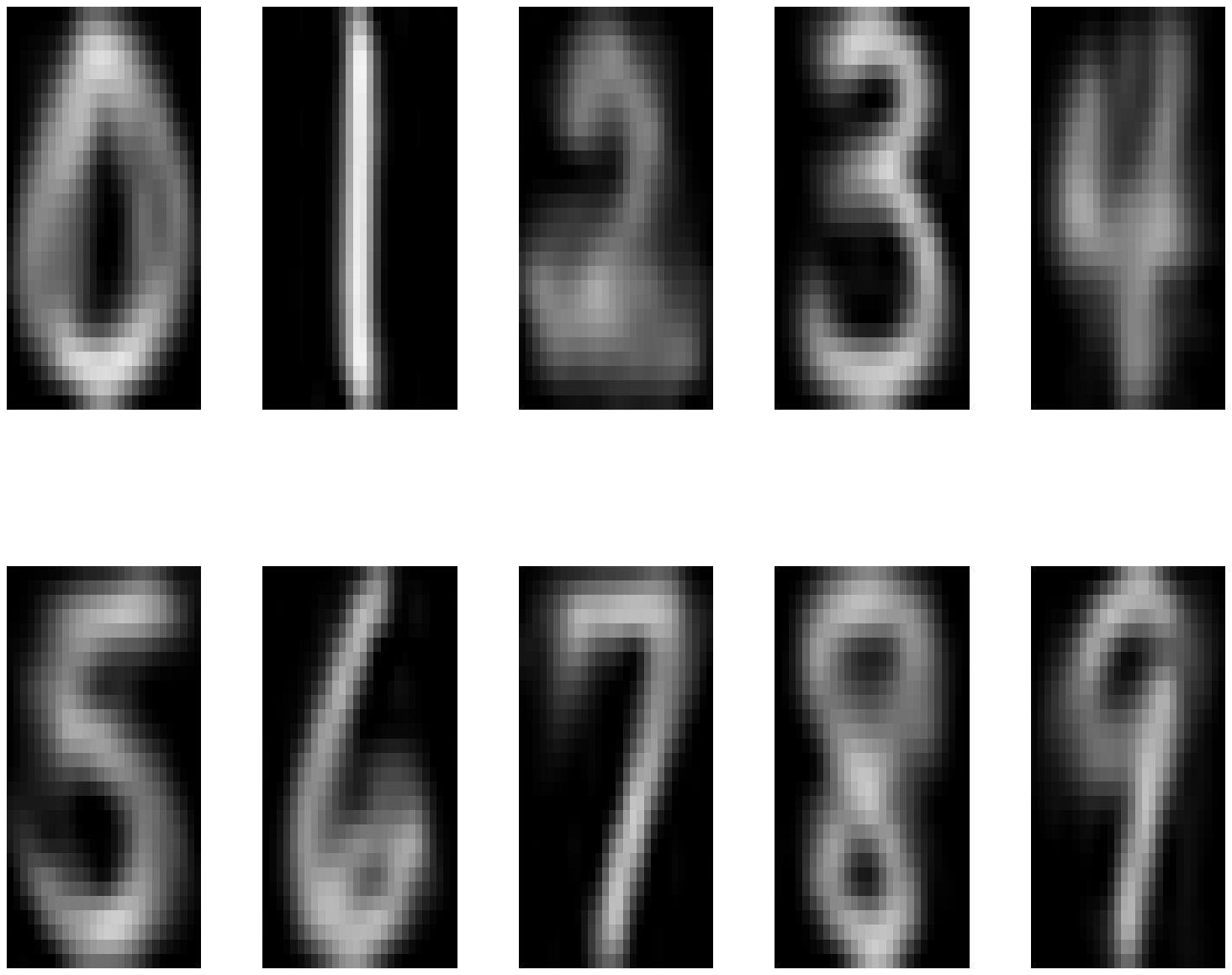}
  \includegraphics[width=4.5cm,height=3.5cm]{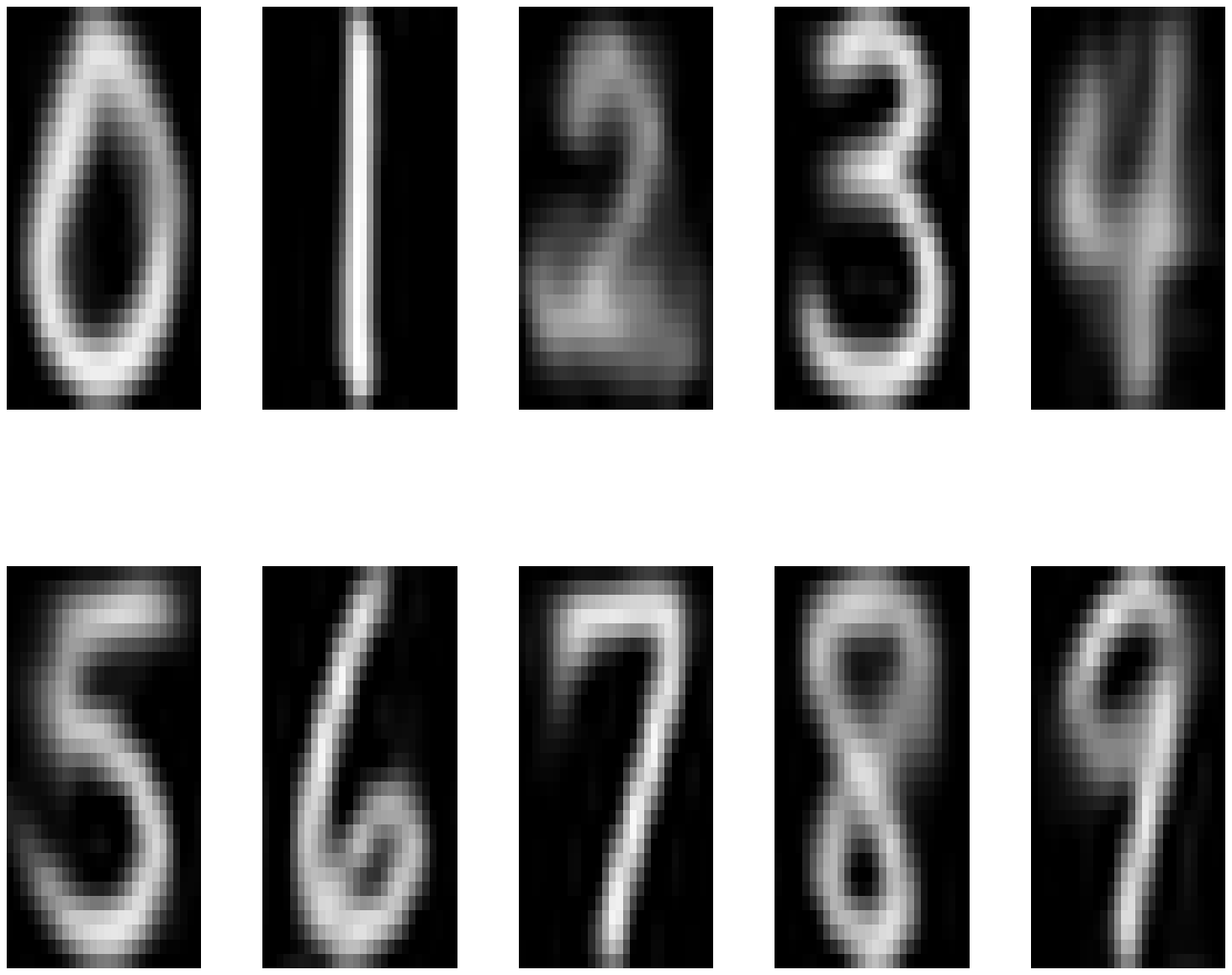}
 \includegraphics[width=4.5cm,height=3.5cm]{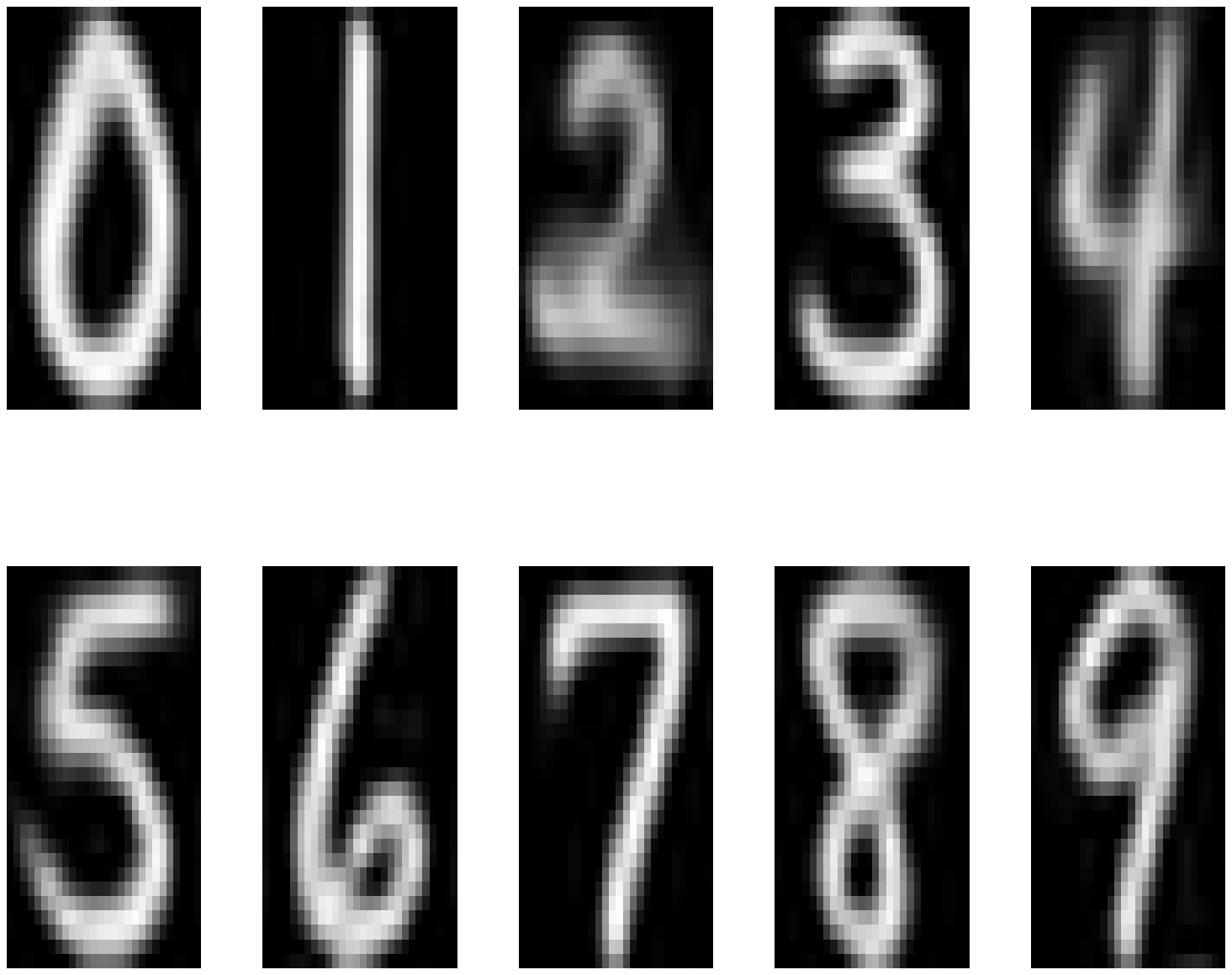}
 \includegraphics[width=4.5cm,height=3.5cm]{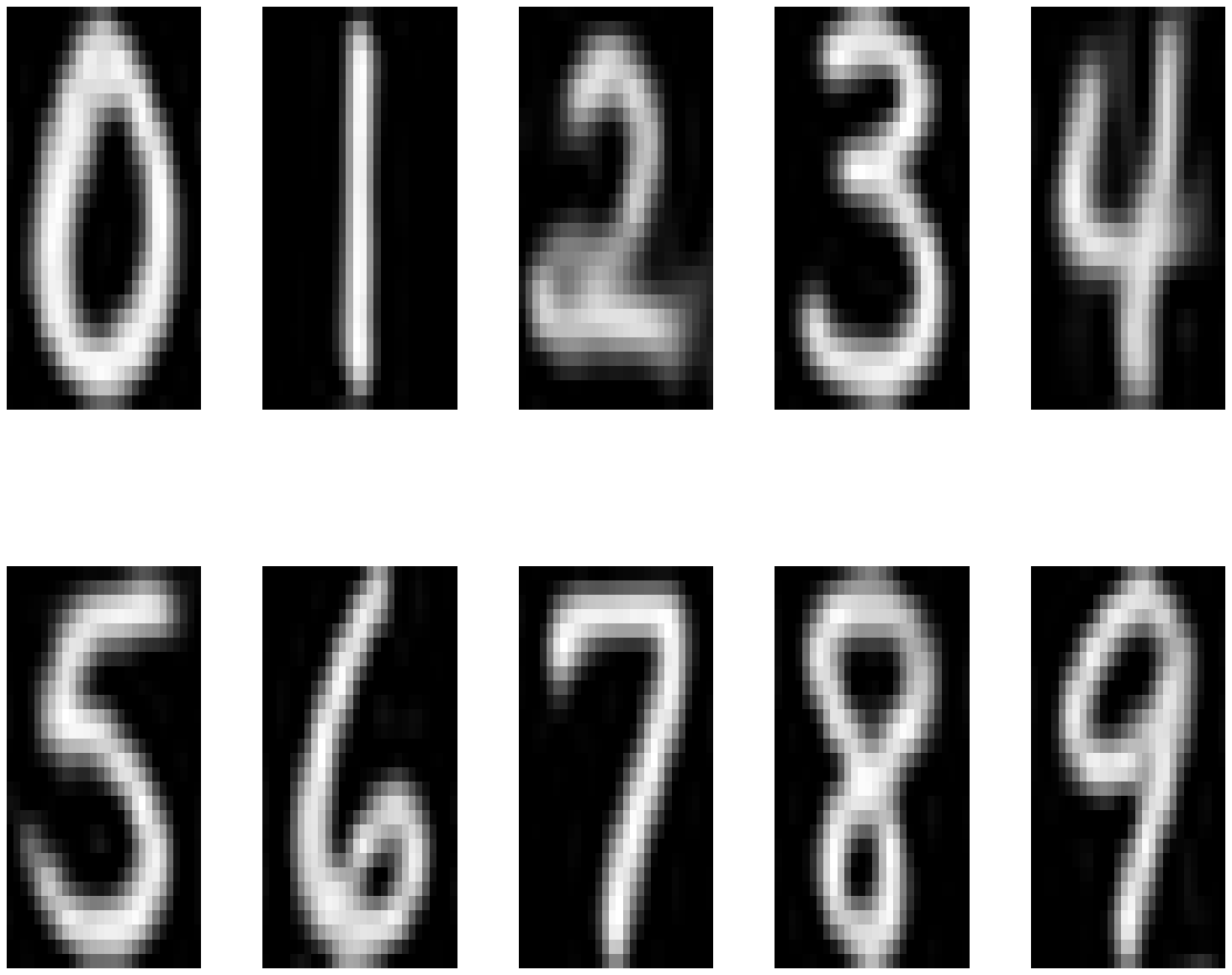}
\end{center}
\caption{Evolution of the templates with the algorithm iterations.
Top line. Left: Mean gray level images of the $20$ training samples.
Middle: template at the 50th
iteration. Right: template at the 100th iteration. Bottom line: template at
the 150th iteration. The improvement is visible~, very fast for some
very simple shapes as digit $1$ and longer for very variable ones as
digit $2$. The higher geometric variability increases the fitting time of
the algorithm.}
\label{fig:evolution}
\end{figure}
 The evolution of the template with  the
 iterations can be viewed in Figure \ref{fig:evolution}. The
 initialization of the template is the mean of the gray level
 images. As the iterations proceed, the templates become sharper.
 In particular, the estimated templates for digits with small
 geometrical  variability converge very fast. For digits like '2'
 or '4', where the geometrical variability is higher, the convergence
 of the coupled parameters (photometry and  geometry) is slowed down.

\subsection{Photometric noise variance}

 \begin{figure}[htbp]\begin{center}
   %\centerline{\epsfxsize=14cm \epsfbox{EvolSig220.eps}}
 \includegraphics[width=14cm,height=8cm]{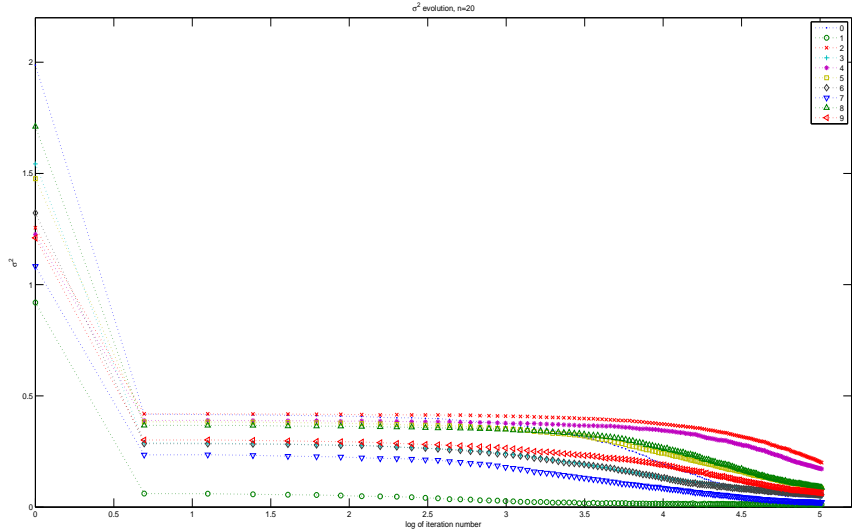}
\end{center}
 \caption{Evolution of the estimated noise variance using $20$ images
   per class along the SAEM-MCMC algorithm. This confirms the visual
   effects seen on the templates~: rapid convergence for some really
   constrained digits and slower convergence for very variable ones.  }
 \label{fig-evolsigp2}
 \end{figure}

The evolution of the noise variance all along the SAEM-MCMC iterations is
the same as the one observed with the ``mode approximation EM''
described in \cite{AAT}.
 During the first
iterations, the noise variance balances the inaccuracy of the estimated
template
which is simply the gray-level mean of the training set.
As the iterations proceed, the template estimates become
sharper as does the estimate of the covariance matrix for the
geometry. This yields very small residual noise. % Figure
                                % (\ref{fig-evolsigp2}) shows the
                                % evolutions of these
% variances for each digit all along the iterations of the algorithm.
 Note that here the final noise variance, which is less than $0.1$,
 for the SAEM-MCMC algorithm for all digits is less than the noise variance
, which is between $0.2$ and $0.3$, for the mode approximation EM
experimented in \cite{AAT}
%% YA what mode? are you refering to the MAX algorithm? %% Alain: On laisse
 in the one component
run.  This can be explained by the stochastic nature of the
algorithm which enables it to escape from local minima provoking early
terminations in the deterministic version.
%% arger part of the posterior density and so fit better the
%% true geometric variability.

 %% The obtained results for the photometric noise variance are small
%%  enough to suggest that
%%  the model we propose to explain the
%%  data is not far from the true generative model which would create the
%%  sample.

\subsection{Estimated geometric distribution}

As mentioned previously we have to fix the value of the hyper-parameter
$a_g$ of the prior on $\Gamma_g$. This quantity plays a significant role in
the results. Indeed, to satisfy the theoretical conditions we have to
choose $a_g$ larger than $4k_g+1$  say $4×36+1 $ in our
examples.
From the geometry update equation, a barycenter between the `sample'
covariance
 and the prior, with the number
$n$  of
images and $a_g$ as coefficients, we find that the prior  dominates
when the training set is small.
The covariance matrix stays close to the prior.
Thus we need to
decrease $a_g$ and find the best trade-off between the degenerate
inverse Wishart and the weight of the prior in the covariance
estimation. We fix this value with a visual criterion: both the
templates and the generated sample with the learnt geometry have to
be satisfactory. This yields $a_g=0.5$ or $0.1$.

%We do note however that the fact that the prior is degenerated does
%not really matter as soon as the posterior distribution is not.
As we have observed from Figure \ref{fig-sample}, parameter
estimation is robust regardless of whether the prior is degenerate or
not. In
addition, considering the update formulas, even if this law does not
have a total weight equal to $1$
it does not affect parameter estimation.
\\

% \begin{figure}[htbp]
%   \centerline{\epsfxsize=15cm \epsfbox{Sample35dig0iv.eps}}
%  \centerline{\epsfxsize=15cm \epsfbox{Sample35dig1iv.eps}}
%  \centerline{\epsfxsize=15cm \epsfbox{Sample35dig2iv.eps}}
% \caption{ 20 synthetic examples for digits $0$, $1$ and $2$ generated
%   with the estimated  parameters.}
% \label{fig-sample}
% \end{figure}

\begin{figure}[htbp]
 \begin{center}
   \includegraphics[width=8cm,height=5cm]{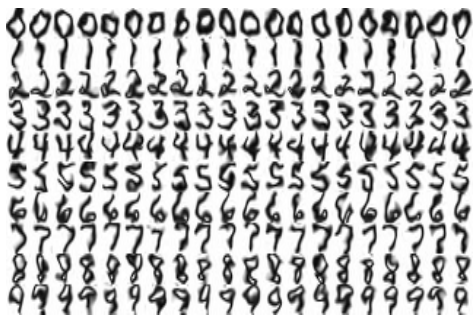}
\end{center}
\caption{Effect of the prior distribution on the deformation~: 20
  synthetic examples per class generated with the estimated
  template but the prior covariance matrix (inverse video).}
\label{fig-samplePrior}
\end{figure}

\begin{figure}[htbp]
 \begin{center}
   \includegraphics[width=8cm,height=9cm]{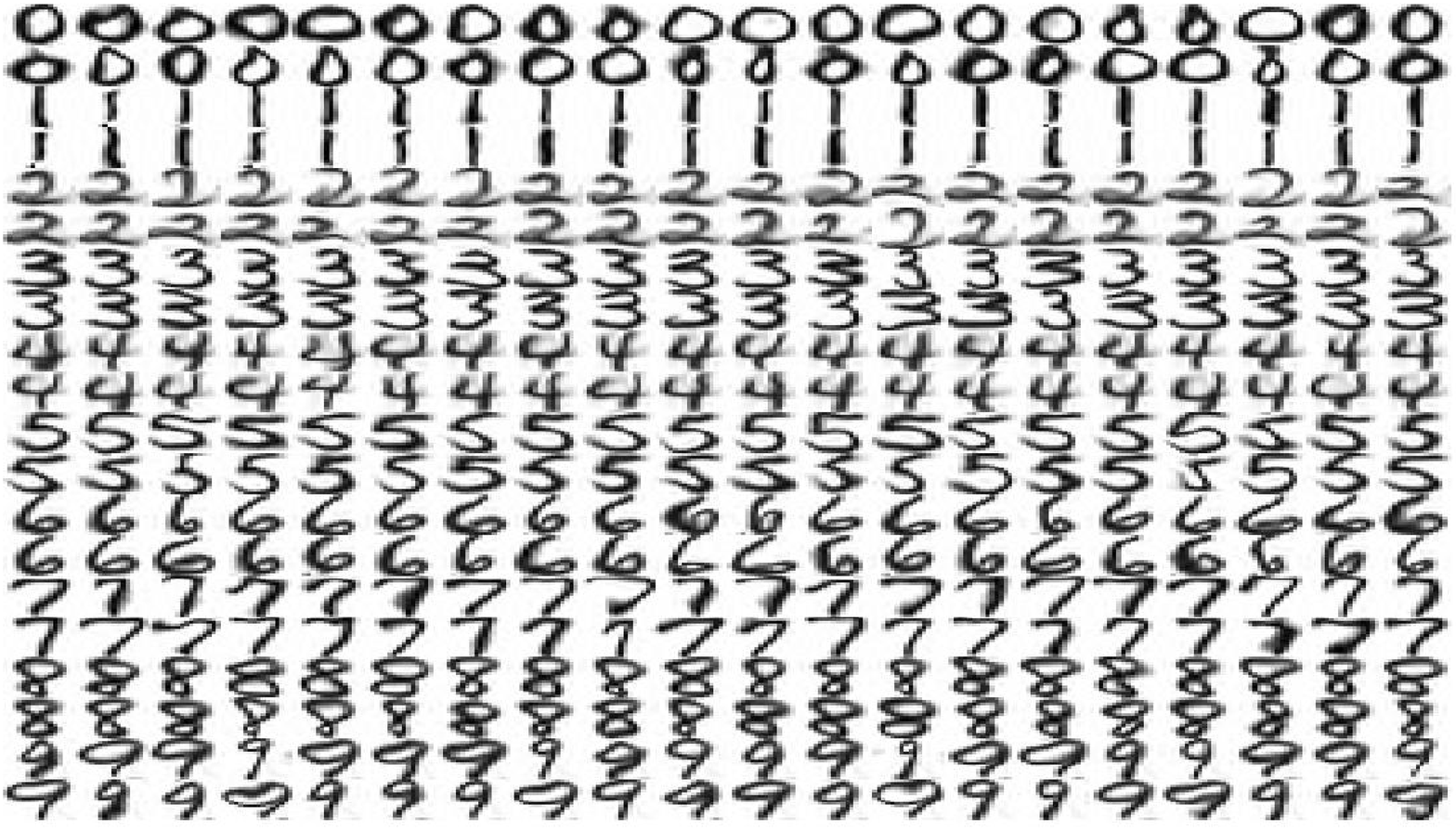}
\end{center}
\caption{Effect of the estimated geometric distribution~:  40
  synthetic examples per class generated with the estimated
  parameters: 20 with the direct deformations and 20 with the
  symmetric deformations (inverse video).}
\label{fig-sample}
\end{figure}

In Figure \ref{fig-sample}, we show a
sample of some synthetic digits modeled by deformation templates drawn
with the estimated parameters.
Note that the resulting digits in Figure \ref{fig-sample} look like
some elements of
the training set and seem to explain these data correctly, whereas the
prior produces some non-relevant local deformations (cf. Figure
\eqref{fig-samplePrior}). In particular,
for some especially geometrically constrained digits such as $0$ or $1$, the
geometry variability reflects their constraints. For digits like the
$2$s, the training set is heterogeneous and shows a large geometrical
variability.
% However it is important to say that they still look like digits which
% is a good point since this covariance matrix was the one used in many
% matching problems before.\\
%
% \begin{figure}[htbp]
%   \centerline{\epsfxsize=19cm \epsfbox{Geo2sur3hyp220.eps}}
% \caption{Top: 3's with its own parameters; Bottom: 3's with the
%   covariance of digit 2.}
% \label{fig-interchangeGeo}
% \end{figure}
%
% To be able to know whether the geometry has been well estimated or
% not in each class, we draw some 3s with its own covariance matrix and
% with the covariance matrix of the
% class 2. This is
% shown in Figure (\ref{fig-interchangeGeo}). We can see the difference
% between the  two lines of samples,  even if we can still recognise the
% class in both examples, the
% deformations in the second case do not seem as natural ones as the
% first line.
%
% We can add an other remark:
When comparing to the deformations obtained by the
mode approximation to EM in \cite{AAT}, it seems that here we obtain a more variable
geometry. This might be because with a stochastic algorithm, we
explore the posterior density and do not only concentrate at its mode. This
allows some more exotic deformations corresponding to realizations of
the missing variable $\bdbeta$ which may belong to the tail of
the law. Another reason may be that for such digits, the mode
approximation gets stuck in a local minimum of
the matching energy. Jumping out of this configuration would require a
large deformation (not allowed by the gradient descent since it would
increase the energy again). However, such a deformation
 can be proposed leading to acceptance by the stochastic
algorithm. Subsequently the deformed template may better fit the observations,
leading to acceptance of these large deformations.
This also leads to a lower value of the
residual noise and may also explain the low noise variance estimated
by the stochastic EM algorithm.

\subsection{Noise effect}

\begin{figure}[htbp]
  \centerline{\epsfxsize=10cm \epsfbox{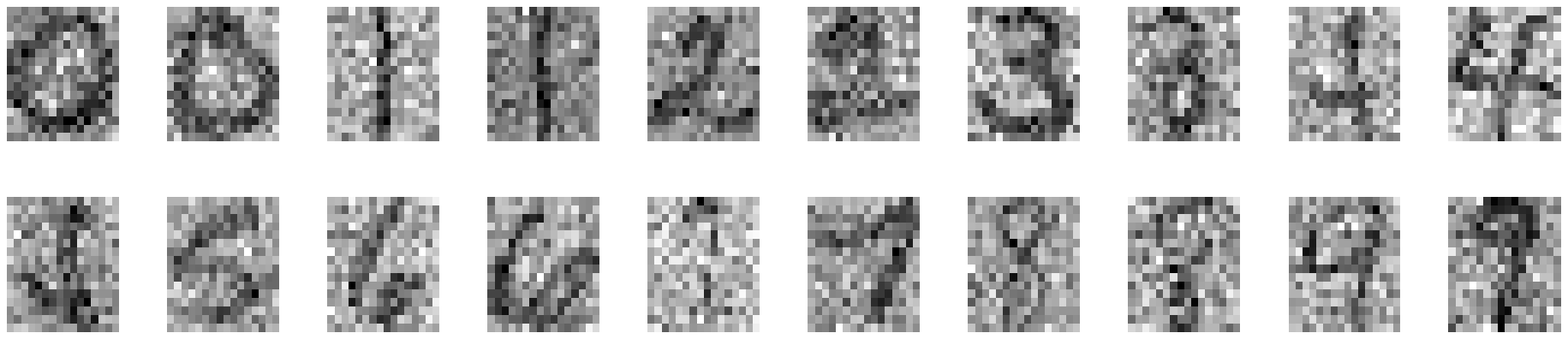}}
  \vspace{0.5cm}

\centerline{\epsfxsize=10cm \epsfbox{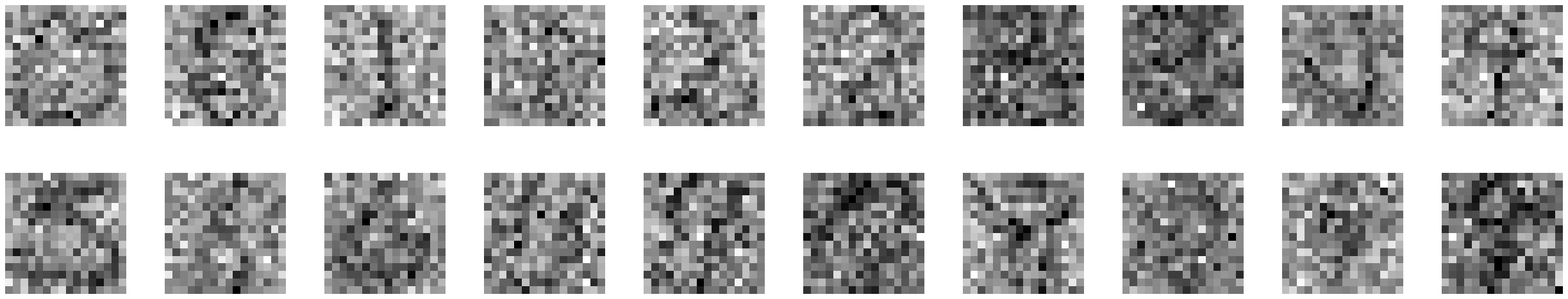}}
\caption{ Two images examples per class of the noisy training set
  (variance: top: $\si^2=1$, bottom:  $\si^2=2$). }
\label{fig-trainingSnr}
\end{figure}
\begin{figure}[htbp]
  \centerline{\epsfxsize=4.5cm
    \epsfbox{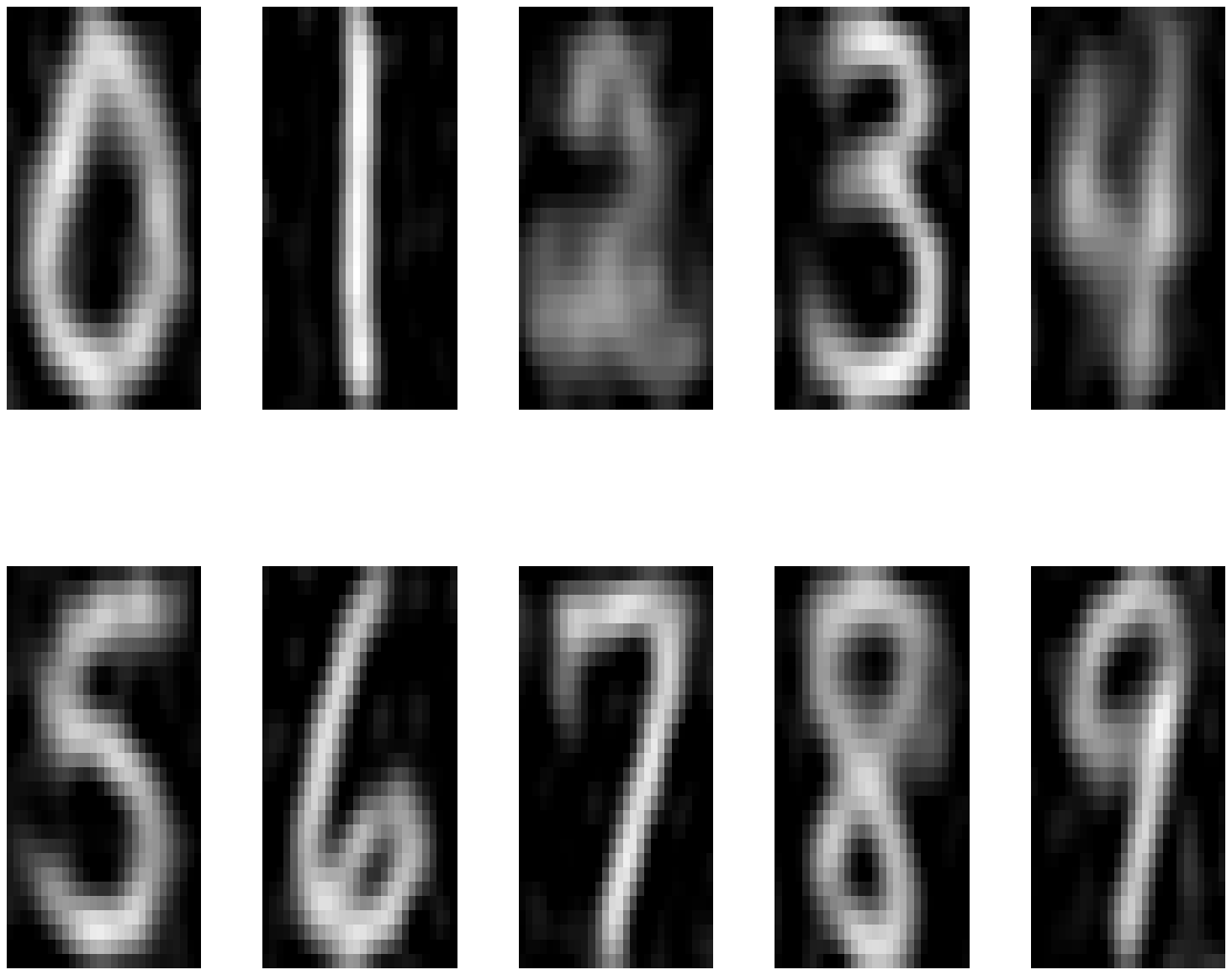} \hspace{0.5cm} \epsfxsize=4.5cm
    \epsfbox{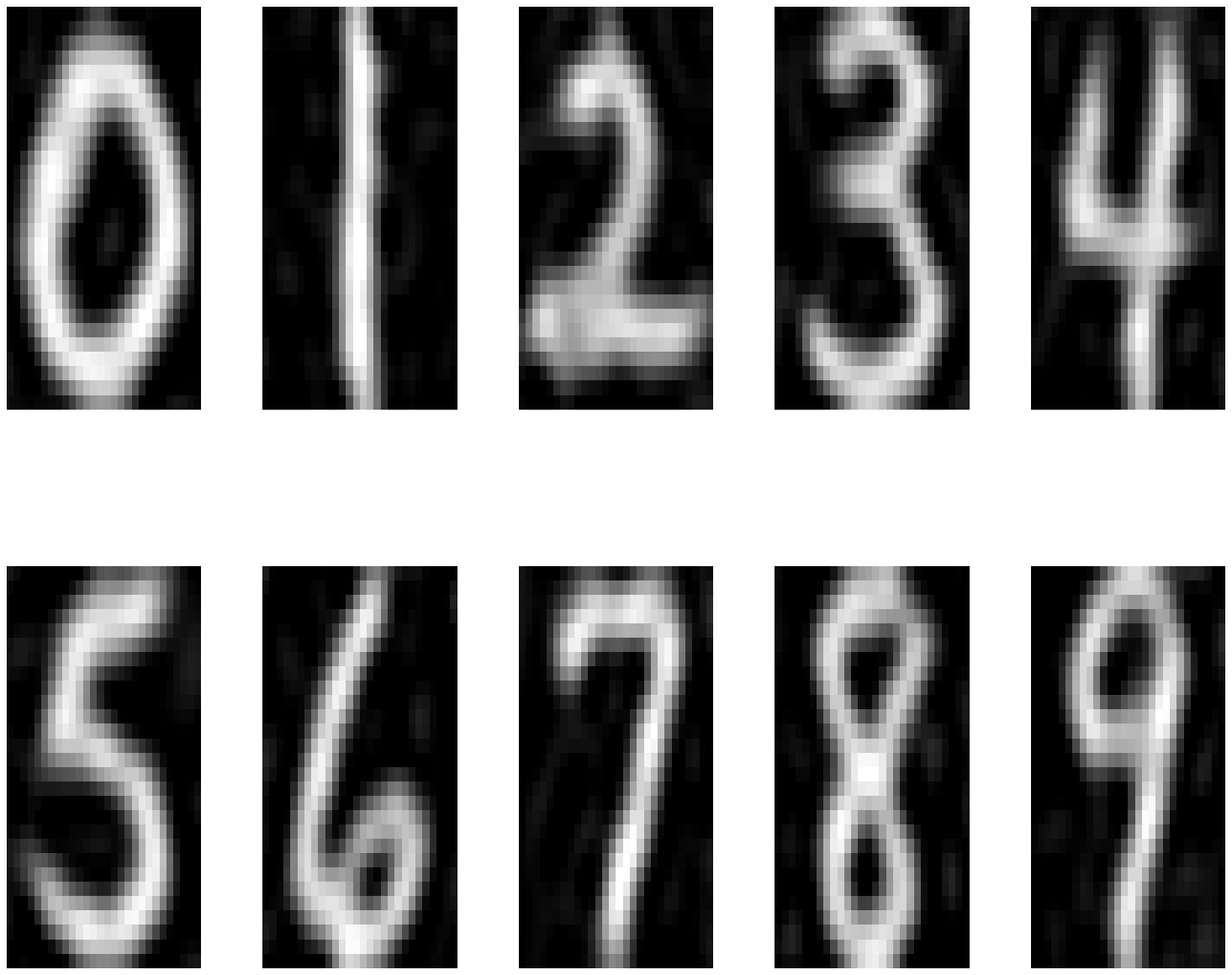}
}
\caption{Estimated prototypes in a noisy setting  $\si^2=1$: Left: with the mode
  approximation algorithm. Right: with the SAEM-MCMC coupling procedure.}
\label{fig-templatesSnr}
\end{figure}

\begin{figure}[htbp]
  \centerline{\epsfxsize=4.5cm
    \epsfbox{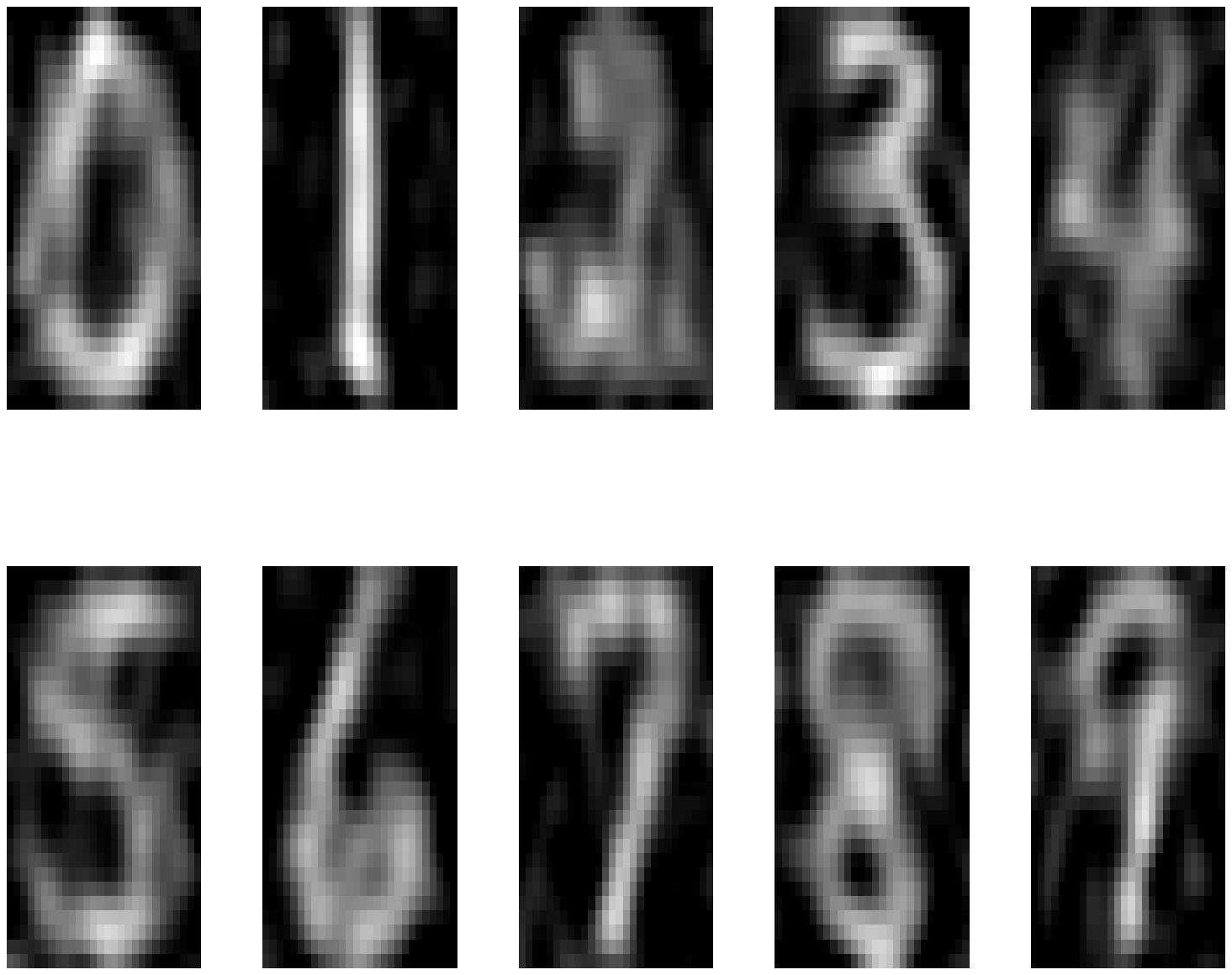} \hspace{0.5cm} \epsfxsize=4.5cm
    \epsfbox{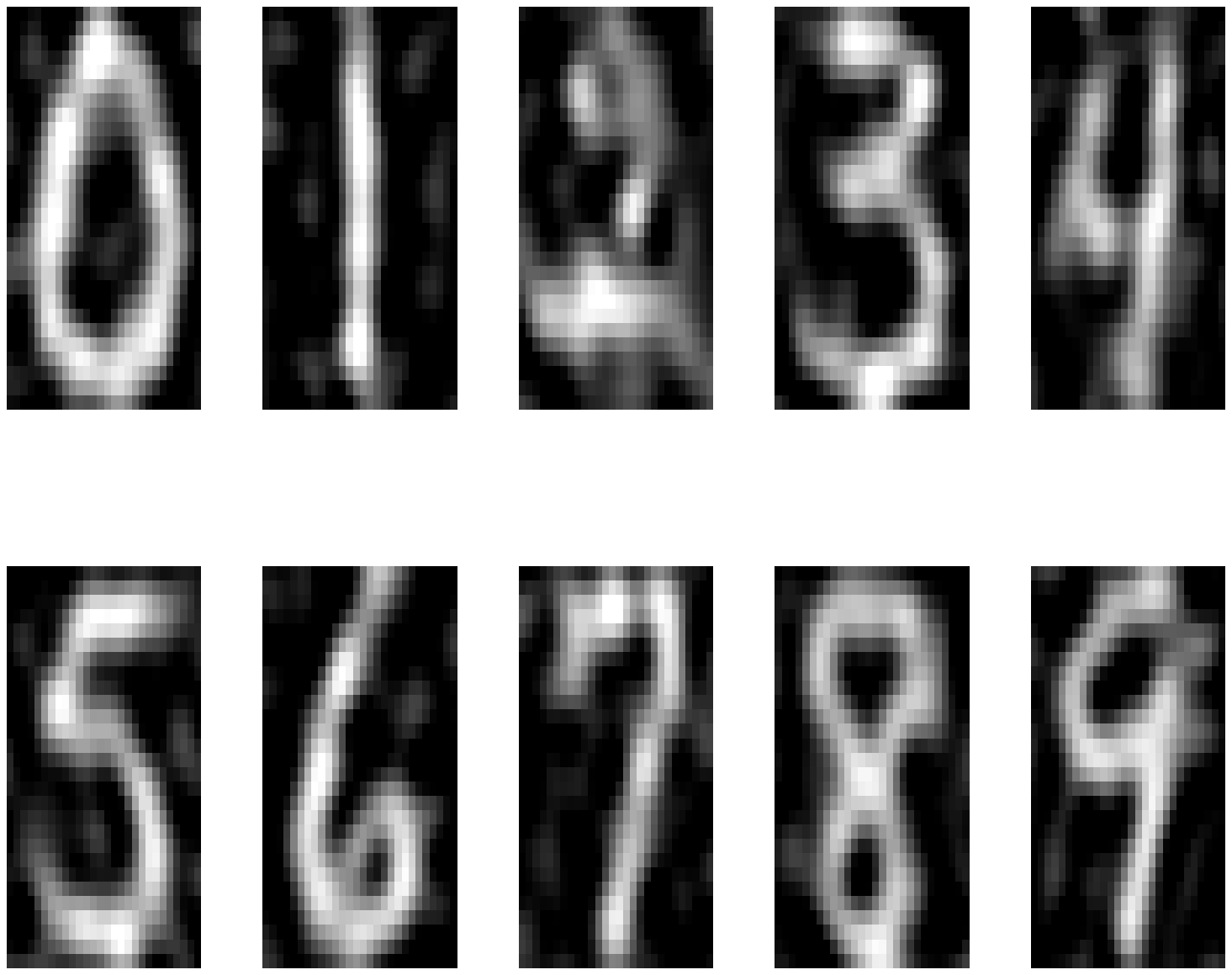}
}
\caption{Estimated prototypes in a noisy setting  $\si^2=2$: Left: with the mode
  approximation algorithm. Right: with the SAEM-MCMC coupling procedure.}
\label{fig-templatesSnr2}
\end{figure}

As shown in \cite{AAT}, in the presence of noise, the mode
approximation algorithm does not converge towards the MAP
estimator. In our setting, the consistency of the ``SAEM like''
algorithm has been
proved independently of the training set, and thus noisy images can also
be treated exactly the same way. These are the results
we present here.
Figure \ref{fig-trainingSnr}  shows two training examples per class
for noise variance values $\si^2=1$  and $\si^2=2$.
In Figures \ref{fig-templatesSnr} and \ref{fig-templatesSnr2}, we show the
estimated templates
for the noisy training set containing 20 images for both methods.
Even if the mode
approximation algorithm does not diverge, it cannot fit the template
for digits with a high variability. In contrast, the stochastic EM
 gives  acceptable contrasted templates which look like
those obtained in Figure \ref{fig-template}. This becomes more
significant as we increase the variance of the additive noise we
introduce in the training set. \\

%% YA The same or different?? %% Alain: Je ne comprends pas bien non plus
Concerning the choice of the hyper-parameters, it is not necessary to
change all of them. For the
photometric variance of the spline kernel, a small one could
create some non-smooth templates and a large kernel would smooth
the noise effect. However, we can keep the geometric hyper-parameters
unchanged.
We are presenting here  only  experiments which
seemed to provide a reasonable tradeoff between these effects.\\

% \begin{figure}[htbp]
%   \centerline{\epsfxsize=16cm \epsfbox{EvolSig35Snr1.eps}}
% \caption{Evolution of the noise variance along the iterations in the
% presence of noise.}
% \label{fig-sigmaSnr}
% \end{figure}

% We also check the estimated noise variance. Figure
% (\ref{fig-sigmaSnr}) shows the evolution of $\si^2$ for all the digits
% along the iterations. The estimated noise variance is below the true
% noise variance for all digits.
% Indeed, this can be explained by a tendency of the maximum likelihood
% estimator of the variance: it has been highlighted that it usually tends to
% underestimate this parameter. To get round this, some authors suggest
% to use the restricted maximum likelihood (see \cite{harville}) which has
% not been tested here. \\

\begin{figure}[htbp]\begin{center}
\includegraphics[width=8cm,height=9cm]{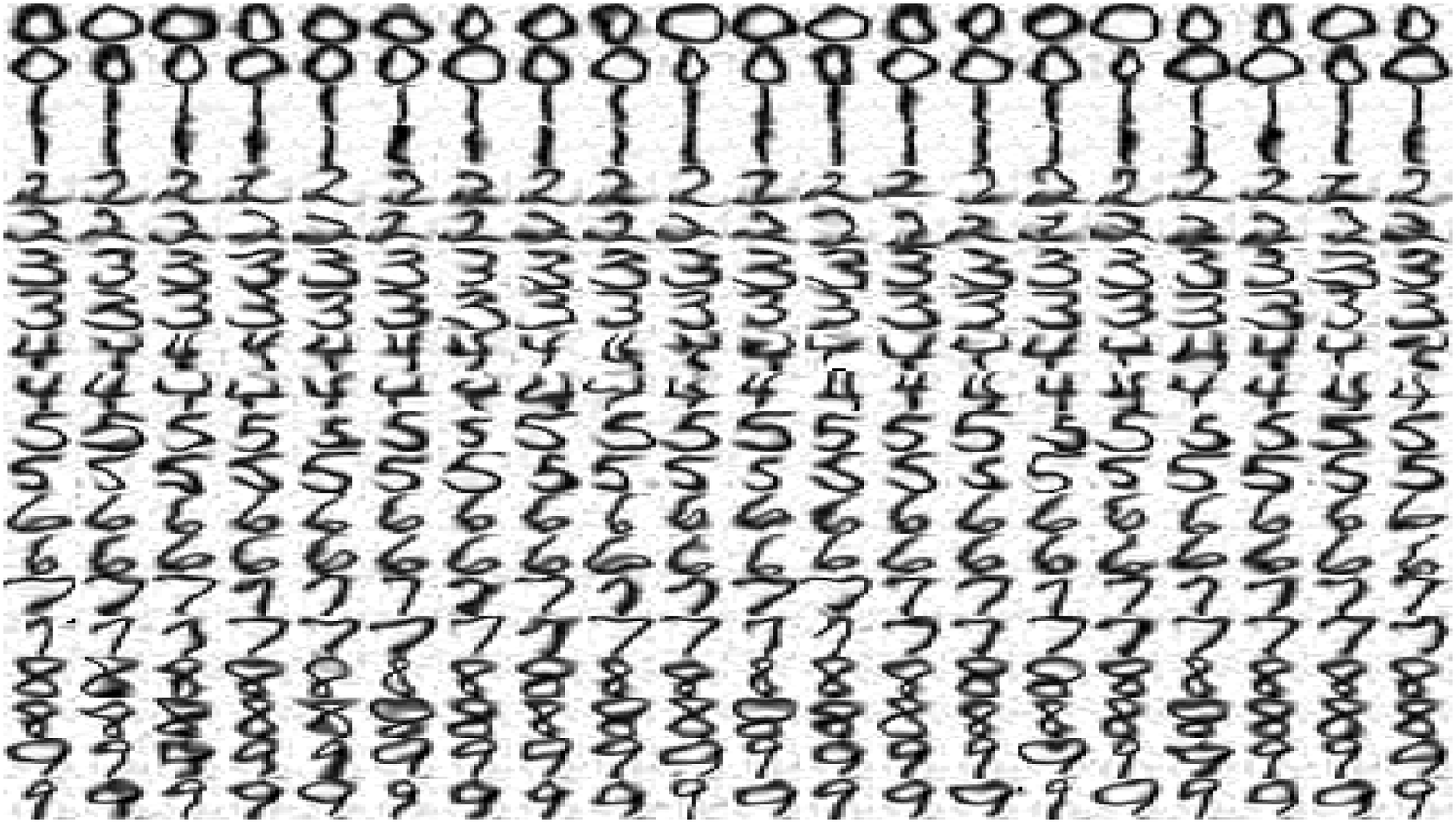}
\end{center}
\caption{Effect of the noise on the geometric parameter estimation~:
  40 synthetic examples per class generated with the
  parameters estimated from the noisy training set (additive noise
  variance of $1$, inverse video).}
\label{fig-sampleNoisy}
\end{figure}

The geometry is also well estimated despite the high level of noise in
the training set. Figure \ref{fig-sampleNoisy} shows some synthetic
examples, in which parameters are learnt from the training set with an
additive noise variance of one. The two lines
correspond to deformations and their symmetric deformation. This sample
looks like the synthetic samples learnt on non-noisy images even if
some examples are not relevant. However, the global behavior
has been learnt.\\

  The
algorithm manages to catch the photometry (a contrasted and smoothed
template), the geometry of the shapes and to ``separate'' the
additive noise.\\

The number of iterations needed to reach the convergence point in the
noisy setting is about twice that of the non-noisy case. The template
takes the longest time to converge and the estimate of $\sigma^2$ converges in a few
iterations.
%convergence
%of the template is the longest whereas the convergence of $\si^2$ takes
%relatively the same number of iterations.
In particular, the templates
obtained in the left panel of Figure \ref{fig-template} with only
$10$ images per training digit set are obtained with a heating period
of $25$ iterations and $5$ more steps with memory. The
templates  of Figure \ref{fig-templatesSnr}, right picture, require
$100$ to $125$
heating iterations in the $150$ global iterations. This is understandable
since the algorithm has to cope with variations due to the noise and
thus needs a longer time to fit the  model.

\section{Proof of Theorem \ref{th:condition} }\label{appendix}

% \subsection{One component case}
Here we demonstrate Theorem \ref{th:condition}, i.e. the
stochastic approximation sequence satisfies assumptions (\textbf{A1'})
(\textbf{ii}), (\textbf{iii}), (\textbf{iv}), (\textbf{A2}) and (\textbf{A3'}).

 We recall that in this
section, the parameter $\sigma^2$ is fixed so that
$\theta=(\alpha,\Gamma)$. The sufficient statistic vector $S$, the set
$\mathcal{S}$ as well as the explicit expression of $\hat{\theta}(s)$
have been  given in
Subsection \ref{ConvTheo}. As noted, $\hat{\theta}$ is
a smooth function of $\mathcal{S}$.

We will prove that these conditions  hold for any $p\geq 1$ and $a\in
]0,1[$.

\subsection{Proof of assumption (\textbf{A1'})}
$\quad$

We recall the functions $H,\ h$ and $\ w$ as in  \cite{DLM} defined as
follows:
\begin{eqnarray*}
  H_s(\bdbeta) & =&  S(\bdbeta) -s\,, \\
h(s)& = & \int_{\R^N} H_s(\bdbeta) \qpost(\bdbeta|\bdy, \hat\theta(s)) d\bdbeta\,, \\
w(s) & = & -l(\hat{\te}(s)) \, .
\end{eqnarray*}

As shown in \cite{DLM}, with these functions,   we satisfy
(\textbf{A1'(iii)}) and (\textbf{A1'(iv)}).

% Moreover, there exists an absorbing closed subset $\Sa$ of
% $\R^{n_s}$ such that

Moreover, since the interpolation kernel $K_p$ is
bounded, there exist  $A>0$ and  $B\in
\Symep$ such that for any $\bdbeta\in \mathbb{R}^N$,
we have
\begin{equation*}
\|S_1(\bdbeta)\|\leq A,\  0\leq S_2(\bdbeta)\leq B \text{ and }0\leq S_3(\bdbeta)
\,, \label{eq:2}
\end{equation*}
where, for any symmetric matrices $B$ and $B'$, we say
that $B\leq B'$ if $B'-B$ is a non-negative symmetric matrix.

We define the set $\Sa$ by
$$\Sa\triangleq\{\ S\in \mathcal{S}\ |\ \|S_1\|\leq A,\ 0\leq S_2\leq B\ \text{
  and } 0\leq
S_3\ \}\,.$$
Since the constraints are obviously convex and closed, we get that
  $\Sa$ is a closed convex subset of $\mathbb{R}^{n_s}$ such that
$$\Sa\subset \mathcal{S}\subset \mathbb{R}^{n_s}$$
and satisfying
\begin{equation*}
s+\rho H_s(\bdbeta)\in\Sa \ \text{ for any $\rho\in [0,1]$
any $s\in \Sa$ and any $\bdbeta\in \R^N$.} \label{eq:3}
\end{equation*}

We now focus on the first two points.
As $l$ and $\hat\te$ are continuous functions, we only need to prove that
$\mathcal{W}_M\cap \Sa$ is a bounded set
for a constant $M \in \R^*_+$ with:
\begin{equation*}
  \mathcal{W}_M = \{
s\in \mathcal{S} , \ w(s) \leq M
\} \,.
\end{equation*}
On $\Sa$, $s_1$ and $s_2$ are bounded; writing
$\hat\theta(s)=(\alpha(s),\Gamma(s))$, we deduce
from \eqref{PhotoUpdateClust} and from the boundedness of
%interpolation
%kernel
$K_p$ that $\alpha(s)$ is bounded on $\Sa$ and
$|\dy_i-K_p^{\beta_i} \alpha(s)|$ is uniformly bounded on $\beta_i\in
\mathbb{R}^{2k_g}$ and $s\in\Sa$. Hence (recall that $\sigma^2$ is
fixed here), there exists an $\eta>0$ such that
$\qcond(\bdy|\bdbeta,\hat\theta(s))\geq
\eta$ for any $s\in \Sa$ and $\bdbeta\in\mathbb{R}^N$. Thus,
\begin{equation*}
  \label{eq:4}
  w(s)\geq -\log\left(\int \qmiss(\bdbeta,\hat\theta(s))d\bdbeta\right)+\text{C}\geq
  -\log(\qpara(\hat\theta(s)))+\text{C}\geq -\log({\qpara}_{|_\Gamma} (\Gamma(s)))+\text{C} \, ,
\end{equation*}
where $\text{C}$ is a constant independent of
$s\in S_a$. Since
$$-\log({\qpara}_{|_\Gamma}(\Gamma_g))= \frac{a_g}{2}\left( \langle
\Gamma_g^{-1},\Sigma_g\rangle_F + \log|\Gamma_g| \right)\geq \frac{a_g}{2}\log
|\Gamma_g|$$
and  $$\lim_{\|s\|\to +\infty, s\in\Sa} \log(|\Gamma_g(s)|)=\lim_{\|s\|\to +\infty, s\in\Sa}
\log(|(s_3+a_g\Sigma_g)/(n+a_g)|)= +\infty,$$
%as $|s|\to +\infty$, $s\in \Sa$
 we deduce that
$$\lim_{\|s\|\to +\infty, s\in\Sa}w(s)=+\infty\,.$$
Since $w$ is continuous and $\Sa$ is closed, this proves  (\textbf{A1'(ii)}).\\

\subsection{Proof of assumption (\textbf{A2})}
$\quad$

We prove a classical sufficient condition (\textbf{DRI1}), used in
\cite{andrieumoulinespriouret}  which will imply (\textbf{A2}) under
the condition that $H_s$ is dominated by $V$ for any $s\in \Kapa$.
\begin{description}
\item[(DRI1)]
For any $s\in \mathcal{S}$, $\ntrans_{\hat{\theta}(s)}$ is $\phi-$irreducible and
aperiodic. In addition there exist %  a small set $\texttt{C}$ (defined
% below),
a function $V:\R^{N} \to
[1,\infty[$ % and  constants  $0\leq b\leq 1$,
and $p\geq 2$ such that for
any compact subset $\Kapa\subset \mathcal{S}$,
 there exist an integer $m$ and constants $0<\lambda<1$, $B>0$, $\kappa>0$,
  $\delta>0$, a subset $\texttt{C} $ of $\R^N$ and a probability measure
  $\nu$ such that
\begin{eqnarray}
\label{6.1}\sup\limits_{s\in\Kapa} \ntrans_{\hat{\theta}(s)} ^mV^p(\bdbeta)& \leq &\lambda V^p(\bdbeta) +
 B \mathds{1}_{\texttt{C}}(\bdbeta) \,, \\ \label{6.2}
 \sup\limits_{s\in\Kapa} \ntrans_{\hat{\theta}(s)} V^p(\bdbeta)& \leq& \kappa V^p(\bdbeta)
 \ \  \forall \bdbeta\in \R^{N}  \,, \\ \label{6.3}
 \inf\limits_{s\in\Kapa} \ntrans_{\hat{\theta}(s)} ^m (\bdbeta,A) &\geq& \delta \nu(A) \ \
 \forall \bdbeta\in \texttt{C}, \forall A \in \mathcal{B}(\R^{N}) \, .
 \end{eqnarray}
\end{description}
\begin{rem}
  Note that condition \eqref{6.3} is equivalent to the existence of a
  small set $\texttt{C} $ (defined below) which only depends on  $\Kapa$.
\end{rem}
\begin{notation}
Let
$(e_j)_{1\leq j\leq N}$ be the canonical basis of $\mathbb{R}^N$. For any
$1\leq j\leq N$, let
$E_{\theta,j} \triangleq  \{\ \bdbeta\in\mathbb{R}^N\ |\ \langle
\bdbeta,e_j\rangle_\theta=0\}$ be the orthogonal space of $\text{Span}\{e_j\}$
and $p_{\theta,j}$ be the orthogonal projection on $E_{\theta,j}$ i.e.
$$p_{\theta,j}(\bdbeta) \triangleq
\bdbeta-\frac{\langle \bdbeta,e_j\rangle_\theta}{\|e_j\|^2_\theta}e_j\,,$$
where $\langle \bdbeta,\bdbeta'\rangle_\theta=\sum_{i=1}^n
\beta_i^t\Gamma_g^{-1}\beta'_i$ for $\theta=(\alpha,\Gamma_g)$ (i.e. the natural dot product
associated with the covariance matrix $\Gamma_g$) and $\|.\|_\theta$ the
corresponding norm.

  We denote for any $1\leq j\leq N$ and $\theta\in\Theta$ by
  $\ntrans_{\theta,j}$ the Markov kernel on $\mathbb{R}^N$ \eqref{eq:kernelj}
associated with the
  Metropolis-Hastings step of the $j$-th Gibbs sampler step on $\bdbeta$.
We have $\ntrans_\theta=\ntrans_{\theta,N}\circ\cdots\circ \ntrans_{\theta,1}$.
\end{notation}

We first recall the definition of a small set:
\begin{Def} (cf. \cite{meyntweedie})
 A set $\mathcal{E} \in \mathcal{B}(\mathcal{X})$ is called a
 \textbf{small set} for the kernel $\ntrans$
 if there exist an $m>0$, and a non trivial measure $\nu_m$ on $
 \mathcal{B}(\mathcal{X})$, such that for all $\bdbeta\in \mathcal{E}$,  $B\in
 \mathcal{B}(\mathcal{X})$,
 \begin{equation}
   \label{eq:smallset}
   \ntrans^m(\bdbeta,B) \geq \nu_m(B).
 \end{equation}
When (\ref{eq:smallset}) holds, we say that $\mathcal{E}$ is $\nu_m$-small.
\end{Def}

\vspace{0.3cm}

We now prove the following lemma which give the existence of the small
set $\texttt{C}$ in \textbf{(DRI1)}:

\vspace{2mm}
\begin{lemma}
 Let $\E$ be a compact subset of $\mathbb{R}^N$ and $\Kapa$ a compact
 subset of $\mathcal{S}$. Then $\E$ is a small
set of $\R^N$ for $\ntrans_{\hat\te(s)} $ for any $s\in\Kapa$.
\end{lemma}
\begin{proof} First note that there exists an $a_c>0$ such that
for any $\theta\in \Te$, any $\bdbeta\in\mathbb{R}^N$ and any
$b\in\mathbb{R}$, the acceptance rate
  $r_{j}(\bdbeta^j,b; \bdbetamj,\theta )$  is uniformly
 bounded below by $a_c$ so that for any $1\leq j\leq N$ and any non-negative
  function $f$,
$$\ntrans_{\theta,j}f(\bdbeta)\geq
a_c \int_\mathbb{R}f(\bdbetamj+be_j)\qj(b|\bdbetamj,\theta)db= a_c
\int_\mathbb{R}f(p_{\theta,j}(\bdbeta)+ze_j/\|e_j\|_\theta)g_{0,1}(z)dz \,,$$
where $g_{0,1}$ is the density of the standard Gaussian distribution
$\mathcal{N}(0,1)$.

By induction, we have
\begin{equation}
\ntrans_{\theta}f(\bdbeta)\geq
a_c^N
\int_{\mathbb{R}^N}f\left(p_{\theta,N,1}(\bdbeta)+\sum_{j=1}^Nz_jp_{\theta,N,j+1}(e_j)/
  \|e_j\|_\theta\right)
\prod_{j=1}^Ng_{0,1}(z_j)dz_j\,, \label{eq:6}
\end{equation}
where $p_{\theta,q,r}=p_{\theta,r}\circ p_{\theta,r-1} \circ  \cdots\circ
  p_{\theta,q}$ for any integers $q\leq r$ and
  $p_{\theta,N,N+1}=\text{Id}$.

 Let $A_\theta\in \mathcal{L}(\mathbb{R}^N)$ be the linear
mapping on $\mathbb{R}^N$ defined by $$A_\theta
z=\sum_{j=1}^Nz_jp_{\theta,N,j+1}(e_j)/\|e_j\|_\theta\,.$$ One easily
checks that for any $1\leq k\leq N$, $\text{Span}\{\
p_{\theta,N,j+1}(e_j),\ k\leq j\leq N\}=\text{Span}\{ e_j\ |\ k\leq j\leq N\}$
so that $A_\theta$ is an invertible mapping. By a change of variable,
we get
$$\int_{\mathbb{R}^N}f\left(p_{\theta,N,1}(\bdbeta)+A_\theta
  z_1^N\right)\prod_{j=1}^Ng_{0,1}(z_j)dz_j=\int_{\mathbb{R}^N}f(u)g_{p_{\theta,N,1}(\bdbeta),A_\theta
  A_\theta^t}(u)du\,, $$
where $g_{\mu,\Sigma}$ stands for the density of the normal law
  $\mathcal{N}(\mu,\Sigma)$. Since
  $\theta\to A_\theta$ is smooth on the set of invertible mappings in
  $\theta$, we deduce that there exist two constants $c_\Kapa>0$ and
  $C_\Kapa>0$ such
  that $c_\Kapa\text{Id}\leq
  A_\theta A_\theta^t\leq \text{Id}/c_\Kapa$ and $g_{p_{\theta,N,1}(\bdbeta),A_\theta
  A_\theta^t}(u)\geq C_\Kapa g_{p_{\theta,N,1}(\bdbeta),\text{Id}/c_\Kapa}(u)$ uniformly
  for $\theta=\hat\theta(s)$ with $s\in\mathcal{K}$. Assuming that
  $\bdbeta\in\E$, since $\theta\to p_{\theta,N,1}$ is smooth and $\E$ is
  compact, we have $\sup\limits_{\bdbeta\in\E,\theta=\hat\theta(s),\ s\in
  \mathcal{K}} \|p_{\theta,N,1}(\bdbeta)\|<\infty$. Therefore, there exist $C_\Kapa'>0$
  and $c_\Kapa'>0$ such that for any
  $(u,\bdbeta)\in\mathbb{R}^N× \E$ and any $\theta=\hat\theta(s),\
  s\in\mathcal{K}$

\begin{equation}
g_{p_{\theta,N,1}(\bdbeta),A_\theta
  A_\theta^t}(u)\geq C_\Kapa'g_{0,\text{Id}/c_\Kapa'}(u) \, .\\  \label{eq:5}
\end{equation}
  Using \eqref{eq:6} and \eqref{eq:5}, we deduce that for any $A$, for
  any $s \in \Kapa$ and $\theta=\hat\theta(s)$,
$$\ntrans_\theta(\bdbeta,A)\geq C_\Kapa'a_c^N \nu_\Kapa(A)\,,$$
with $\nu_\Kapa$ equals to the density of the normal law
$\mathcal{N}(0,\text{Id}/c_\Kapa')$.

This yields the
existence of the small set as well as equation
(\ref{6.3}).
\end{proof}\\

This property also implies the $\phi$-irreducibility of the Markov
chain $(\bdbeta_k)_k$ and its
aperiodicity  (cf. \cite{meyntweedie} p121).

We set $V:\mathbb{R}^N\to [1,+\infty [$ as the following function
\begin{equation}\label{defv}
V(\bdbeta) = 1 + \|\bdbeta\|^2\,. %\sum_{i=1}^n|\bdbeta_i|^2.
\end{equation}
We, in fact, have the
following property~: $\exists \ C_\Kapa >0$ such that~: $\forall \bdbeta \in \R^N$,
\begin{equation*}
 \sup\limits_{s\in \Kapa}  \| H_s (\bdbeta) \|\leq C_\Kapa \ V(\bdbeta)\,.
\end{equation*}
This condition is required for the implication of (\textbf{A2}) by
(\textbf{DRI1}).
\\

We now prove condition (\ref{6.2}).

Let $\Kapa$ be a compact subset of $\mathcal{S}$ and $p\geq 1$.
For any $1\leq j\leq N$, any $s\in \Kapa$ and $\theta=\hat\theta(s)$, we have
$$\ntrans_{\theta,j}V^p(\bdbeta)\leq
V^p(\bdbeta)+\int_{\mathbb{R}}V^p(p_{\theta,j}(\bdbeta)+ze_j/\|e_j\|_\theta)g_{0,1}(z)dz\,.$$
Since $V(\bdbeta+h)\leq 2(V(\bdbeta)+V(h))$  for any
$\bdbeta,h\in\mathbb{R}^N$ and since
 there exist two constants $c_\Kapa>0$ and
  $C_\Kapa>0$ such
  that for any $\bdbeta\in \R^N$, $\theta\in \hat\theta(\Kapa)$,
$\|p_{\theta,j}(\bdbeta)\|\leq C_\Kapa\|\bdbeta\|$ and $\|e_j\|_\theta \geq 1/c_\Kapa$, we have
$$\int_\mathbb{R}V^p(p_{\theta,j}(\bdbeta)+ze_j/\|e_j\|_\theta)g_{0,1}(z)dz
\leq 2^{p}C_\Kapa^pV^p(\bdbeta)\int_\mathbb{R}(1+V(c_\Kapa ze_j))^pg_{0,1}(z)dz \, .$$
We deduce that there exists an $C_\Kapa'>0$ such that for any
$\bdbeta\in\mathbb{R}^N$
\begin{equation*}
\sup_{\theta=\hat\theta(s),s\in\mathcal{K}}\ntrans_{\theta,j}V^p(\bdbeta)\leq
C_\Kapa'V^p(\bdbeta)\, .\label{eq:11}
\end{equation*}
Then, by composition $\ntrans_\theta V^p(\bdbeta)\leq C_\Kapa'^N V^p(\bdbeta)$
and (\ref{6.2}) holds for any $p\geq 1$.\\

Now consider the Drift condition (\ref{6.1}).

 To  prove this
inequality, we prove the same inequality for a subsidiary function
$V_\theta$ which depends on the parameters $\theta$ and then we deduce the
result for $V$.\\
So let us define for any $\theta=(\alpha,\Gamma_g)$ the function $V_\theta(\bdbeta)\triangleq
1+\|\bdbeta\|_\theta^2$.
% where $|\beta|_\theta^2\triangleq \langle\beta,\beta\rangle_\theta=\beta^t \Gamma_g^{-1}\beta$ is the natural dot
% product induced by the covariance operator $\Gamma_g$.

\begin{lemma}
\label{lem:a1}
 Let $K$ be a compact subset of $\Theta$. For any $p\geq 1$, there
 exist an $0\leq \rho_K<1$ and
  an $C_K>0$ such that for any $\theta\in K$, any $\bdbeta\in\mathbb{R}^N$ we
  have
  $$\ntrans_\theta V_\theta^p(\bdbeta)\leq \rho_K V_\theta^p(\bdbeta)+C_K\,.$$
\end{lemma}
\begin{proof}
The proposal distribution for $\ntrans_{\theta,j}$ is given by $q(\bdbeta \ |\
\bdbetamj,\dy,\theta)\stackrel{\text{law}}{=}p_{\theta,j} (\bdbeta)+z\frac{e_j}{\|e_j\|_{\te}}$ where $z\sim \mathcal{N}(0,1)$.  % Since we
% easily check that the acceptance rate $a_{\te,\bdbeta} $ is uniformly
% bounded from below
% by a positive number $a_c>0$,
Then, there exists  $C_K$ such
that for any $\bdbeta\in\mathbb{R}^N$ and  any
measurable set $A\in\mathcal{B}(\mathbb{R}^N)$
$$\ntrans_{\theta,j}(\bdbeta,A)=(1-a_{\theta,\bdbeta})\mathds{1}_A(\bdbeta)+a_{\theta,\bdbeta}\int_\mathbb{R}
\mathds{1}_A\left(p_{\theta,j}(\bdbeta)+z\frac{e_j}{\|e_j\|_{\te}}\right) g_{0,1}(z)dz
\,,$$
where $a_{\theta,\bdbeta}\geq a_c$ ($a_c$ is a lower bound for the acceptance
rate),
%$\gamma_{\theta}\leq C_K \gamma_K$ and
%$\gamma_K$ equals the density of the normal law $\mathcal{N}(0,\sup_{\theta\in
%  K}\|e_j\|_\theta^{-2})$.

Since $\langle
p_{\theta,j}(\bdbeta),e_j\rangle_\theta=0$, we get
$V_\theta\left(p_{\theta,j}(\bdbeta)+z\frac{e_j}{\|e_j\|_{\te}}\right)=V_\theta(p_{\theta,j}(\bdbeta))+z^2$
and
\begin{multline*}
  \ntrans_{\theta,j} V_\te ^p(\bdbeta)=
  (1-a_{\theta,\bdbeta})V^p_\theta(\bdbeta)+a_{\theta,\bdbeta}
  \int_\mathbb{R}\left(V_\theta(p_{\theta,j}(\bdbeta)) + z^2\right)^p
  g_{0,1}(z)dz\label{eq:1}\\
\leq
  (1-a_{\theta,\bdbeta})V^p_\theta(\bdbeta)+a_{\theta,\bdbeta}
  \left(V^p_\theta(p_{\theta,j}(\bdbeta))+ C_K V^{p-1}_\theta(p_{\theta,j}(\bdbeta))
    \int_\mathbb{R}(1+z^2 )^{p} g_{0,1}(z)dz\right)\\
\leq
  (1-a_{\theta,\bdbeta})V^p_\theta(\bdbeta)+a_{\theta,\bdbeta}V^p_\theta(p_{\theta,j}
  (\bdbeta))+C'_KV^{p-1}_\theta (p_{\theta,j}(\bdbeta)) \,.
\end{multline*}
We have used in the last inequality the fact that a Gaussian variable
has bounded
moments of any order. Since $a_{\theta,\bdbeta}\geq a_c$ and
$\|p_{\theta,j}(\bdbeta)\|_\theta\leq \|\bdbeta\|_\theta$ ($p_{\theta,j}$ is an
orthonormal projection for the dot product $\langle \cdot,\cdot
\rangle_\theta$), we get that $\forall \, \eta>0$,
$\exists \, C_{K,\eta}$ such that  $ \forall \, \bdbeta\in\mathbb{R}^N$ and $ \forall \, \theta\in K$
$$  \ntrans_{\theta,j} V^p_\te(\bdbeta)\leq
(1-a_c)V^p_\theta(\bdbeta)+(a_c+\eta)V_\theta^p(p_{\theta,j}(\bdbeta))+C_{K,\eta}\,.$$
By induction, we show that
$$\ntrans_\theta V^p_\te(\bdbeta)\leq
\sum_{u\in\{0,1\}^N}\prod_{j=1}^N(1-a_c)^{1-u_j}(a_c+\eta)^{u_j}
V^p_\te(p_{\theta,u}(\bdbeta))+\frac{C_{K,\eta} }{\eta } ( (1+\eta)^{N+1} -1) \,,$$
where $p_{\theta,u}=((1-u_N)\text{Id}+u_Np_{\theta,N})\circ\cdots\circ
((1-u_1)\text{Id}+u_1p_{\theta,1})$. Let
$p_\theta=p_{\theta,N}\circ\cdots\circ
p_{\theta,1}$ and note that $p_{\theta,j}$ is  contracting so that

$$\ntrans_\theta V^p_\te(\bdbeta)\leq b_{c,\eta}V^p_\theta(\bdbeta)+(a_c+\eta)^N
V^p_\theta(p_{\theta}(\bdbeta))+\frac{C_{K,\eta}  }{\eta } ( (1+\eta)^{N+1})\,,$$
for $b_{c,\eta}=\left(\sum_{u\in\{0,1\}^N,\ u\neq
  \mathbf{1}}\prod_{j=1}^N(1-a_c)^{1-u_j}(a_c+\eta)^{u_j}\right)$.

To
end the proof, we need to check that $p_\theta$ is strictly
contracting uniformly on $K$. Indeed,
$\|p_{\theta}(\bdbeta)\|_\theta=\|\bdbeta\|_\theta$  implies
that $p_{\theta,j}(\bdbeta)=\bdbeta$ for any $1\leq j\leq N$. This yields $\langle
\bdbeta,e_j\rangle_\theta=0$ and thus $\bdbeta=0$ since $(e_j)_{1\leq j\leq N}$ is a
basis. Using the continuity of the norm of $p_\theta$ in $\theta$ and
the compactness of $K$, we deduce that there exists  $0<\rho_K<1$ such
that $\|p_{\theta}(\bdbeta)\|_\theta\leq \rho_K\|\bdbeta\|_\theta$ for any $\bdbeta$
and $\theta\in K$. Changing $\rho_K$ for $1>\rho'_K>\rho_K$ we get
$(1+\rho_K^2\|\bdbeta\|_\theta^2)^p\leq {\rho'}_K^{2p}(1+\|\bdbeta\|_\theta^2)^p+C''_K$ for
some uniform constant $C''_K$. Therefore,
$$\ntrans_\theta V^p_\te(\bdbeta)\leq b_{c,\eta}V^p_\theta(\bdbeta)+{\rho'}_K^{2p}(a_c+\eta)^N
V^p_\theta(\bdbeta)+C''_{K,\eta}.$$
Since we have  $ \inf_{\eta>0}b_{c,\eta}+{\rho'}_K^{2p}(a_c+\eta)^N<1$ the result
is immediate.
\end{proof}

Next, we prove the expected inequality for the function $V$.
\begin{lemma}
\label{lem:2}
  For any compact set $K\subset \Theta$, any $p\geq 1$, there exist
$0<\rho_K< 1$, $C_K>0$ and $m_0$ such that $\forall m \geq m_0$ , $\forall \te \in K $,
$ \forall \ \bdbeta\in \R^N$
$$\ntrans_\theta^mV^p(\bdbeta)\leq \rho_K V^p(\bdbeta)+C_K\,.$$
\end{lemma}
\begin{proof}
  Indeed, there exist  $0\leq c_1\leq c_2$ such that $c_1V(\bdbeta)\leq
  V_\theta(\bdbeta)\leq c_2 V(\bdbeta)$ for any $(\bdbeta,\theta)\in\mathbb{R}^N×
  K$. Then, using the previous lemma, we have $\ntrans^m_\theta V^p(\bdbeta)\leq
  c_1^{-p}\ntrans^m_\theta V_\theta^p(\bdbeta)\leq c_1^{-p}(\rho_K^m
  V^p_\theta(\bdbeta)+C_K/(1-\rho_K))\leq (c_2/c_1)^p(\rho_K^m
  V^p(\bdbeta)+C_K/(1-\rho_K))$. Choosing $m$ large enough for $
  (c_2/c_1)^p\rho_K^m<1$ gives the result.
\end{proof}

This finishes the proof of (\ref{6.1}) and at the same time of
(\textbf{A2}).

\subsection{Proof of assumption (\textbf{A3'})}
$\quad$

 The geometric ergodicity of the Markov chain, implied by the Drift
 condition (\ref{6.1}), ensures the existence of  a solution
 of the Poisson
  equation (cf. \cite{meyntweedie}):
  \begin{equation*}
  g_{\hat\theta(s)} (\bdbeta)= \sum\limits_{k\geq 0} (\ntrans^k_{\hat{\te}(s)}
  H_s(\bdbeta) - h(s)).
  \end{equation*}

We first prove condition (\textbf{A3'(i)}).

Since $H_s(\bdbeta)=\tS(\bdbeta) -s $ with $\tS(\bdbeta) $ at most
quadratic in $\bdbeta$, the choice of $V$ directly ensures
(\ref{A2i1}). \\

%Considering (\ref{A2i2}):
Due to the result presented in
\cite{doucmoulinesrosenthal}, there exist upper bounds for the
convergence rates and the
constants involved in the quantification of the geometrical ergodicity
of all the chains indexed by $s\in \Kapa$
which only depend on $m,\lambda, B,\delta $. Therefore, these constants only
depend on the fixed compact set $\Kapa$. This yields the uniform
ergodicity of the family of Markov chains on $\Kapa$.
So
there exist constants $0<\gamma_\Kapa<1$ and  $C_\Kapa>0$  such that
\begin{eqnarray*}
  \|g_{\hat\theta(s)}\|_V = \| \sum\limits_{k\geq 0} (\ntrans^k_{\hat{\te}(s)}
  H_s(\bdbeta) - h(s))\|_V \leq  \sum\limits_{k\geq 0} C_\Kapa \gamma_\Kapa^k \|H_s\|_V <\infty \ .
\end{eqnarray*}
Thus $\forall s \in \Kapa , \ g_{\hat\theta(s)}$ belongs to $\mathcal{L}_V = \{ g~ :
\R^N \to \R
, \| g\| _V < \infty \} $.

Repeating the same calculation as above, it is immediate that
$\ntrans_{\hat{\te}(s)}
g_{\hat\theta(s)}$ belongs to $\mathcal{L}_V$ too. This ends the proof of
(\textbf{A3'(i)}).\\

We now move to the Hölder  condition  (\textbf{A3'(ii)}). We will use
the following lemmas which state Lipschitz conditions on the
transition kernel and its iterates:
\begin{lemma}\label{lem:Holder}
Let $\mathcal{K}$ be a compact subset of $\mathcal{S}$. There exists a
constant $C_{\mathcal{K}}$ such that
for any $p\geq 1$ and any function $f \in \mathcal{L}_{V^p}$, $\forall (s,s')
\in \mathcal{K}^2$ we have~:

  \begin{eqnarray*}
  \| \ntrans_{\hat\te(s)} f -
  \ntrans_{\hat\te(s')} f \|_{V^{p+1/2}}  \leq
C_\mathcal{K}  %\\
\|f\|_{V^{p}} \  \|s-s'\| \ .
\end{eqnarray*}
\end{lemma}
\begin{proof}
 For any $1\leq j\leq N$
and $f\in \mathcal{L}_{V^p}$, we have
\begin{equation*}
\ntrans_{\theta,j}f(\bdbeta)=(1-r_{j}(\bdbeta,\theta))f(\bdbeta)+
\int_\mathbb{R}f(\betajb)r_{j}(\bdbeta^j,b;\bdbetamj,\theta)\qj(b|\bdbetamj,\theta)db
\,,
\label{eq:9}
\end{equation*}
where $r_j(\bdbeta,\theta)=\int_\mathbb{R}r_j(\bdbeta^j, b; \bdbetamj,
\theta)\qj(b|\bdbetamj,\theta)db$ is the average acceptance rate.

Let $s$ and $s'$ be two points in $\mathcal{K}$ and
$s(\epsilon)=(1-\epsilon)s+\epsilon s'$ for $\epsilon\in [0,1]$
be a linear interpolation between $s$ and $s'$ (since $\mathcal{S}$ is
convex,
we can assume that $\mathcal{K}$ is a
convex set so that $s(\epsilon)\in\mathcal{K}$ for any $\epsilon\in
[0,1]$). We denote also by $\theta(\epsilon)\triangleq
\hat{\theta}(s(\epsilon))$ the associated path in $\Theta$
which is a  continuously differentiable  function.
To study the difference $\|(\ntrans_{\theta(1),j}-\ntrans_{\theta(0),j})f(\bdbeta)\|$,
 introduce $\ntrans_{\theta,j}^1f(\bdbeta)\triangleq (1-r_{j}(\bdbeta,\theta))f(\bdbeta) $ and
$\ntrans_{\theta,j}^2f(\bdbeta)\triangleq
\int_\mathbb{R}f(\betajb)r_j(\bdbeta^j, b; \bdbetamj, \theta) \qj(b|\bdbetamj,\theta)db$.
We start with
the difference
$\|(\ntrans_{\theta(1),j}^2-\ntrans_{\theta(0),j}^2)f(\bdbeta)\|$.
First note that under the conditional law $\qj(b|\bdbetamj,\theta)$,
$b\sim\mathcal{N}(b_{\theta,j}(\bdbeta),1/\|e_j\|_\theta^2)$ where
\begin{equation*}
\label{a1}
b_{\theta,j}(\bdbeta)\triangleq e_j^tp_{\theta,j}(\bdbeta)=e_j^t\bdbeta-\langle \bdbeta,e_j\rangle_\theta/\|e_j\|_\theta^2
\end{equation*}
 is the $j$-th coordinate of $p_{\theta,j}(\bdbeta)$. We have
$$ \ntrans_{\theta,j}^2f(\bdbeta)=
\int_\mathbb{R}f(\bdbeta_{0\to j}+be_j)r_{j}(\bdbeta^j,b; \bdbetamj,\theta)
\exp\left(-\frac{(b-b_{\theta,j}(\bdbeta))^2\|e_j\|^2_\theta}{2}\right)
\frac{\|e_j\|_\theta}{\sqrt{2\pi}}db\,.$$

 Since
 $r_j(\bdbeta^j, b; \bdbetamj, \theta)=\tilde{r}_{j}(\bdbeta^j,b;\bdbetamj,\theta) \land 1$ where
 $\tilde{r}_{j}(\bdbeta^j,b;\bdbetamj,\theta) \triangleq
 \frac{\qobs(\bdy|\betajb,\theta)}{\qobs(\bdy|\bdbeta,\theta)}$ is a smooth
 function in $\theta$, we have
$$\|(\ntrans_{\theta(1),j}^2-\ntrans_{\theta(0),j}^2)f(\bdbeta)\|\leq
\int_0^1\int_\mathbb{R}\|f(\bdbeta_{0\to j}+be_j) \| \left|\frac{d}{d\epsilon}\left(r_j(\bdbeta^j,
    b; \bdbetamj, \theta)
    \exp(-\frac{(b-b_{\theta,j}(\bdbeta))^2\|e_j\|^2_\theta}{2})\frac{\|e_j\|_\theta}{\sqrt{2\pi}}\right)\right|db\,.$$
However, one easily checks that there exists a constant $C_\mathcal{K}$ such
that for any $s,s'\in \mathcal{K}$, $\epsilon$, $j$ and $\bdbeta$ (with $\theta=\theta(\epsilon)$):
\begin{multline} \label{dexpdeps}
\left| \frac{d}{d\epsilon}\exp\left(-\frac{(b-b_{\theta,j}(\bdbeta))^2\|e_j\|^2_\theta}{2}
\right)\frac{\|e_j\|_\theta}{\sqrt{2\pi}}\right| \\
\leq C_\mathcal{K} (1+|b-b_{\theta,j}(\bdbeta)|)^2
\exp\left(-\frac{(b-b_{\theta,j}(\bdbeta))^2\|e_j\|^2_\theta}{2}\right)
\frac{\|e_j\|_\theta}{\sqrt{2\pi}}
\left(\left|\frac{d}{d\epsilon}b_{\theta,j}(\bdbeta)\right|
+\left|\frac{d}{d\epsilon}\|e_j\|_\theta\right|\right) \ .
\end{multline}
Since $\frac{d}{d\epsilon}\|e_j\|_\theta=
\frac{1}{2\|e_j\|_\theta}e_j^t\frac{d}{d\epsilon}
\Gamma_{\theta}^{-1}e_j$,
$\frac{d}{d\epsilon}\Gamma_{\theta}^{-1}
=-\Gamma_{\theta}^{-1}\frac{d}{d\epsilon}\Gamma_{\theta}
\Gamma_{\theta}^{-1}$ and
$\frac{d}{d\epsilon}\Gamma_{\theta}
=\frac{s'_3-s_3}{n+a_g}$ (see  \eqref{PhotoUpdateClust}), we deduce that there
exists another constant $C_{\mathcal{K}}$ such that
\begin{equation}\label{a0}
\left|\frac{d}{d\epsilon}\|e_j\|_\theta\right|\leq C_{\mathcal{K}}\|s'-s\|\,.
\end{equation}
Similarly, updating the
constant $C_\mathcal{K}$, we have\footnote{\label{note}Note that the
  extra factor $(1+\|\bdbeta\|)$ appearing in the RHS of \ref{a2}
  compared to the RHS of \ref{a0} alleviate the need to show
  the usual Lipschitz condition $ \| \ntrans_{\hat\te(s)} f -
  \ntrans_{\hat\te(s')} f \|_{V^{p'}}  \leq
C_\mathcal{K}
\|f\|_{V^{q}} \  \|s-s'\| $ with $q=p$. Weaker Lipschitz conditions as
conditions \textbf{A3' (ii)} of Theorem \ref{maintheo} are needed}
\begin{equation}\label{a2}
\left|\frac{d}{d\epsilon}b_{\theta,j}(\bdbeta)\right|\leq
C_{\mathcal{K}}(1+\|\bdbeta\|)\|s'-s\|\,.
\end{equation}
Now, concerning the derivative of
$\tilde{r}_{j}(\bdbeta^j,b;\bdbetamj,\theta)$,  since
$$\log(\tilde{r}_{j}(\bdbeta^j,b;\bdbetamj,\theta))=\frac{1}{2}
\sum_{i=1}^n\left(\|\dy_i-K^{\tilde{\bdbeta}_i}_p\alpha\|^2
-\|\dy_i-K^{{\bdbeta}_i}_p\alpha\|^2\right)\,,$$
 with $\tilde{\bdbeta}_i = \bdbeta_{i,b\to j} $, $i$ corresponding to the
 $i^{th}$ image,
 only one term of the previous sum is nonzero. We
 deduce from the fact that
$K_p$ is bounded and from \eqref{PhotoUpdateClust} that
$|\frac{d}{d\epsilon}\log(\tilde{r}_{j}(\bdbeta^j,b;\bdbetamj,\theta))|
\leq C_\Kapa|\frac{d}{d\epsilon}\alpha|\leq C_{\mathcal{K}}\|s-s'\|$, so
 that using the fact that
 $\tilde{r}_{j}(\bdbeta^j,b;\bdbetamj,\theta)$  is uniformly
 bounded for $\theta\in\hat{\theta}(\mathcal{K})$, $\bdbeta\in\mathbb{R}^N$ and
 $b\in\mathbb{R}$, there exists a new constant $C_\mathcal{K}$ such that
\begin{equation*}\label{a3}
|\frac{d}{d\epsilon}\tilde{r}_{j}(\bdbeta^j,b;\bdbetamj,\theta))|
\leq C_{\mathcal{K}}\|s-s'\|\,.
\end{equation*}
Thus, using \eqref{dexpdeps}, \eqref{a0} and \eqref{a2}, %and \eqref{a3}
 we get for a new
constant $C_{\mathcal{K}}$ that
\begin{multline*}%\label{a4}
\left|\frac{d}{d\epsilon}r_j(\bdbeta^j, b; \bdbetamj, \theta) \exp\left(-\frac{(b-b_{\theta,j}(\bdbeta))^2\|e_j\|^2_\theta}{2}
\right)\frac{\|e_j\|_\theta}{\sqrt{2\pi}}\right|\\
\leq C_\mathcal{K} (1+\|\bdbeta\|)\|s'-s\|(1+|b-b_{\theta,j}(\bdbeta)|)^2
\exp\left(-\frac{(b-b_{\theta,j}(\bdbeta))^2\|e_j\|^2_\theta}{2}\right)
\frac{\|e_j\|_\theta}{\sqrt{2\pi}}\,.
\end{multline*}
Since $\|f(\bdbeta)\|\leq \| f\|_{V^p}V^p(\bdbeta)$ and $V(a+b)=1+\|a+b\|^2\leq
2(V(a)+V(b))$, we have $\|f(\bdbeta_{0\to j}+be_j)\|\leq  C \|
f\|_{V^p}(V^p(\bdbeta_{0\to j})+V^p(be_j))$ with $C=2^{2p-1}$. Hence, there
exists an $C_\mathcal{K}$ such that  $\forall \, (s,s')\in\mathcal{K}^2$,
$\forall \ 1 \leq j \leq N$, $\forall \, \bdbeta \in \R^N$ and $\forall \, \epsilon\in[0,1]$:
\begin{multline*}
  \int_\mathbb{R}\|f(\bdbeta_{0\to j }+be_j)\|\left|\frac{d}{d\epsilon}\left(r_j(\bdbeta^j,
      b; \bdbetamj, \theta)
      \exp\left(-\frac{(b-b_{\theta,j}(\bdbeta))^2\|e_j\|^2_\theta}{2}\right)\frac{\|e_j\|_\theta}
{\sqrt{2\pi}}\right)\right|db\\
\leq  C_{\mathcal{K}}\|
f\|_{V^p}V^p(\bdbeta_{0\to j})(1+\|\bdbeta\|)\|s'-s\|\leq  C_{\mathcal{K}}\|
f\|_{V^p}V^p(\bdbeta)(1+\|\bdbeta\|)\|s'-s\|\,, %\label{eq:7}
\end{multline*}
where we have used the fact that a Gaussian variable has finite moments
of all order. Since $(1+\|\bdbeta\|)\leq (2V(\bdbeta))^{1/2}$, we get (updating
$C_{\mathcal{K}}$) that
\begin{equation}
  \label{eq:8}
 \|(\ntrans_{\theta(1),j}^2-\ntrans_{\theta(0),j}^2)f(\bdbeta)\|\leq  C_{\mathcal{K}}\|
f\|_{V^p}V^{p+1/2}(\bdbeta)\|s'-s\|\,.
\end{equation}
Now, looking at the first term
 in \eqref{eq:9}, we
deduce easily from the previous study for $f\equiv f(\bdbeta)$ that
\begin{equation}
\|(\ntrans_{\theta(1),j}^1-\ntrans_{\theta(0),j}^1)f(\bdbeta)\|\leq C_\mathcal{K}
V(\bdbeta)^{1/2}\|s'-s\|\|f(\bdbeta)\|\leq  C_{\mathcal{K}}\|
f\|_{V^p}V^{p+1/2}(\bdbeta)\|s'-s\| \,;\label{eq:10}
\end{equation}
so that adding \eqref{eq:8} and \eqref{eq:10}, we get (again updating
$C_{\mathcal{K}}$) that
\begin{equation}
\|(\ntrans_{\theta(1),j}-\ntrans_{\theta(0),j})f\|_{V^{p+1/2}}\leq  C_{\mathcal{K}}\|
f\|_{V^p}\|s'-s\|\,.\label{eq:11b}
\end{equation}
We end the proof, saying that
$\ntrans_{\theta(1)}-\ntrans_{\theta(0)}=\sum_{j=1}^N
  \ntrans_{\theta(1),j+1,N}\circ
  (\ntrans_{\theta(1),j}-\ntrans_{\theta(0),j})\circ  \ntrans_{\theta(0),1,j-1}$
  where
  $\ntrans_{\theta,q,r}=\ntrans_{\theta,r}\circ\ntrans_{\theta,r-1} \circ  \cdots\circ
  \ntrans_{\theta,q}$ for any integer $q\leq r$ and any $\theta\in\Theta$
  so that using \eqref{eq:11} and \eqref{eq:11b}, the result is
  straightforward.
\end{proof}

\hspace{0.5cm}

\begin{lemma}\label{lem67}
Let $\mathcal{K}$ be a compact subset of $\mathcal{S}$. There exists a
constant $C_\mathcal{K}$ such that
for  all $p\geq 1$ and any function $f \in \mathcal{L}_{V^p}$, $\forall (s,s')
\in \mathcal{K}^2$, $\forall k\geq 0$, we have for
$\theta=\hat{\theta}(s)$ and $\theta'=\hat{\theta}(s')$ that:
  \begin{eqnarray*}
\| \ntrans_\te ^k f -
  \ntrans_{\te'} ^k f \|_{V^{p+1/2}}  \leq
C_\mathcal{K}  %\\
\|f\|_{V^{p}} \|s-s'\| \ .
\end{eqnarray*}
\end{lemma}
\begin{proof}
We use the same decomposition of the difference as previously:
  \begin{eqnarray*}
    \ntrans_\te ^k f -
  \ntrans_{\te'} ^k f = \sum\limits_{i=1}^{k-1}
\ntrans_{\te}^i (\ntrans_\te -\ntrans_{\te'}) (\ntrans_{\te'} ^{k-i-1}
f - \lstat_{\te'}(f) ) \ .
\end{eqnarray*}
Using Lemma \ref{lem:Holder},  the fact that $\|\ntrans^k_\theta
(f-\pi_\theta(f))\|_{V^p}\leq \gamma_\Kapa^k \|f\|_{V^p}$ with $\gamma_\Kapa<1$
(geometric ergodicity) and $\sup\limits_{j\geq 0}\sup\limits_{\theta \in K} \|
\ntrans^j_\theta V^q\| _{V^q} < \infty $ we get:

\begin{eqnarray*}
  \|  \ntrans_\te ^k f -
  \ntrans_{\te'} ^k f \|_{V^{p+1/2}} &\leq &C_\Kapa \sum\limits_{i=1}^{k-1}
\| (\ntrans_\te -\ntrans_{\te'}) (\ntrans_{\te'} ^{k-i-1}
f - \lstat_{\te'}(f) )\|_{V^{p+1/2}} \\
&\leq& C_\Kapa \|f\|_{V^{p}} |s-s'|  \sum\limits_{i=1}^{k-1}  \gamma_\Kapa^{k-i+1}
\end{eqnarray*}
and the lemma is proved.
\end{proof}\\

We now prove that $h$ is a Hölder function, adapting linearly
Appendix B of \cite{andrieumoulinespriouret}.

Let $\bdbeta\in\mathbb{R}^N$ and denote by $\theta=\hat{\theta}(s)$ and $\theta'=\hat{\theta}(s')$.
Write $h(s)-h(s')=A(s,s') +B(s,s') + C(s,s')$, where
\begin{eqnarray*}
  A(s,s') &=& (h(s) -\ntrans_{\te}^k H_s(\bdbeta)) +
  (\ntrans_{\theta'}^k H_{s'}(\bdbeta) - h(s'))\,,\\
B(s,s')&=& \ntrans_{\theta}^k H_s(\bdbeta) -
\ntrans_{\theta'}^k H_s(\bdbeta) \,,\\
C(s,s')& =&\ntrans_{\theta'}^k H_s(\bdbeta)
-\ntrans_{\theta'}^k H_{s'}(\bdbeta) \, .
\end{eqnarray*}
Using the geometric ergodicity, Lemma \ref{lem:Holder} and Lemma
\ref{lem67}, we get that there exists an $C>0$, independent of $k$
such that:
\begin{eqnarray*}
  \| A(s,s')\| &\leq& C \gamma^k \sup\limits_{\s \in \Kapa} \| H_s\|_{V} V(\bdbeta),\\
 \|B(s,s') \|&\leq& C  \sup\limits_{\s \in \Kapa} \| H_s\|_{V} \|s-s'\| V^{3/2}(\bdbeta),\\
\| C(s,s')\|&\leq& C  \sup\limits_{\s \in \Kapa} \| H_s\|_{V} \|s-s'\|  V(\bdbeta) \, .
\end{eqnarray*}
This yields
\begin{eqnarray*}
  \|h(s)-h(s')\| \leq C V^{3/2}(\bdbeta) (\gamma ^k + \|s-s'\| ) \, .
\end{eqnarray*}
Hence, setting $k=[\log \|s-s'\| / \log (\gamma)]$ if $\|s-s'\|<1$ and $1$
otherwise, we get the result.\\

We can now end the proof of (\textbf{A3'(ii)}):
On one hand we have:
\begin{multline*}
\|(\ntrans_{\theta} ^k H_s(\bdbeta) - h(s)) - (\ntrans_{\theta'}
    ^k H_{s'}(\bdbeta) - h(s')) \| \leq
\|\ntrans_{\theta} ^k  H_s(\bdbeta)- \ntrans_{\theta} ^k
H_{s'}(\bdbeta) \| \\  + \|\ntrans_{\theta} ^k  H_{s'}(\bdbeta)- \ntrans_{\theta'} ^k
H_{s'}(\bdbeta) \|+ \|  h(s) -h(s')\|\leq
C \|s-s'\| V^{3/2}(\bdbeta) \, .
\end{multline*}
On the other hand, we have thanks to the geometric ergodicity,
\begin{eqnarray*}
  \|(\ntrans_{\theta} ^k H_s(\bdbeta) - h(s)) - (\ntrans_{\theta'}
    ^k H_{s'}(\bdbeta) - h(s')) \| \leq C \gamma^k V^{3/2}(\bdbeta) \, .
\end{eqnarray*}
Hence for any $t\geq 0$ and $T\geq t$, we have
\begin{multline*}
  \|\ntrans_{\theta}^t g_{\hat\theta(s)} (\bdbeta) - \ntrans_{\theta'}^t
  g_{{\hat\theta(s')}} (\bdbeta) \| \leq \sum\limits_{k=t}^\infty \|(\ntrans_{\theta} ^k H_s(\bdbeta) - h(s)) - (\ntrans_{\theta'}
    ^k H_{s'}(\bdbeta) - h(s')) \| \leq \\
 C  V^{3/2}(\bdbeta) \left[
T\|s-s'\| + \frac{\gamma^{T+t}}{1-\gamma}
\right] \ .
\end{multline*}
Setting $T=[\log\|s-s'\| / \log(\gamma)] $ for $\|s-s'\| \leq \delta <1 $
and  $T=t$
otherwise,  using also the fact that for any $0<a<1$ we have $\|s-s'\|\log\|s-s'\| =
o(\|s-s'\|^a) $, we get the result.\\

This proves condition (\textbf{A3'(ii)}) for any $a<1$.

We finally focus on the proof of (\textbf{A3'(iii)}). Once again we first prove a specific result for each function
$V_\te$ and obtain after a result for the function $V$.
\begin{lemma}
  \label{lem:2b}Let $\mathcal{K}$ be a compact subset of
  $\mathcal{S}$ and $p\geq 1$. There exists $C_{\Kapa,p}>0$ such that for any
  $s,s'\in\mathcal{K}$, for any $\bdbeta\in \R^N$,
$$|V^p_{\hat\theta(s)}(\bdbeta)-V^p_{\hat\theta(s')}(\bdbeta)|\leq
C_{\Kapa,p}\|s-s'\|V^p_{\hat\theta(s)}(\bdbeta)\,.$$
\end{lemma}
\begin{proof}
  Indeed, there exists  $C>0$ such that for any $\hat\theta(s)=(\alpha,\Gamma_g)$ and
  $\hat\theta(s')=(\alpha',\Gamma'_g)$, $|\Gamma_g-\Gamma'_g|\leq
  C\|s-s'\|$. Therefore, there exists an $C$ such that
  $\forall \, (s,s')\in\mathcal{K}^2$, $|\Gamma_g^{-1}-(\Gamma'_{g})^{-1}|\leq
  C\|s-s'\|$ and
  $$|V_{\hat\theta(s)}(\bdbeta)-V_{\hat\theta(s')}(\bdbeta)|\leq \sum_{i=1}^n
  \beta_i^t (\Gamma_g^{-1}-(\Gamma'_g )^{-1}) \beta_i\leq C\|s-s'\|V(\bdbeta)\,.$$ The result
  follows from the existence of a constant $C$ such that
  $\frac{1}{c}V(\bdbeta) \leq
  V_{\hat\theta(s)}(\bdbeta)\leq C V(\bdbeta)$ for any
  $(\bdbeta,s)\in\mathbb{R}^N × \mathcal{K}$.
\end{proof}
\begin{lemma}
\label{lem:3}
  Let $\mathcal{K}$ be a compact subset of $\mathcal{S}$ and $p\geq 1$.
There exist  $\bar\varepsilon>0$
  and  $C>0$ such that for any sequence $\boldsymbol\varepsilon=(\varepsilon_k)_{k\geq 0}$
  such that $\varepsilon_k\leq \bar\varepsilon$ for $k$ large enough, any
  sequence $\boldsymbol\Delta=(\Delta_k)_{k\geq 0}$ and any $\bdbeta\in\mathbb{R}^N$,
$$\sup_{s\in \mathcal{K}}\sup_{k\geq
  0}\mathbb{E}_{\bdbeta,s}^{\boldsymbol\Delta}
[V^p(\bdbeta_k)\mathds{1}_{\sigma(\mathcal{K})\land \nu(\boldsymbol\varepsilon)\geq
  k}]\leq CV^p(\bdbeta)\,.$$
\end{lemma}
\begin{proof}
  Let $K$ be a compact subset of $\Theta$ such that
  $\hat\theta(\mathcal{K})\subset K$. We note in the sequel,
  $\theta_k=\hat\theta(s_k)$. We have for $k\geq 2$, using the Markov property and
  Lemmas \ref{lem:a1} and \ref{lem:2b},
\begin{multline*}
  \mathbb{E}_{\bdbeta,s}^{\boldsymbol\Delta} [V^p_{\theta_{k-1}}(\bdbeta_k)\mathds{1}_{\sigma(\mathcal{K})\land
    \nu(\boldsymbol\varepsilon)\geq k}]\leq \mathbb{E}_{\bdbeta,s}^{\boldsymbol\Delta}
  [\ntrans_{\theta_{k-1}}V^p_{\theta_{k-1}}(\bdbeta_{k-1})\mathds{1}_{\sigma(\mathcal{K})\land
    \nu(\boldsymbol\varepsilon)\geq k}]\\
\leq \rho \left(  \mathbb{E}_{\bdbeta,s}^{\boldsymbol\Delta}
  [V^p_{\theta_{k-2}}(\bdbeta_{k-1})\mathds{1}_{\sigma(\mathcal{K})\land
    \nu(\boldsymbol\varepsilon)\geq k}]+ \mathbb{E}_{\bdbeta,s}^{\boldsymbol\Delta}
[(V^p_{\theta_{k-1}}(\bdbeta_{k-1})-V^p_{\theta_{k-2}}(\bdbeta_{k-1}))
\mathds{1}_{\sigma(\mathcal{K})\land
    \nu(\boldsymbol\varepsilon)\geq k}]\right)+ C\\
\leq \rho \left(  \mathbb{E}_{\bdbeta,s}^{\boldsymbol\Delta}
  [V^p_{\theta_{k-2}}(\bdbeta_{k-1})\mathds{1}_{\sigma(\mathcal{K})\land
    \nu(\boldsymbol\varepsilon)\geq k-1}]+ C'\epsilon_{k-1} \mathbb{E}_{\bdbeta,s}^{\boldsymbol\Delta}
[V^{p}_{\theta_{k-2}}(\bdbeta_{k-1})\mathds{1}_{\sigma(\mathcal{K})\land
    \nu(\boldsymbol\varepsilon)\geq k-1}]\right)+ C \,.
\end{multline*}
By induction, we show that
$$  \mathbb{E}_{\bdbeta,s}^{\boldsymbol\Delta} [V^p_{\theta_{k-1}}(\bdbeta_k)
\mathds{1}_{\sigma(\mathcal{K})\land
    \nu(\boldsymbol\varepsilon)\geq k}]\leq  \prod_{l=1}^{k-1}(\rho (1+C'\varepsilon_l)
    )V^p_{\hat\theta(s)}(\bdbeta)+\frac{C}{(1-\rho (1+C'\bar\varepsilon))} \ .$$
Choosing $\bar\varepsilon$ such that $\rho (1+C'\bar\varepsilon)<1$ and introducing
    again  $0\leq c_1\leq c_2$ such that $c_1V(\bdbeta)\leq
  V_\theta(\bdbeta)\leq c_2 V(\bdbeta)$ for any $(\bdbeta,\theta)\in\mathbb{R}^N×
  K$ end the proof.
\end{proof}

This yields (\textbf{A3'(iii)}).\\

This concludes the demonstration of Theorem \ref{th:condition}.

\section{Conclusion and discussion}
We have proposed a stochastic algorithm for constructing Bayesian non-rigid
deformable models in the same context as \cite{AAT} together with
a proof of convergence toward a critical point of the observed
likelihood. To the best of our best knowledge, this is the first theoretical
result on convergence
%a well defined statistical point of view
in the context of deformable template.
The algorithm is based on a stochastic approximation  of
the EM algorithm using  an MCMC approximation of the posterior distribution and
truncation on random boundaries. Although
our main contribution is theoretical, the
preliminary experiments presented here  on the US-postal database
show that the stochastic approach can be easily implemented and is
robust to noisy situations, yielding better results than the previous
deterministic schemes.

Many interesting questions remain open.
One may ask what is the convergence rate of such stochastic
algorithms.
A first result has been proved in \cite{DLM} for the standard SAEM
algorithm. Under mild conditions, the authors state a central limit
theorem for an average
sequence of the estimated parameters $(\theta_k)_k$.
Concerning the generalization when introducing MCMC, a first step has been
tackled in \cite{andrieumoulines}. Under some restrictive
assumptions the authors can prove a central limit theorem for an ergodic
adaptive Monte Carlo Markov chain. We truly think that it is possible
to obtain this kind of convergence rates for the SAEM-MCMC algorithm
proposed in this paper.

Another question refers to the extension
of the stochastic scheme to mixture of deformable models (defined as the
multicomponent model in \cite{AAT}) where the
parameters are the weights of the individual components and for each
component, the associated template and deformation law. This is of
particular importance for real data analysis where the restriction to a
unique deformable model could be too limiting. The design
of such mixtures corresponds to some kind of deformation invariant
clustering approach of the data which is a basic issue in any
unsupervised data analysis scheme.
This extension is, however, not as straightforward as it would appear
at first glance: due to
the high dimensional hidden deformation variables, a naive
extension of the Markovian dynamics to the component variables
will have extremely poor mixing properties leading
to an impractical algorithm. A less straightforward extension involving
multiple MCMC chains is under study.

Another interesting extension is to consider diffeomorphic
mappings and not only displacement fields for the hidden
deformation. This appears to be particularly interesting in the context
of Computational Anatomy where a one to one correspondence between
the template and the observation is usually needed and cannot be
guaranteed with linear spline interpolation schemes. This extension
could be done in principle using tangent models based on geodesic
shooting in the spirit of \cite{vmty04}.
%Many numerical as well as
%theoretical work need to be done.% on this side.
% This provides a range of learning algorithms which can be chosen according to
% the goal pursued: accuracy (depending on the signal to noise ratio of
% the training images) or speed of convergence, in particular for
% the multicomponent model whose way to carry it out is slower using the
% stochastic procedure.
% The same choice is given for a classification purpose.
%\Alain{A faire...}

% This provides a range of learning algorithms which can be chosen according to
% the goal pursued: accuracy (depending on the signal to noise ratio of
% the training images) or speed of convergence, in particular for
% the multicomponent model whose way to carry it out is slower using the
% stochastic procedure.
% The same choice is given for a classification purpose.
\bibliographystyle{abbrv}
\bibliography{bibcomp2}

\end{document}